\documentclass[11pt]{article}
\usepackage{fullpage}
\usepackage{parskip}

\usepackage{amsthm,amsmath,amssymb,booktabs,xcolor,graphicx}
\usepackage{thm-restate}
\usepackage{bbm}
\usepackage{listings}
\usepackage{xcolor}
\usepackage{appendix}
\usepackage{enumerate}
\usepackage{mathbbol}
\usepackage[numbers]{natbib}
\usepackage{comment}
\usepackage[shortlabels]{enumitem}

\usepackage[margin=1in]{geometry}
\usepackage[utf8]{inputenc}
\usepackage[T1]{fontenc}

\usepackage[disable]{todonotes}

\usepackage[hidelinks]{hyperref}
\usepackage{url}
\usepackage{etoolbox}
\usepackage{appendix}
\usepackage{xspace}
 \usepackage{algorithm}
\usepackage[noend]{algpseudocode}
\usepackage[labelfont=bf]{caption}
\usepackage[noabbrev,capitalize]{cleveref}
\crefname{equation}{}{} 
\AtBeginEnvironment{appendices}{\crefalias{section}{appendix}} 

\usepackage[color,final]{showkeys} 

\colorlet{refkey}{orange!20}
\colorlet{labelkey}{blue!30}

\numberwithin{equation}{section}
\newtheorem{theorem}{Theorem}[section]

\newtheorem{lemma}[theorem]{Lemma}
\newtheorem{claim}[theorem]{Claim}

\crefname{claim}{Claim}{Claims}

\newtheorem{corollary}[theorem]{Corollary}

\newtheorem*{question*}{Question}

\theoremstyle{definition}
\newtheorem{definition}[theorem]{Definition}

\newtheorem*{definition*}{Definition}

\theoremstyle{remark}

\newcommand{\poly}{\mathrm{poly}}

\newcommand{\fC}{\mathcal{C}}

\newcommand{\rad}{\mathrm{rad}}

\newcommand{\eps}{\varepsilon}

\newcommand{\E}{\mathbb{E}}

\newcommand{\N}{\mathbb{N}}

\newcommand{\diam}{\textrm{diam}}
\newcommand{\diff}{\textrm{diff}}

\renewcommand{\d}{d}
\newcommand{\Dpad}{D_\textrm{pad}}
\newcommand{\fO}{\mathcal{O}}
\newcommand{\tO}{\widetilde{O}}

\newcommand{\Vrem}{S}

\newcommand{\fR}{\mathcal{R}}
\newcommand{\fB}{\mathcal{B}}

\newcommand{\fA}{\mathcal{A}}

\newcommand{\fZ}{\mathcal{Z}}

\newcommand{\polylog}{\poly\log}

\newcommand{\congest}{$\mathsf{CONGEST}\,$}

\newcommand{\pram}{$\mathsf{PRAM}\,$}

\newcommand{\Lref}[1]{{\text{L\ref{#1}}}}
\newcommand{\Tref}[1]{{\text{T\ref{#1}}}}

\newcommand{\dist}{\operatorname{dist}}

\newcommand{\oDist}{\fO^{Dist}}
\newcommand{\oWeak}{\fO^{Dist\textrm{-}Weak}}
\newcommand{\oPot}{\fO^{Pot}}
\newcommand{\oForestAgg}{\fO^{Forest\textrm{-}Agg}}
\newcommand{\oGlobalAgg}{\fO^{Global\textrm{-}Agg}}

\newcommand{\Valive}{V^\textrm{alive}}
\newcommand{\Vgood}{V^\textrm{good}}
\newcommand{\Vbad}{V^\textrm{bad}}
\newcommand{\Vactive}{V^\textrm{active}}

\newcommand{\parent}{\textrm{par}}
\renewcommand{\P}{\textrm{P}}

\renewcommand{\root}{\textrm{root}}
\newcommand{\del}{\textrm{del}}

\newcommand{\Vmiddle}{V^{\textrm{middle}}}
\newcommand{\Vunsafe}{V^{\text{unsafe}}}
\newcommand{\Vpmiddle}{V'^{\text{middle}}}
\newcommand{\Sbig}{S^{\text{big}}}
\newcommand{\Blur}{\text{Blur}}
\newcommand{\shortcutQuality}[1]{\mathrm{Shortcut}\-\mathrm{Quality}(#1)}
\newcommand{\hopDiameter}[1]{\mathrm{Hop}\-\mathrm{Diam}(#1)}

\title{Deterministic Low-Diameter Decompositions for Weighted Graphs and Distributed and Parallel Applications}

\author{
Václav Rozhoň \thanks{\scriptsize{Supported by the European Research Council (ERC) under the European Unions Horizon 2020 research and innovation programme (grant agreement No.~853109).}}\\
\small ETH Zurich \\
\small rozhonv@inf.ethz.ch\\
\and
\textcircled{r}
\and
Michael Elkin\thanks{\scriptsize{This research was supported by the ISF grant No. (2344/19) }}\\
\small Ben-Gurion University of the Negev\\
\small elkinm@cs.bgu.ac.il\\
\and
\textcircled{r}
\and
Christoph Grunau \footnotemark[4]   \\
\small ETH Zurich \\
\small cgrunau@inf.ethz.ch\\
\and
\textcircled{r}\footnote{\scriptsize{The author ordering was randomized using \url{https://www.aeaweb.org/journals/policies/random-author-order/generator}. 
It is requested that citations of this work list the authors separated by \texttt{\textbackslash textcircled\{r\}} instead of commas: Elkin \textcircled{r} Haeupler \textcircled{r} Rozhoň \textcircled{r} Grunau.}}
\and
Bernhard Haeupler\thanks{\scriptsize{Supported in part by NSF grants CCF-1814603, CCF-1910588, NSF CAREER award CCF-1750808, a Sloan Research Fellowship, funding from the European Research Council (ERC) under the European Union's Horizon 2020 research and innovation program (ERC grant agreement 949272), and the Swiss National Foundation (project grant 200021-184735).}} \\
  \small ETH Zurich \& Carnegie Mellon University\\
  \small bernhard.haeupler@inf.ethz.ch
}

\date{}

\begin{document}

\maketitle
\thispagestyle{empty}

\begin{abstract}
This paper presents new deterministic and distributed low-diameter decomposition algorithms for weighted graphs. 
In particular, we show that if one can efficiently compute approximate distances in a parallel or a distributed setting, one can also efficiently compute low-diameter decompositions. 
This consequently implies solutions to many fundamental distance based problems using a polylogarithmic number of approximate distance computations. 

Our low-diameter decomposition generalizes and extends the line of work starting from \cite{rozhon_ghaffari2019decomposition} to weighted graphs in a very model-independent manner. 
Moreover, our clustering results have additional useful properties, including strong-diameter guarantees, separation properties, restricting cluster centers to specified terminals, and more. Applications include:

-- The first near-linear work and polylogarithmic depth randomized and deterministic parallel algorithm for low-stretch spanning trees (LSST) with polylogarithmic stretch. Previously, the best parallel LSST algorithm required $m \cdot n^{o(1)}$ work and $n^{o(1)}$ depth and was inherently randomized. No deterministic LSST algorithm with truly sub-quadratic work and sub-linear depth was known. 
        
-- The first near-linear work and polylogarithmic depth deterministic algorithm for computing an  $\ell_1$-embedding into polylogarithmic dimensional space with polylogarithmic distortion. The best prior deterministic algorithms for $\ell_1$-embeddings either require large polynomial work or are inherently sequential. 
        
Even when we apply our techniques to the classical problem of computing a ball-carving with strong-diameter $O(\log^2 n)$ in an unweighted graph, our new clustering algorithm still leads to an improvement in round complexity from $O(\log^{10} n)$ rounds \cite{chang_ghaffari2021strong_diameter} to $O(\log^{4} n)$.
 \end{abstract}

\tableofcontents
\newpage

\section{Introduction}
\label{sec:intro}

This paper gives deterministic parallel \& distributed algorithms for low-diameter clusterings in weighted graphs. 
The main message of this paper is that once you can deterministically and efficiently compute $(1+1/\poly(\log n))$-approximate distances in undirected graphs in your favorite parallel/distributed model, you can also deterministically and efficiently solve various clustering problems with $\poly\log(n)$ approximate distance computations. 
Since low-diameter clusterings are very basic objects and approximate distances can efficiently and deterministically be computed in various parallel and distributed models, our clustering results directly imply efficient deterministic algorithms for various problems.

In the literature, a multitude of different clustering problems are defined -- you may have encountered buzzwords like low-diameter clusterings, sparse covers, network decompositions, etc. -- most of which are tightly related in one way or another. 
To give an example of a problem that we consider in this paper, suppose you are given a parameter $D$ and you want to partition the vertex set of an input graph $G$ into clusters of diameter $\tO(D)$\footnote{The $\tO$-notation hides polylogarithmic factors in the number of vertices.} such that every edge $e$ is cut, that is, connecting different clusters, with probability at most $\ell(e)/D$. 
In a deterministic variant of the problem, we instead want the number of edges cut to be at most $\sum_{e \in E(G)} \ell(e) / D$. 
This clustering problem is usually known as a low-diameter clustering problem. 
Another problem we consider is that of computing a $D$-separated clustering: there, we are supposed to cluster each node with probability at least $1/2$ (or at least half of the nodes if the algorithm is deterministic) in clusters such that each cluster has diameter $\tO(D)$ and any pair of clusters has distance at least $D$.


Our main clustering result solves a very general clustering problem that essentially generalizes both examples above. The algorithm deterministically reduces the clustering problem to $\poly(\log n)$ approximate distance computations in a parallel/distributed manner. 
The clustering comes with several additional useful properties. 
We produce strong-diameter clusters; on the other hand, some results in the literature only give a so-called weak-diameter guarantee where every two nodes of the cluster are close in the original graph but the cluster itself may be even disconnected.  
Moreover, it can handle several generalizations which are crucial for some applications such as the low-stretch spanning tree problem. Most notably, our clustering result generalizes to the case when a set of terminals is given as part of the input and each final cluster should contain at least one terminal.

\subsection{Main Results}
\label{sec:intro_our_results}
While we think of our general clustering result as the main result of this paper, it is not necessary to state it in this introductory section in full generality. 
Instead, we start by discussing its following corollary (see \cref{fig:terminals}). The following type of a clustering result is needed in known approaches to compute low-stretch spanning trees. 

\begin{figure}
    \centering
    \includegraphics[width = \textwidth]{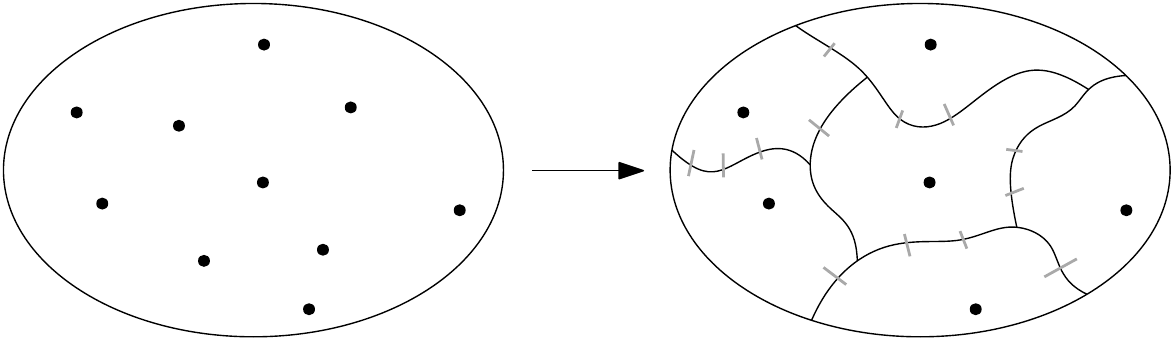}
    \caption{Clustering from \cref{thm:steroids_simple}: we are given a set of terminals $Q$. We should construct a partition of the input graph into small-diameter clusters such that each cluster contains at least one terminal. Moreover, only a small number of edges should be cut by the clustering (grey edges on the right). 
    }
    \label{fig:terminals}
\end{figure}

\begin{restatable}{theorem}{steroidsSimple}[A corollary of \cref{thm:steroids}]
\label{thm:steroids_simple}
Let $G$ be a weighted graph. We are given a set of \emph{terminals} $Q \subseteq V(G)$ and a parameter $R > 0$ such that for every $v \in V(G)$ we have $d(Q, v) \le R$. Also, a precision parameter $0  < \eps < 1$ is given.  There is a deterministic distributed and parallel algorithm outputting a partition $\fC$ of the vertices into clusters and a subset of terminals $Q' \subseteq Q$ with the following properties: 

\begin{enumerate}
    \item Each cluster $C \in \fC$ contains exactly one terminal $q \in Q'$. Moreover, for any $v \in C$ we have $d_{G[C]}(q, v) \le (1+\eps)R$. 
    \item For the set $E^{bad}$ of edges connecting different clusters of $\fC$ we have 
    \[
    |E^{bad}| = \tO\left( \frac{1}{\eps R}   \right)\cdot \sum_{e \in E(G)} \ell(e). 
    \]
\end{enumerate}

The \pram variant of the algorithm has work $\tO(m)$ and depth $\tO(1)$.
The \congest variant of the algorithm runs in $\tO(\sqrt{n} + \hopDiameter{G})$ rounds.
\end{restatable}

Our result above is in fact quite model-independent as we essentially reduce the problem to $\poly(\log n)$ $(1+1/\poly(\log n))$-approximate distance computations.
The final complexities then follow from the recent work of \cite{RGHZL2022sssp}: the authors give efficient deterministic parallel and distributed approximate shortest path algorithms in \pram and \congest.

\paragraph{Low-Stretch Spanning Trees}
As a straightforward corollary of the clustering result in \cref{thm:steroids_simple}, we obtain an efficient deterministic parallel and distributed algorithm for computing low-stretch spanning trees. 
Low-stretch spanning trees were introduced in a seminal paper by Alon et al. \cite{alon_karp_peleg_west1995low_stretch_spanning_tree}, where they were shown useful for the online $k$-server problem. 
The algorithm of \cite{alon_karp_peleg_west1995low_stretch_spanning_tree} constructed spanning trees with average stretch $\textrm{exp}(\sqrt{\log n \log\log n)}$. 
In a subsequent work Bartal \cite{bartal1996probabilistic,Bartal98} and Fakchraenphol et al. \cite{FRT04} showed that one can get logarithmic stretch if one allows the trees to use edges that are not present in the original graph. 
In \cite{elkin2008lower} it was shown that the original problem of low-stretch spanning trees admits a solution with polylogarithmic stretch. That bound was later improved to a nearly-logarithmic bound in \cite{ABN07}. These constructions have important applications to the framework of spectral sparsification \cite{ST04}. 

In the distributed setting the problem was studied in \cite{becker_emek_ghaffari_lenzen2019low_stretch_spanning_trees}.
However, the latter algorithm relies on the computation of exact distances. 
Our approach, on the other hand, only relies on approximate distance computations that, unlike exact distances, can be computed with near-optimal parallel and distributed complexity \cite{RGHZL2022sssp}. 
Hence, we are able to present the first distributed and parallel algorithm for this problem that provides polylogarithmic stretch, polylogarithmic depth and near-linear work.

\begin{theorem}[Deterministic Low-Stretch Spanning Tree]
Let $G$ be a weighted graph. Each edge $e$ has moreover a nonnegative importance $\mu(e)$. 
There exists a deterministic parallel and distributed algorithm which outputs a spanning tree $T$ of $G$ such that 

\begin{align}
\label{eq:deadline_is_soon}
\sum_{e = \{u,v\} \in E(G)} \mu(e) d_T(u,v) = \tO\left( \sum_{e = \{u,v\} \in E(G)} \mu(e) d_G(u,v) \right). 
\end{align}

The \pram variant of the algorithm has work $\tO(m)$ and depth $\tO(1)$. 
The \congest variant of the algorithm runs in $\tO(\sqrt{n} + \hopDiameter{G})$ rounds.
\end{theorem}

Note that plugging in $\mu(e) := \mu'(e) / \ell(e)$ into \cref{eq:deadline_is_soon} and using $d_G(u,v) \le \ell(e)$, we also get the following similar guarantee of
\begin{align*}
    \sum_{e = \{u,v\} \in E(G)} \mu'(e) \cdot  \frac{d_T(u,v)}{\ell(e)} 
    = \tO\left( 
    \sum_{e = \{u,v\} \in E(G)} \mu'(e)
    \right)
\end{align*}

The stretch is optimal up to polylogarithmic factors.

\paragraph{$\ell_1$ Embedding}
Embeddings of networks in low dimensional spaces like $\ell_1$-space are a basic tool with a number of applications. For example, the parallel randomized approximate shortest path algorithm of  \cite{li19parallel_shortest_path} uses $\ell_1$-embeddings as a crucial subroutine. 
By using our clustering results, we can use an approach similar to the one from \cite{bartal2021advances} to obtain an efficient deterministic parallel and distributed algorithm for $\ell_1$-embedding.

\begin{theorem}[$\ell_1$-Embedding]
Let $G$ be a weighted graph.
There exists a deterministic parallel and distributed algorithm which computes an embedding in $\tO(1)$-dimensional $\ell_1$-space with distortion $\tO(1)$.
The \pram variant of the algorithm has work $\tO(m)$ and depth $\tO(1)$. 
The \congest variant of the algorithm runs in $\tO(\sqrt{n} + \hopDiameter{G})$ rounds.
\end{theorem}

\paragraph{Other Applications}
Since low-diameter clusterings are an important subroutine for numerous problems, there are many other more standard applications for problems like ($h$-hop) Steiner trees or Steiner forests, deterministic variants of tree embeddings, problems in network design, etc. \cite{becker_emek_lenzen2020blurry_ball_growing,haeupler2021tree}
We do not discuss these applications here due to space constraints. 
We also note that the distributed round complexities of our algorithms are almost-universally-optimal. We refer the interested reader to \cite{ghaffar2020hop_constrained_oblivious_routing,haeupler2021universally,cool_stuff} for more details regarding the notion of universal optimality.

\subsection{Previous Work and Barriers}

We will now discuss two different lines of research that study low-diameter clusterings and mention some limits of known techniques that we need to overcome. 

\paragraph{Building Network Decompositions}

One line of research 
\cite{awerbuch89,panconesi-srinivasan,linial_saks93,ghaffari_harris_kuhn2018derandomizing,ghaffari_kuhn_maus2017slocal,rozhon_ghaffari2019decomposition,ghaffari_grunau_rozhon2020improved_network_decomposition,chang_ghaffari2021strong_diameter} is motivated by the desire to understand the deterministic distributed complexity of various fundamental symmetry breaking problems such as maximal independent set and $(\Delta+1)$-coloring. In the randomized world, there are classical and efficient distributed algorithms solving these problems, the first and most prominent one being Luby's algorithm \cite{alon86lubys_algorithm,luby86_lubys_alg} from the 1980s running in $O(\log n)$ rounds. 
Since then, the question whether these problems also admit an efficient deterministic algorithm running in $\poly(\log n)$ rounds was open until recently \cite{rozhon_ghaffari2019decomposition}. 

A general way to solve problems like maximal independent set and $(\Delta+1)$-coloring is by first constructing a certain type of clustering of an unweighted graph known as network decomposition \cite{ghaffari_harris_kuhn2018derandomizing,ghaffari_kuhn_maus2017slocal}.  
A $(C,D)$-network decomposition is a decomposition of an unweighted graph into $C$ clusterings: each clustering is a collection of non-adjacent clusters of diameter $D$. 
A network decomposition with parameters $C, D = O(\log n)$ exist and can efficiently be computed in the \congest model if one allows randomization \cite{linial_saks93}. 
However, until recently the best known deterministic algorithms for network decomposition \cite{awerbuch89,panconesi-srinivasan} needed $n^{o(1)}$ rounds and provided a decomposition with parameters $C,D = n^{o(1)}$. 
Only in a recent breakthrough, \cite{rozhon_ghaffari2019decomposition} gave a deterministic algorithm running in $O(\log^7 n)$ \congest rounds and outputting a network decomposition with parameters $C = O(\log n), D = O(\log^3 n)$.

This result was subsequently improved by \cite{ghaffari_grunau_rozhon2020improved_network_decomposition}: their algorithm runs in $O(\log^5 n)$ rounds with parameters $C = O(\log n), D = O(\log^2 n)$. 
However, both of these discussed results offer only a so-called weak-diameter guarantee. 
Recall that this means that every cluster has the property that any two nodes of it have distance at most $D$ in the original graph. However, the cluster can even be disconnected. 

The more appealing strong-diameter guarantee, matching the state-of-the-art weak-diameter gurantee of \cite{ghaffari_grunau_rozhon2020improved_network_decomposition},  was later achieved by \cite{chang_ghaffari2021strong_diameter}. 
However, their algorithm needs $O(\log^{11} n)$ \congest rounds. 

Despite the exciting recent progress, many questions are still open: Can we get faster algorithms with better guarantees? 
Can the algorithms output $D$-separated strong-diameter clusters for $D>2$? 
Can we get algorithms that handle terminals (cf. \cref{thm:steroids_simple})? 
In this work we introduce techniques that help us make some progress on these questions. 


\paragraph{Tree Embeddings and Low-Stretch Spanning Trees}

A very fruitful line of research started with the seminal papers of \cite{alon_karp_peleg_west1995low_stretch_spanning_tree,bartal1996probabilistic} and others.
The authors were interested in approximating metric spaces by simpler metric spaces. In particular, Bartal \cite{bartal1996probabilistic} showed that distances in any metric space can be probabilistically approximated with polylogarithmic distortion by a carefully chosen distribution over trees. 
The proof is constructive and based on low-diameter decompositions. Results of this type are known as probabilistic tree embeddings. 
In \cite{alon_karp_peleg_west1995low_stretch_spanning_tree} showed that the shortest path metric of a weighted graph $G$ can even be approximated by the shortest path metric on a \emph{spanning tree} of $G$ sampled from a carefully chosen distribution. 
A tree sampled from such a distribution is known as a \emph{low-stretch spanning tree}.

Probabilistic tree embeddings and low-stretch spanning trees  are an especially useful tool and have found numerous applications in areas such as approximation algorithms, online algorithms, and network design problems \cite{borodin2005online,haeupler2021tree}.
Importantly, most of the constructions of these objects are based on low-diameter clusterings. 

Many of the \emph{randomized} low-diameter clustering type problems can elegantly be solved in a very parallel/distributed manner using an algorithmic idea introduced in \cite{miller2013parallel}.
We will now sketch their algorithm and then explain why new ideas are needed for our results. 
Consider as an example the randomized version of the low-diameter clustering problem with terminals. 
That is, consider the problem from \cref{thm:steroids_simple}, but instead of the deterministic guarantee (2) on the total number of edges cut, we require that a given edge is cut with probability $\tO\left(\frac{\ell(e)}{\eps R}\right)$.

One way to solve the problem is as follows:
every terminal samples a value from an exponential distribution with mean $\frac{\eps R}{\Theta(\log n)}$. This value is the \emph{head start} of the respective terminal. 
Next, we compute a shortest path forest from all the terminals taking the head starts into account. 
Note that with high probability, the head start of each terminal is at most $\eps R$ and therefore each node $v$ gets clustered to a terminal of distance at most $d(Q,v) + \eps R$. 

To analyze the probability of an edge $e$ being cut, let $u$ be one of the endpoints of $e$. If, taking the head starts into account, the closest terminal is more than $2\ell(e)$ closer to $u$ compared to the second closest terminal, then a simple calculation shows that $e$ is not cut. Therefore, using the memoryless property of the exponential distribution, one can show that $e$ gets cut with probability at most $\frac{\ell(e)}{\eps R / \Theta(\log n)} = \tO \left( \frac{\ell(e)}{\eps R} \right)$.

Unfortunately, this simple and elegant algorithm critically relies on exact distances: if one replaces the exact distance computation with an approximate distance computation with additive error $d_{error}$, then a given edge of length $\ell(e)$ can be cut with probability $\frac{\ell(e) + d_{error}}{\eps R / \Theta(\log n)}$, which is insufficient for short edges.
The left part of \cref{fig:mpx} illustrates this problem.

\begin{figure}
    \centering
    \includegraphics[width=\textwidth]{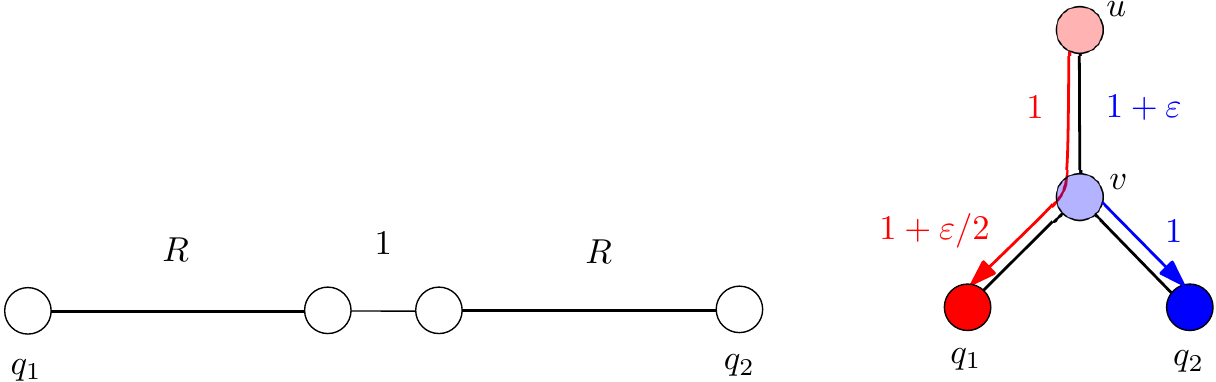}
    \caption{Left: The set of terminals is $Q = \{q_1, q_2\}$. Note that $Q$ is $R$-ruling. Moreover, the head starts of both $q_1$ and $q_2$ are both at most $\eps R$, with high probability. Hence, the probability that $e$ is cut can be equal to one, while we want it to be of order $\widetilde{\Theta}(1/(\eps R))$.\\ 
    Right: All three edge in the graph have length one. Assume that both $q_1, q_2$ have no head starts. The red and blue numbers indicate the computed approximate distances of $q_1, q_2$ to the other two nodes $u,v$ in the graph. Clustering each node to the closest terminal with respect to the computed distances (red and blue arrows) results only in a weak-diameter guarantee.  
   }
    \label{fig:mpx}

\end{figure}

The high-level reason why the algorithm fails with approximate distances is that first the randomness is fixed and only then the approximate distances are computed. 
One way to solve this issue could be to first compute approximate distances $\tilde{d}(q,.)$ from each terminal $q$ separately, then sampling a random head start $hs_q$ for each terminal $q$, followed by clustering each node $v$ to the terminal $q$ minimizing $\tilde{d}(q,v) - hs_q$. 
Even with this approach, an edge might be cut with a too large probability. Moreover, it is no longer possible to obtain a strong-diameter guarantee, as illustrated in the right part of \cref{fig:mpx}.
Also, note that it is not clear how to efficiently compute weak-diameter clusterings with this approach as one has to perform one separate distance computation from each terminal.

In a recent work, \cite{becker_emek_lenzen2020blurry_ball_growing} managed to obtain an efficient low-diameter clustering algorithm using $\poly(\log n)$ approximate distance computations. However, their algorithm has three disadvantages compared to our result: (1) it is randomized, (2) it only gives a weak-diameter guarantee and (3) their result is less general; for example it is not obvious how to extend their algorithm to the setting with terminals.

\subsection{Our Techniques and Contributions}

We give a clean interface for various distributed clustering routines in weighted graphs that allows to give results in different models (distributed and parallel).

\paragraph{Simple Deterministic Strong-Diameter Network Decomposition in \congest}

In the previous section, we mentioned that the state-of-the-art strong-diameter network decomposition algorithm of \cite{chang_ghaffari2021strong_diameter} runs in $O(\log^{11})$ \congest rounds and produces clusters with diameter $D = O(\log^2 n)$.

Our first result improves upon their algorithm by giving an algorithm with the same guarantees running in $O(\log^5 n)$ \congest rounds.

\begin{theorem}
\label{thm:decomposition_informal}
There is a deterministic \congest algorithm computing a network decomposition with $C = O(\log n)$ clusterings such that each cluster has strong-diameter $O(\log^2 n)$. 
The algorithm runs in $O(\log^5 n)$ \congest rounds. 
\end{theorem}

Note that the round complexity of our algorithm matches the complexity of the weak-diameter algorithm of \cite{ghaffari_grunau_rozhon2020improved_network_decomposition}. 
This is because both our result and the result of \cite{chang_ghaffari2021strong_diameter} use the weak-diameter algorithm of \cite{ghaffari_grunau_rozhon2020improved_network_decomposition} as a subroutine. 

We prove \cref{thm:decomposition_informal} in \cref{sec:unweighted_strong_clustering} and the technical overview of our approach is deferred to \cref{sec:strong_intuition}. Here, we only note that on a high-level our algorithm can be seen as a derandomization of the randomized algorithm of \cite{miller2013parallel}. That is, instead of clusters, the algorithm operates with nodes and assigns ``head starts'' to them in a careful manner. 

\paragraph{Simple Blurry Ball Growing Procedure}

The blurry ball growing problem is defined as follows: given a set $S$ and distance parameter $D$, we want to find a superset $S^{sup} \supseteq S$ such that the following holds. First, for any $v \in S^{sup}$ we have $d_{G[S^{sup}]}(S, v) \le D$, that is, the set $S$ ``does not grow too much''. 
On the other hand, in the randomized variant of the problem we ask for each edge $e$ to be cut by $S^{sup}$ with probability $O(\ell(e) / D)$, while in the deterministic variant of the problem we ask for the total number of edges cut to be at most $O(\sum_{e \in E(G)} \ell(e)/D)$. 

Here is a simple application of this problem: suppose that we want to solve the low-diameter clustering problem where each edge needs to be cut with probability $\ell(e)/D$ and clusters should have diameter $\tO(D)$. Assume we can solve the separated clustering problem, that is, we can construct a clustering $\fC$ such that the clusters are $D$-separated and their diameter is $\tO(D)$. To solve the former problem, we can simply solve the blurry ball growing problem with $S = \bigcup_{C \in \fC} C$ and $D_{blurry} = D/3$. This way, we ``enlarge'' the clusters of $\fC$ only by a nonsignificant amount, while achieving the edge cutting guarantee. 

The blurry ball growing problem was defined and its randomized variant was solved in \cite[Theorem 3.1]{becker_emek_lenzen2020blurry_ball_growing}. 
Since blurry ball growing is a useful subroutine in our main clustering result, we generalize their result by giving an efficient algorithm solving the deterministic variant. 
Furthermore, we believe that our approach to solving that problem is simpler: we require the approximate distance oracle to be $(1+1/\log n)$-approximate instead of $\left(1+\left( \frac{\log\log n}{\log n} \right)^2\right)$-approximate. 

\begin{restatable}{theorem}{blurryInformal}
\label{thm:blurry_informal}
Given a weighted graph $G$, a subset of its nodes $S$ and a parameter $D > 0$, there is a deterministic algorithm computing a superset $S^{sup} \supseteq S$ such that $\max_{v \in S^{sup}} d_{G[S^{sup}]}(S,v) \leq D$, and moreover, 
\[
\sum_{e \in E(G) \cap (S^{sup} \times (V(G) \setminus S^{sup}))} \ell(e) = O\left(\sum_{e \in E(G)} \ell(e)/D\right).
\]
The algorithm uses $O(\log D)$ calls to an $(1+1/\log D)$-approximate distance oracle. 
\end{restatable}

Our deterministic algorithm is a standard derandomization of the following simple randomized algorithm solving the randomized variant of the problem. The randomized algorithm is based on a simple binary search idea: in each step we flip a fair coin and decide whether or not we ``enlarge'' the current set $S_i$ by adding to it all nodes of distance at most roughly $D/2^i$. We start with $S_0 = S$ and the final set $S_{\log_2 D} = S^{sup}$. Hence, we need $O(\log D)$ invocations of the approximate distance oracle. 
We prove a more general version of \cref{thm:blurry_informal} in \cref{sec:blurry} and give more intuition about our approach in \cref{sec:blurry_intuition}. 

\paragraph{Main Contribution: A General Clustering Result}

We will now state a special case of our main clustering result. The clustering problem that we solve generalizes the already introduced low-diameter clustering problem that asks for a partition of the vertex set into clusters such that only a small amount of edges is cut. 
In our more general clustering problem we are also given a set of \emph{terminals} $Q \subseteq V(G)$ as input. Moreover, we are given a parameter $R$ such that $Q$ is $R$-ruling. 
Each cluster of the final output clustering has to contain at least one terminal. Moreover, one of these terminals should $(1+\eps)R$-rule its cluster. 

We note that in order to get the classical low-diameter clustering with parameter $D$ as an output of our general result, it suffices to set $Q = V(G), R = D$ and $\eps = 1/2$. 

A more general version of \cref{thm:steroids_simple} is proven in \cref{sec:main_clustering}. The intuition behind the algorithm is explained in \cref{sec:main_clustering_intro}. 
Here, we note that the algorithm combines the clustering idea of the algorithm from \cref{thm:decomposition_informal} and uses as a subroutine the blurry ball growing algorithm from \cref{thm:blurry_informal}. 

Another corollary of our general clustering result is the following theorem.

\begin{restatable}{theorem}{steroidsPadded}[A corollary of \cref{thm:steroids}]
\label{thm:steroids_padded}
We are given an input weighted graph $G$, a distance parameter $D$ and each node $v \in V(G)$ has a \emph{preferred radius} $r(v) > 0$. 

There is a deterministic distributed algorithm constructing a partition $\fC$ of $G$ that splits $V(G)$ into two sets $V^{good} \sqcup V^{bad}$ such that

\begin{enumerate}
    \item Each cluster $C \in \fC$ has diameter $\tO(D)$.
    \item For every node $v \in V^{good}$ such that $v$ is in a cluster $C$ we have $B_G(v, r(v)) \subseteq C$.
    \item For the set $V^{bad}$ of nodes we have 
    \[
    \sum_{v \in V^{bad}} r(v) = \frac{1}{2D} \cdot \sum_{v \in V(G)} r(v). 
    \]
\end{enumerate}

The algorithm needs $\tO(1)$ calls to an $(1+1/\poly\log n)$-approximate distance oracle. 

\end{restatable}

One reason why we consider each node to have a preferred radius is that it allows us to deduce \cref{thm:steroids_simple} from our general theorem by considering the subdivided graph where each edge is split by adding a node ``in the middle of it'', with a preferred radius of $\ell(e)$. 

Let us now compare the clustering of \cref{thm:steroids_padded} with the $D$-separated clustering that we already introduced. Recall that in the $D$-separated clustering problem, we ask for clusters with radius $\tO(D)$ and require the clusters to be $D$-separated. Moreover, only half of the nodes should be unclustered. 

In our clustering, we can choose $r(v) = D$ for all nodes $v \in V(G)$, we again get clusters of diameter $\tO(D)$ and only half of the nodes are bad. 
The difference with the $D$-separated clustering is that we cluster all the nodes, but we require the good nodes to be ``$D$-padded''. 

This is a slightly weaker guarantee then requiring the clusters to be $D$-separated: we can take any solution of the $D$-separated problem, and enlarge each cluster by adding all nodes that are at most $D/3$ away from it. 
We mark all original nodes of the clusters as good and all the new nodes as bad. Moreover, each remaining unclustered node forms its own cluster and is marked as bad. This way, we solve the special case of \cref{thm:steroids_padded} with the padding parameter $D/3$. 
We do not know of an application of $D$-separated clustering where the slightly weaker $D$-padded clustering of \cref{thm:steroids_padded} does not suffice. 
However, we also use a different technique to solve the $D$-separated problem. 

\begin{theorem}
\label{thm:separated_clustering_simple}
We are given a weighted graph $G$ and a separation parameter $D>0$. 
There is a deterministic algorithm that outputs a clustering $\fC$ of $D$-separated clusters of diameter $\tO(D)$ such that at least $n/2$ nodes are clustered. 

The algorithm needs $\tO(1)$ calls to an $(1+1/\poly\log(n))$-approximate distance oracle computing approximate shortest paths from a given set up to distance $\tO(D)$. 
\end{theorem}

The algorithm is based on the ideas of the weak-diameter network decomposition result of \cite{rozhon_ghaffari2019decomposition} and the strong-diameter network decomposition of \cite{chang_ghaffari2021strong_diameter}. 
Since shortest paths up to distance $D$ can be computed in unweighted graphs by breadth first search, we get as a corollary that we can compute a separated strong-diameter network decomposition in unweighted graphs. No $\tO(D)$-round deterministic \congest algorithm for separated strong-diameter network decomposition was known. 

\begin{restatable}{corollary}{separatedDecomposition}[$D$-separated strong-diameter network decomposition]
\label{thm:separated_decomposition}
We are given an unweighted graph $G$ and a separation parameter $D>0$. 
There is a deterministic algorithm that outputs $O(\log n)$ clusterings $\{\fC_1, \dots, \fC_{O(\log n)}\}$ such that
\begin{enumerate}
    \item Each node $u \in V(G)$ is contained in at least one clustering $\fC_i$.
    \item Each clustering $\fC_i$ consists of $D$-separated clusters of diameter $\tO(D)$. 
\end{enumerate}
The algorithm needs $\tO(D)$ \congest rounds. 
\end{restatable}

\subsection{Roadmap}

The paper is structured as follows. 
In \cref{sec:preliminaries} we define some basic notions and models that we work with in the paper. 
\cref{sec:unweighted_strong_clustering} contains the proof of \cref{thm:decomposition_informal}. We believe that an interested reader should understand \cref{sec:unweighted_strong_clustering} even after she skips \cref{sec:preliminaries}. 
In \cref{sec:blurry}, we prove a general version of \cref{thm:blurry_informal} and the main clustering result that generalizes \cref{thm:steroids_simple} is proven in \cref{sec:main_clustering}.

\section{Preliminaries}
\label{sec:preliminaries}

In this preliminary section, we first explain the terminology used in the paper. Then, we review the notation we use to talk about clusterings and the distributed models we work with. 
Finally, we explain the language of distance oracles that we use throughout the paper to make our result as independent on a particular choice of a distributed/parallel computational model as possible. 

\paragraph{Basic Notation} 

The subgraph of a graph $G$ induced by a subset of its nodes $U \subseteq V(G)$ is denoted by $G[U]$. 
A weighted graph $G$ is an unweighted graph together with a weight (or length) function $\ell$. This function assigns each edge $e \in E(G)$ a polynomially bounded nonnegative weight $\ell(e) \ge 0$. 
We will assume that all lengths are polynomially bounded, i.e., $\ell(e) \le n^C$ for some absolute constant $C$. This implies that each weight can be encoded by $O(\log n)$ bits. 
We denote by $d_{G}(u, v)$ the weight (i.e., length) of the shortest path between two nodes $u, v \in V(G)$. 
We sometimes drop the subscript when the graph is clear from context and write just $d(u, v)$. 
The distance function naturally extends to sets by $d(A, B) = \min_{a \in A, b \in B} d(a,b)$ and we also write $d(u , S)$ instead of $d(\{u\}, S)$. 

We say that $H$ is a subgraph of a weighted graph $G$ and write $H \subseteq G$ if $V(H) \subseteq V(G)$, $E(H) \subseteq E(G)$ and for every $e \in E(H)$, $\ell_H(e) = \ell_G(e)$. 
Given a weighted graph $G$, a weighted rooted (sub)forest in $G$ is a forest $F$ which is a subgraph of $G$. 
Moreover, each component of $F$ contains a special node -- a root -- that defines a natural orientation of edges of $F$ towards a unique root. By $d_F(u, v)$ we mean the distance in the unoriented graph $F$, i.e., $d_F$ is a metric. 
For any $v \in V(F)$ we denote by $\root_F(v)$ the unique root node in $V(F)$ that lies in the same component of $F$ as $v$. 
We also use the shorthand $d_F(v) = d_F(\root_F(v), v)$. 
A ball $B_G(u, r) \subseteq V(G)$ is a set of nodes consisting of those nodes $v \in V(G)$ with $d_G(u,v) \le r$.

\paragraph{Weight, Radius and Delay Functions}

Sometimes we need nonnegative and polynomially bounded functions that assign each vertex or edge of a given graph such that their domain is the set $V(G)$, $E(G)$ or a subset. 
One should think of these functions as parameters of the nodes (edges) of the input graph in the sense that during the algorithms, each node $u$ starts with an access to the value of these functions at $u$. 

There are three functions that we need:

\begin{enumerate}
    \item A function $\mu$ assigning each vertex $v$ (edge $e$) of a given graph a \emph{weight} $\mu(v)$ ($\mu(e)$); we use $\mu(U) := \sum_{u \in U} \mu(u)$. 
    \item A function $r$ assigning each vertex $v$ of a given graph a \emph{preferred radius} $r(v)$.
    \item A function $\del$ assigning a subset of nodes $Q$ a \emph{delay}; For a subset $Q' \subseteq Q$, we define $d_{\del}(Q', v) := \min_{q \in Q'} \del(q) + d(q,v)$. 
\end{enumerate}

\paragraph{Clustering Notation}

Next, we define the notation that is necessary for stating our clustering results. 

\begin{description}

\item[Cluster] A cluster $C$ is simply a subset of nodes of $V(G)$. 
We use $\diam(C)$ to denote the diameter of a cluster $C$, i.e., the diameter of the graph $G[C]$. 
When we construct a cluster $C$, we are also often constructing a (small diameter) tree $T_C$ with $V(T_C) = C$. 
In \cref{sec:unweighted_strong_clustering} we use a result from \cite{ghaffari_grunau_rozhon2020improved_network_decomposition} that constructs so-called weak-diameter clusters. 
A weak-diameter cluster is a cluster $C$ together with a (small diameter) tree $T_C$ such that $V(T_C) \supseteq C$. 

\item[Padding]
A node $v \in C$ is $r$-padded in the cluster $C$ of a graph $G$ if it is the case that $B_G(v, r) \subseteq C$. 

\item[Separation]
Suppose we have two disjoint clusters $C_1, C_2$. 
We say that they are $D$-separated in $G$ if $d_G(C_1, C_2) \ge D$. 

\item[Clustering and Partition] A \emph{clustering} $\fC$ is a family of disjoint clusters.
If the clustering covers all nodes of $G$, that is, if $\bigcup_{C \in \fC} C = V(G)$, we refer to the clustering as a \emph{partition}. 

The \emph{diameter} $\diam(\fC)$ of a clustering $\fC$ is defined as $\diam(\fC) = \max_{C \in \fC} \diam(C)$. 
A clustering $\fC$ is $D$-separated if every two clusters $C_1 \not= C_2 \in \fC$ are $D$-separated. 

\item[Cover] A \emph{cover} $\{\fC_1, \fC_2, \dots, \fC_q\}$ is a collection of clusterings or partitions. 

\end{description}

\subsection{Computational Models}

\paragraph{\congest Model \cite{peleg00book}}

We are given an undirected graph $G$ also called the ``communication network''. 
Its vertices are also called nodes and they are individual computational units, i.e., they have their own processor and private memory. 
Communication between the nodes occurs in synchronous rounds. 
In each round, each pair of nodes adjacent in $G$ exchange an $b = O(\log n)$-bit message. 
Nodes perform arbitrary computation between rounds. 
Initially, nodes only know their unique $O(\log n)$-bit ID and the IDs of adjacent nodes in case of deterministic algorithms. 
In case of randomized algorithms, every node starts with a long enough random string containing independently sampled random bits. 
Each node also starts with a polynomial upper bound on the number of nodes, $n$. 

Unless stated otherwise, we always think of $G$ as a weighted graph, where the weights are provided in a distributed manner. 

In all our results, in each round, each node $v$ can run an algorithm whose \pram work is at most $\tO(\deg(v))$ and depth at most $\tO(1)$. 
Note that this allows the model to compute e.g. simple aggregation operations of the messages received by the neighbors such as computing the minimum or the sum.

\paragraph{Oracle Definition}

Except of simple local communication and computation captured by the above \congest model, our algorithms can be stated in terms of simple primitives such as computing approximate shortest paths or aggregating some global information.  
To make our results more model-independent and more broadly applicable, we abstract these primitives away as calls to an oracle. 
We next define the oracles used in the paper.

\begin{restatable}[Approximate Distance Oracle $\oDist_{\eps,D}$]{definition}{distanceoracle}
\label{def:oracle_dist}
This oracle is parameterized by a  \emph{distance} parameter $D > 0$ and a \emph{precision} parameter $\eps \geq 0$. 

The input to the oracle consists of three parts. First, a weighted graph $H \subseteq G$. Second, a subset $S \subseteq V(H)$. Third, for each node $s \in S$ a delay $\del(s)$. If the third input is not specified, set $\del(s) = 0$ for every $s \in S$.

The output is a weighted forest $F \subseteq H$ rooted at some subset $S' \subseteq S$. The output has to satisfy the following:
\begin{enumerate}
    \item For every $v \in V(F)$, $\del(\root_F(v)) + d_F(v) \leq (1+\eps)d_{H,\del}(S,v) \leq (1+\eps)D$.
    \item For every $v \in V(H)$, if $d_{H,\del}(S,v) \leq D$, then $v \in V(F)$.
\end{enumerate}

\end{restatable}

\begin{restatable}[Forest Aggregation Oracle $\oForestAgg_D$]{definition}{forestoracle}

\label{def:oracle_forest_agg}
The input consists of two parts. First, a weighted and rooted forest $F \subseteq G$ with $d_F(v) \leq D$ for every $v \in V(F)$.
Second, an integer value $x_v \in \{0,1,\ldots,\poly(n)\}$ for every node $v \in V(F)$.
The oracle can be used to compute a sum or a minimum. If we compute a sum, the oracle outputs for each node $v \in V(F)$ the two values $\sum_{v \in A(v)} x_v$ and $\sum_{v \in D(v)} x_v$, where $A(v)$ and $D(v)$ denote the set of ancestors and descendants of $v$ in $F$, respectively. Computing the minimum is analogous. 
\end{restatable}

\begin{restatable}[Global Aggregation Oracle $\oGlobalAgg$]{definition}{aggregationoracle}

\label{def:oracle_global_agg}
The input consists of an integer value $x_v \in \{0,1,\ldots,\poly(n)\}$ for every node $v \in V(G)$.
The output of the oracle is $\sum_{v \in V(G)} x_v$.
\end{restatable}

Whenever we say e.g. that ``the algorithm runs in $T$ steps, with each oracle call having distance parameter at most $D$ and precision parameter $\eps$'', we mean that the algorithm runs in $T$ \congest rounds, and in each \congest round the algorithm performs at most one oracle call. Moreover, when the oracle is parameterized by a distance parameter or/and a precision parameter, then the distance parameter is at most $D$ and the precision parameter is at most $\eps$ .

\paragraph{Compilation to Distributed and Parallel Models}

The theorem below is a direct consequence of the deterministic approximate shortest path paper of \cite{RGHZL2022sssp}. They show that the approximate distance oracle with precision parameter $\eps = 1/\poly\log(n)$ can be implemented in the bounds claimed in bullet points 1 to 4. We note that the results 2 to 4 follow from the theory of universal-optimality in the \congest model \cite{goranci2022universally,ghaffari2021universally}. The bullet point 5 follows from the fact that the distance oracle in unweighted graphs can be implemented by breadth first search. 

\begin{theorem}\label{thm:congestpa-simulation}
  Suppose that for a given problem there is an algorithm that runs in  $T = \poly\log(n)$ steps, with each oracle call having precision parameter $\eps = 1/\poly\log(n)$. 
  Then, the problem can be solved in the following settings with the following bounds on the complexity. 
  
  \begin{enumerate}
  \item In \pram, there is a deterministic algorithm with $\tO(m+n)$ work and $\poly\log n$ depth. 
    \item In \congest, there is a deterministic algorithm with $\tO(\hopDiameter{G} + \sqrt{n})$ rounds.~\cite{ghaffari_haeupler2016shortcuts_planar}
  \item In \congest, there is a deterministic algorithm for any minor-free graph family with $\tilde{O}(\hopDiameter{G})$ rounds (the hidden constants depend on the family).~\cite{ghaffari_haeupler2021shortcuts_in_minor_closed}
  \item If $\shortcutQuality{G} \le n^{o(1)}$, there is a randomized algorithm in the \congest model with $n^{o(1)}$ rounds. See ~\cite{GHR21} for the definition of $\shortcutQuality{G}$ and the proof. 
  \item If only the distance oracle $\oDist_{\eps, D}$ is used and the graph is unweighted, there is a deterministic \congest algorithm with $\tO(D)$ rounds, even for $\eps = 0$. 
\end{enumerate}
  
\end{theorem}

\section{Strong-Diameter Clustering in $O(\log^4 n)$ \congest Rounds}
\label{sec:unweighted_strong_clustering}

In this section, we present an algorithm clustering a constant fraction of the vertices into non-adjacent clusters of diameter $O(\log^2 n)$ in $O(\log^4 n)$ \congest rounds.

\begin{theorem}
\label{thm:strong_diameter_clustering}
Consider an unweighted $n$-node graph $G$ where each node has a unique $b = O(\log n)$-bit identifier. There is a deterministic algorithm computing a  $2$-separated $O(\log^2 n)$-diameter clustering $\fC$  with $|\bigcup_{C \in \fC} C| \ge n/3$ in $O(\log^4 n)$ \congest rounds.
\end{theorem}

Note that in the above theorem, $2$-separated clustering is equivalent to positing that the clusters of $\fC$ are not adjacent. 

Our algorithm is quite simple and in some aspects similar to the deterministic distributed clustering algorithm of \cite{rozhon_ghaffari2019decomposition}. 
Let us explain the main difference. During their algorithm, one works with a set of clusters that expand or shrink and which progressively become more and more separated. 
In our algorithm, we instead focus on \emph{potential cluster centers}. These centers preserve a ``ruling property'' that asserts that every node is close to some potential cluster center. This in turn implies that running a breadth first search from the set of potential cluster centers always results in a set of (not necessarily separated) small diameter clusters. This way, we make sure that the final clusters have small strong-diameter, whereas in the algorithm of \cite{rozhon_ghaffari2019decomposition} the final clusters have only small weak-diameter. 

An important downside compared to their algorithm is that we rely on global coordination. To make everything work, we hence need to start by using the state-of-the-art algorithm for weak-diameter clustering. This allows us to use global coordination inside each weak-diameter cluster and run our algorithm in each such cluster in parallel. This is the reason why our algorithm needs $O(\log^4 n)$ \congest rounds. 
In fact, the main routine runs only in $O(\log^3 n)$ rounds, but first we need to run the fastest deterministic distributed algorithm for weak-diameter clustering from \cite{ghaffari_grunau_rozhon2020improved_network_decomposition} that needs $O(\log^4 n)$ rounds which dominates the round complexity.

\subsection{Intuition and Proof Sketch of \cref{thm:strong_diameter_clustering}}
\label{sec:strong_intuition}
Our algorithm runs in $b$ phases; one phase for each bit in the $b$-bit node identifiers.
During each phase, up to $\frac{n}{3b}$ of the nodes are removed and declared as unclustered.
Hence, at most $\frac{n}{3}$ nodes are declared as unclustered throughout the $b$ phases, with all the remaining nodes
being clustered.

We set $G_0 = G$ and define $G_{i+1}$ as the graph one obtains from $G_{i}$ by deleting all the nodes from $G_i$ which are declared as unclustered during phase $i$.
Besides removing nodes in each phase $i$, the algorithm works with a set of  \emph{potential cluster centers} $Q_i$.
Initially, all the nodes are potential cluster centers, that is, $Q_0 = V(G)$. 
During each phase, some of the potential cluster centers stop being potential cluster centers. 
At the end, each potential cluster center will in fact be a cluster center. More precisely, each connected component of $G_b$ contains exactly one potential cluster center and the diameter of each connected component of $G_b$ is $O(\log^2 n)$.

For each $i \in \{0,1,\ldots, b\}$, the algorithm maintains two invariants. The \emph{ruling invariant} states that each node in $G_i$ has a distance of at most $6ib$ to the closest potential cluster center in $Q_i$. 
The \emph{separation invariant} states that two potential cluster centers $u$ and $v$ can only be in the same connected component of $G_i$ if the first $i$ bits of
their identifiers coincide. 
Note that this condition is trivially satisfied at the beginning for $i = 0$ and for $i = b$ it implies that each connected component contains at most one 
potential cluster center.

The goal of the $i$-th phase is to preserve the two invariants.
To that end, we partition the potential cluster centers in $Q_i$ based on the $(i+1)$-th bit of their identifiers into two sets $Q^R_{i}$ and $Q^B_{i}$.
In order to preserve the separation invariant, it suffices to separate the nodes in $Q^R_{i}$ from the nodes in $Q^B_{i}$. 
One way to do so is as follows: Each node in $G_{i}$ clusters itself to the closest potential cluster center in $Q_{i}$. 
In that way, each node is either part of a red cluster with a cluster center in $Q^R_{i}$ or a blue cluster with a cluster center in $Q^B_{i}$. 
Now, removing all the nodes in blue clusters neighboring a red cluster would preserve both the separation invariant as well as the ruling invariant. However, the number of removed nodes might be too large.

In order to ensure that at most $\frac{n}{3b}$ nodes are deleted, we do the following: for each node $v$, let $\diff(v) = d_i(Q_i^R,v) - d_i(Q_i^B,v)$. We now define $K_j = \{v \in V_i \colon \diff(v) = j\}$ and let $j^* = \arg \min_{j \in \{0,2,4,\ldots, 6b - 2\}} |K_j \cup K_{j+1}|$. Then, we declare all the nodes in $K_{j^*} \cup K_{j^*+1}$ as unclustered. Moreover, each blue potential cluster center $v$ with $\diff(v) < j^*-1$ stops being a potential cluster center. 

\begin{figure}
    \centering
    \includegraphics[width = \textwidth]{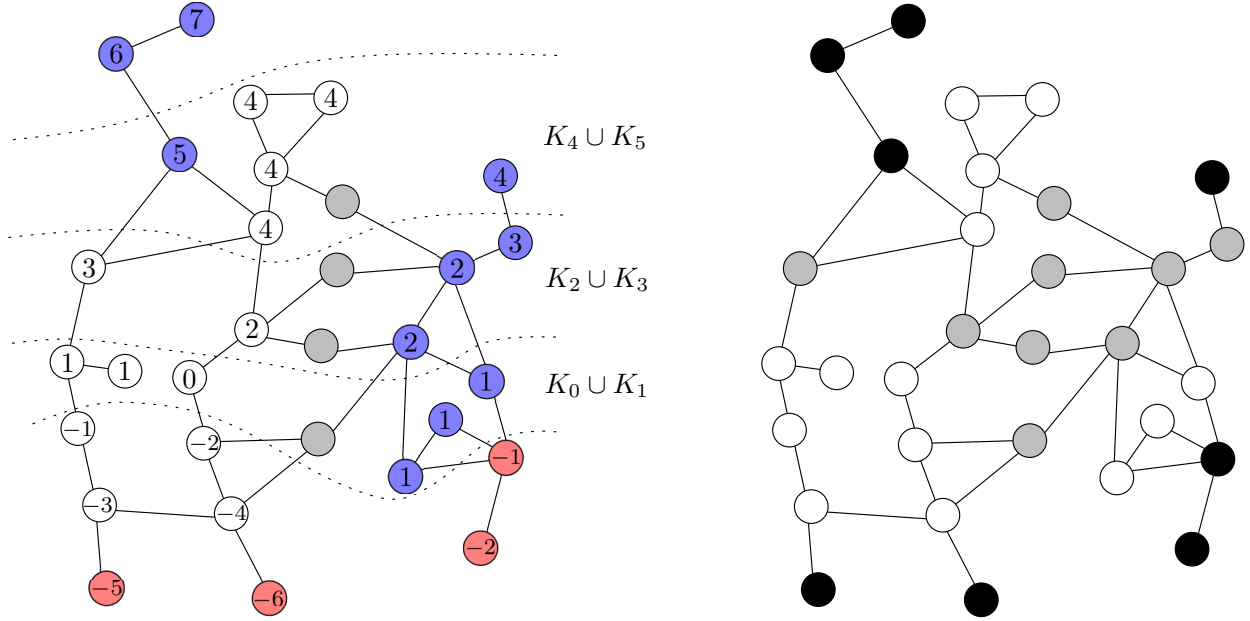}
    \caption{Left: The left picture shows a possible situation after one phase of the algorithm and what happens during the second phase. Note that by deleting several nodes (grey) in the first phase we have already separated the potential cluster centers in two disconnected components. 
    It can be checked that the set $Q$ of potential cluster centers (red and blue nodes) is $3$-ruling. We split it into two sets $Q^\fR \sqcup Q^\fB$ (red and blue, respectively) based on the corresponding bit in their unique identifier. We compute the difference $\diff(u) = d(Q^\fR, u) - d(Q^\fB, u)$ for every node $u \in V(G)$ for which $\diff(u)$ is not too large or too small (in the picture, the differences of all nodes are computed and written in the node). 
    We define $K_j$ as the set of nodes whose difference is $j$. We choose the smallest of the sets $K_0 \cup K_1, K_2 \cup K_3, K_4 \cup K_5$ -- it is the set $K_2 \cup K_3$ containing five nodes. \\
    Right: We delete all nodes in $K_2 \cup K_3$ and blue nodes in $K_0 \cup K_1$ are not potential cluster centers anymore. This way we successfully separate the blue potential cluster centers from red ones. Note that the new set $Q' \subseteq Q$ of potential cluster centers (black) is only $4$-ruling. 
    }
    \label{fig:log4}
\end{figure}

We remark that the described algorithm cannot efficiently be implemented in the distributed \congest model. 
The reason is that deciding whether a given index $j$ is good or not requires global coordination.
In particular, the communication primitive we need can be described as follows: each node $v$ is assigned $O(b)$ numbers $\{x_{v,j} \in \{0,1\} \colon j \in \{0,2,4,\ldots,6b-2\} \}$
with $x_{v,j} = 1$ if $v$ would be removed if $j = j^*$ and $0$ otherwise.
Now, each node has to learn $\sum_{v \in V(G)} x_{v,j}$ for each $j \in \{0,2,4,\ldots,6b-2\}$.
We denote by $\fO^{count}$ the oracle for this communication primitive. By formalizing the high-level overview, we then obtain the following theorem.

\begin{theorem}
\label{thm:congest_clustering_oracle}
    Consider an unweighted $n$-node graph $G$ where each node has a unique $b$-bit identifier. There is a deterministic algorithm computing a $2$-separated $O(b^2)$-diameter clustering $\fC$  with $|\bigcup_{C \in \fC} C| \ge (2/3)n$ in $O(b^3)$ \congest rounds and performing $O(b)$ oracle calls to $\fO^{count}$.
\end{theorem}

We note that the theorem does not assume $b = O(\log n)$.

Before giving a formal proof of \cref{thm:congest_clustering_oracle} in \cref{sec:formal_oracle_version}, we first show formally how one can use it to proof \cref{thm:strong_diameter_clustering}.

\subsection{Proof of \cref{thm:strong_diameter_clustering}}

\begin{proof}[Proof of \cref{thm:strong_diameter_clustering}]

The algorithm starts by computing a clustering $\fC^{weak} = \{C^{weak}_1,C^{weak}_2,\ldots,C^{weak}_N\}$ with weak-diameter $O(\log^2 n)$ such that $\fC^{weak}$ clusters at least half of the nodes. This can be computed in $O(\log^4 n)$ \congest rounds by invoking the following theorem from \cite{ghaffari_grunau_rozhon2020improved_network_decomposition}.

\begin{theorem}[Restatement of Theorem 2.2 in \cite{ghaffari_grunau_rozhon2020improved_network_decomposition}]
\label{thm:ggr_clustering}
Consider an arbitrary $n$-node graph $G$ where each node has a unique $b = O(\log n)$-bit identifier. There is a deterministic algorithm that in $O(\log^4 n)$ rounds of the \congest model computes a $2$-separated clustering $\fC$ such that $|\bigcup_{C \in \fC} C| \ge n/2$.
For each cluster $C$, the algorithm returns a tree $T_C$ with diameter $O(\log^2 n)$ such that $C \subseteq V(T_C)$. Each vertex in $G$ is in $O(\log n)$ such trees.
\end{theorem}

Now, for each $i \in [N]$, let $\fC^{strong}_i$ denote the clustering one obtains by invoking \cref{thm:congest_clustering_oracle} with input graph $G[C^{weak}_i]$. Then, the algorithm returns the clustering $\fC^{strong} = \bigcup_{i=1}^n \fC^{strong}_i$. \\
We first show that $\fC^{strong}$ is a $2$-separated clustering with diameter $O(\log^2 n)$ clustering at least $(1/3)n$ nodes.

First, the clustering $\fC^{strong}$ is $2$-separated: This directly follows from the fact that $\fC^{weak}$ is $2$-separated and for $i \in [N]$, $\fC^{strong}_i$ is $2$-separated. 
Moreover, the clustering has diameter $O(\log^2 n)$: This follows from the fact that each cluster in $\fC^{weak}_i$ has diameter $O(b^2)$ and $b = O(\log n)$.

It remains to show that $|\bigcup_{C \in \fC^{strong}} C| \geq n/3$.
We have
\[|\bigcup_{C \in \fC^{strong}} C| = \sum_{i=1}^N |\bigcup_{C \in \fC^{strong}_i} C| \geq \sum_{i=1}^N \frac{2|C^{weak}_i|}{3} = \frac{2}{3} |\bigcup_{C \in \fC^{weak}} C| \geq \frac{2}{3}\frac{1}{2}n = \frac{1}{3}n\]
where the first inequality follows from the guarantees of \cref{thm:congest_clustering_oracle} and the second follows from the guarantees of \cref{thm:ggr_clustering}. 

Finally, we discuss an efficient \congest implementation of the algorithm:
For every $i \in [N]$, we need to show that we can compute $\fC^{strong}_i$ in $O(\log^3 n)$ \congest rounds in each cluster $C^{weak}_i \in \fC^{weak}$. 
Moreover, the communication capacity in each round is limited: For the computation inside $C^{weak}_i$, we can use the full capacity of $O(\log n)$ bits along edges contained in $G[C^{weak}_i]$, but only a single bit for edges contained in the tree $T_{C^{weak}_i}$, and no communication for all other edges. The reason why we have the capacity of one bit per edge of $T_{C^{weak}_i}$ is that by \cref{thm:ggr_clustering}, each node of $G$ and hence each edge of $G$ is contained in $O(\log n)$ different trees $T_C$, hence by assuming without loss of generality that $b$ is large enough, each edge can allocate one bit per tree it is in. 


According to \cref{thm:congest_clustering_oracle}, for $i \in [N]$, we can compute $\fC^{strong}_i$ in $O(b^3) = O(\log^3 n)$ \congest rounds together with performing $O(b) = O(\log n)$ oracle calls to $\fO^{count}$ in $G[C^{weak}_i]$.
The \congest rounds use only edges of $G[C^{weak}_i]$, so we only need to discuss the implementation of the calls to $\fO^{count}$.  To that end, we will use the following variant of \cite[Lemma 5.1]{ghaffari_grunau_rozhon2020improved_network_decomposition}. \footnote{The Lemma 5.1 in \cite{ghaffari_grunau_rozhon2020improved_network_decomposition} proves this result only for $k=1$, but the generalization for bigger $k$ is straightforward. }

\begin{lemma}[A variant of the ``pipelining'' Lemma 5.1 from \cite{ghaffari_grunau_rozhon2020improved_network_decomposition}]
\label{lem:pipelining}
Consider the following problem. Let $T$ be a rooted tree of depth $d$. Each node $u$ has $k$ $m$-bit numbers $x_u^1, x_u^2, \dots, x_u^k$. In one round of communication, each node can send a $b$-bit message, $b \le m$, to all its neighbors in $T$. There is a protocol such that in $O(d + km/b)$ message-passing rounds on $T$ performs the following operations:
\begin{enumerate}
    \item Broadcast: the root of $T$, $r$, sends $x_r^1, \dots, x_r^k$ to all nodes in $T$. 
    \item Sum: The root $r$ computes the value of $\sum_{u \in T} x_u^i \mod 2^{O(m)}$ for every $1 \le i \le k$. 
\end{enumerate}
\end{lemma}

In our case, to implement $\fO^{count}$ we first use the sum operation and afterwards the broadcast operation from the statement of \cref{lem:pipelining} on the tree $T_{C^{weak}_i}$. 
Note that we use the following parameters: $d_{\ref{lem:pipelining}} = O(\log^2 n)$ since this is the diameter of $T_{C^{weak}_i}$ by \cref{thm:ggr_clustering}; $ k_{\ref{lem:pipelining}} = O(\log n)$ since we need to aggregate $O(\log n)$ different sums; $m_{\ref{lem:pipelining}} = O(\log n)$ as this is the size of the messages we are broadcasting; $b_{\ref{lem:pipelining}} = 1$ as this is the capacity of the channel. Therefore, the oracle $\fO^{count}$ is implemented in 
$O(d_{\ref{lem:pipelining}} + k_{\ref{lem:pipelining}} m_{\ref{lem:pipelining}}/b_{\ref{lem:pipelining}}) = O(\log^2 n)$ rounds. We need to call it $O(\log n)$ times, hence the overall round complexity is $O(\log^3 n)$, as desired. 

\end{proof}

\subsection{Proof of \cref{thm:congest_clustering_oracle}}
\label{sec:formal_oracle_version}
In this section, we formalize the proof sketch given in \cref{sec:strong_intuition}. 

\begin{proof}[Proof of \cref{thm:congest_clustering_oracle}]

The algorithm computes two sequences $V_0 := V(G) \supseteq V_1 \supseteq \ldots \supseteq V_b$ and $Q_0 := V(G) \supseteq Q_1 \supseteq \ldots \supseteq Q_b$. For $i \in \{0,1,\ldots,b\}$, we define $G_i = G[V_i]$ and $d_i = d_{G_i}$. Besides $Q_i \subseteq V_i$, the following three invariants will be satisfied:

\begin{enumerate}
    \item Separation Invariant: Let $u,v \in Q_i$ be two nodes that are contained in the same connected component in $G_i$. Then, the first $i$ bits of the identifiers of $u$ and $v$ agree.
    \item Ruling Invariant: For every node $v \in V_i$, $d_i(Q_i,v)\leq 6ib$.
    \item Deletion Invariant: We have $|V_i| \ge n - \frac{in}{3b}$. 
\end{enumerate}

It is easy to verify that setting $V_0 = Q_0 = V(G)$ results in the three invariants being satisfied for $i = 0$. For $i = b$, the separation invariant implies that every connected component in $G_b$ contains at most one vertex in $Q_b$. Together with the ruling invariant, this implies that the diameter of every connected component in $G_b$ is $O(b^2)$. Moreover, the deletion invariant states that $|V_b| \geq n - \frac{bn}{3b} = (2/3)n$. Hence, the connected components of $G_b$ define a $2$-separated clustering in $G$ with diameter $O(b^2)$ that clusters at least $(2/3)n$ of the vertices, as desired.

Let $i \in \{0,1,\ldots, b-1\}$. It remains to describe how to compute $(V_{i+1},Q_{i+1})$ given $(V_{i},Q_{i})$ while preserving the three invariants.

Our algorithm makes sure that the following three properties are satisfied. First, let $u,v \in Q_i$ be two arbitrary nodes that are contained in the same connected component in $G_i$ and whose identifiers disagree on the $(i+1)$-th bit. Then, at least one of them is not contained in $Q_{i+1}$ or $u$ and $v$ end up in different connected components in $G_{i+1}$. Second, for every node $v \in V_{i+1}$, $d_{i+1}(Q_{i+1},v) \leq d_i(Q_i,v) + 6b$. Third, $|V_i \setminus V_{i+1}| \leq \frac{n}{3b}$, i.e., the algorithm `deletes' at most $\frac{n}{3b}$ many nodes.

Note that satisfying these three properties indeed suffices to preserve the invariants. Let $Q_i = Q^B_i \sqcup Q^R_i$ with $Q^B_i$ containing all the nodes in $Q_i$ whose $(i+1)$-th bit in their identifier is $0$. We compute $(V_{i+1},Q_{i+1})$ from $(V_{i},Q_{i})$ as follows:

\begin{enumerate}
    \item For $j \in \{0,1,\ldots,6b - 1\}$, let $K^j_i = \{u \in V_i \colon \diff_i(u) = j\}$ with $\diff_i(u) = d_i(Q_i^R,u) - d_i(Q_i^B,u)$
    \item $j^* = \arg \min_{j \in \{0,2,4,\ldots, 6b - 2\}} |K^{j}_i \cup K^{j+1}_i|$
    \item $V_{i+1} = V_i \setminus (K_i^{j^*} \cup K_i^{j^* + 1})$ 
    \item $Q_{i+1} = Q_i \setminus (\bigcup_{j=0}^{j^* + 1} K_i^j)$
\end{enumerate}

We first show that computing $(Q_{i+1},V_{i+1})$ in this way indeed satisfies the three properties stated above. Afterwards, we show that we can compute $(Q_{i+1},V_{i+1})$ given $(Q_i,V_i)$ in $O(b^2)$ \congest rounds and using the oracle $\fO^{count}$ once.
It directly follows from the pigeonhole principle that $|K^{j^*}_i \cup K^{j^*+1}_i| \leq \frac{n}{3b}$ and therefore $|V_i \setminus V_{i+1}| \leq \frac{n}{3b}$.
Hence, it remains to verify the other two properties.

\begin{claim}
Let $u,v \in Q_i$ be two arbitrary nodes that are contained in the same connected component in $G_i$ and whose identifiers disagree on the $(i+1)$-th bit. Then, either at least one of them is not contained in $Q_{i+1}$ or $u$ and $v$ end up in different connected components in $G_{i+1}$. 
\end{claim}
\begin{proof}
We assume without loss of generality that $u \in Q_i^R$ and $v \in Q_i^B$. Furthermore, assume that $u,v \in Q_{i+1}$. We need to show that $u$ and $v$ are in different connected components in $G_{i+1}$. 
To that end, consider an arbitrary $u$-$v$-path $\langle u = w_1,w_2, \ldots, v = w_k \rangle$ in $G_i$. 
From the definition of $\diff_i$ and the fact that $u \in Q_i^R$ and $v \in Q_i^B \cap Q_{i+1}$, it follows that $\diff_i(u) < 0$ and $\diff_i(v) > j^* + 1$. 
Together with the fact that $|\diff_i(w_\ell) - \diff_i(w_{\ell + 1})| \leq 2$ for every $\ell \in [k-1]$, we get that there exists an $\ell \in \{2,3,\ldots, k-1\}$ with $\diff_i(w_\ell) \in \{j^*,j^* + 1\}$. 
For this $\ell$, $w_\ell \notin V_{i+1}$ and therefore the $u$-$v$-path is not fully contained in $G_{i+1}$. Since we considered an arbitrary $u$-$v$-path, this implies that $u$ and $v$ are in different connected components in $G_{i+1}$, as desired.
\end{proof}

\begin{claim}
For every $u \in V_{i+1}$, $d_{i+1}(Q_{i+1},u) \leq d_i(Q_i,u) + 6b$.
\end{claim}
\begin{proof}

Consider any $u \in V_{i+1}$ and recall that $\diff_i(u) = d_i(Q_i^R, u) - d_i(Q_i^B, u)$. As $u \notin K_i^{j^*} \cup K_i^{j^* + 1}$, either $\diff_i(u) < j^*$ or $\diff_i(u) > j^* + 1$.
\begin{enumerate}
    \item $\diff_i(u) < j^*$: Consider a shortest path from $Q_i^R$ to $u$. 
    Note that any node $v$ on this path also satisfies $\diff_i(v) < j^*$. 
    
    Therefore, the path is fully contained in $G_{i+1}$. Moreover, $Q_i^R \subseteq Q_{i+1}$ and therefore
    \begin{align*}
    d_{i+1}(Q_{i+1}, u)
    \leq d_{i}(Q_i^R, u)
    \le d_i(Q_i, u) + \max(0, \diff_i(u))
    \le d_{i}(Q_{i}, u) + 6b,
    \end{align*}
    as needed.
    \item $\diff_i(u) > j^* + 1$: Consider a shortest path from $Q_i^B$ to $u$. Note that any node $v$ on this path also satisfies $\diff_i(v) > j^*+1$. 
    In particular, the path is fully contained in $G_{i+1}$. Also, the start vertex $v' \in Q_i^B$ of the path satisfies $\diff_i(v') > j^*+ 1$ and thus it is contained in $Q_{i+1}$. Hence,
    \begin{align*}
    d_{i+1}(Q_{i+1}, u) 
    \leq d_{i}(Q_i^B, u) 
    = d_{i}(Q_i, u) 
    \le d_{i}(Q_i, u) + 6b, 
    \end{align*}
    as needed. 

\end{enumerate}
\end{proof}

First, each node $v$ computes the two values $\min(d_i(Q_i^{\fR},v), 6(i+1)b)$ and $\min(d_i(Q_i^{\fB},v), 6(i+1)b)$.
This can be done in $O(b^2)$ \congest rounds by computing a BFS forest from both $Q_i^\fR$ and $Q_i^\fB$ up to distance $6(i+1)b$.
As $d_i(Q_i,v) \leq 6ib$, it holds for each $j \in \{0,1,\ldots,6b-1\}$ that $\diff_i(u) := d_i(Q_i^{\fR},v) - d_i(Q_i^{\fB},v)  = j$ if and only if
$\min(d_i(Q_i^{\fR},v), 6(i+1)b) - \min(d_i(Q_i^{\fB},v), 6(i+1)b) = j$.
Thus, a node can decide with no further communication whether it is contained in $K_i^j$.
Now, one can use the oracle $\fO^{count}$ to compute $j^*$.
Given $j^*$, each node can decide whether it is contained in $V_{i+1}$ and $Q_{i+1}$, as needed.

Hence, the overall algorithm runs in $O(b^3)$ \congest rounds and invokes the oracle $\fO^{count}$ $O(b)$ times.

\end{proof}

\section{Blurry Ball Growing}
\label{sec:blurry}

The blurry ball growing problem asks for the following: in its simplest variant (randomized, edge-cutting), we are given a set $S$ and a distance parameter $D$. 
The goal is to construct a superset $S^{sup}$ of $S$ with $S^{sup} \subseteq B_{G}(S,D)$ such that every edge $e$ of length $\ell(e)$ is ``cut'' by $S^{sup}$ (that is, neither contained in $S^{sup}$, nor in $V(G) \setminus S^{sup}$) with probability $O(\ell(e)/D)$. 

This section is dedicated to prove \cref{thm:blurry_growing} that generalizes \cref{thm:blurry_informal} that we restate here for convenience.

\blurryInformal*


First, in \cref{sec:blurry_intuition}, we sketch a proof for the randomized edge-cutting variant of the problem. The main result, \cref{thm:blurry_growing}, is proven in \cref{sec:blurry_general}. Finally, in \cref{sec:blurry_corollaries} we derive simple corollaries of \cref{thm:blurry_growing} used later in the paper. 

\subsection{Intuition and Proof Sketch}
\label{sec:blurry_intuition}

We will sketch a proof of the randomized variant of \cref{thm:blurry_informal} (change the guarantee on the sum of the lengths of edges cut to the individual guarantee that each edge $e$ is cut with probability $O(\ell(e) /D)$). 
First, note that it is easy to solve the blurry ball growing problem using an exact distance oracle: 
one can simply pick a number $\overline{D} \in [0, D)$ uniformly at random and define $S^{sup} := \{u : d(S, u) \le \overline{D} \}$. 
From now on, let $e = \{u,v\}$ be an arbitrary edge with $d_u \leq d_v$ for $d_u := d(S, u)$ and $d_v := d(S,v)$. 
Choosing $\overline{D}$ as above, we indeed have 
\[
\P(\text{$e$ is cut by $S^{sup}$})
= \P(\overline{D} \in [d_u, d_v))
= \frac{| [0, D) \cap [d_u, d_v) |}{D}
\le \ell(e)/D,
\]
as needed. 

What happens if we only have access to a $(1+\eps)$-approximate distance oracle, i.e., if we define $S^{sup} = \{u \colon \tilde{d}(S,u) \leq \overline{D}\}$ with $\tilde{d}$ being $(1+\eps)$-approximate?
The calculation above would only give
\[
\P(\text{$e$ is cut by $S^{sup}$}) = \P(\overline{D} \in [\tilde{d}(S,u), \tilde{d}(S,v) ))
\leq \P(\overline{D} \in [d_u, d_v + \eps D ))
\le \ell(e)/D + \eps. 
\]
This bound is only sufficient for edges of length $\Omega(\eps D)$. 

\newcommand{\Blurexact}{\textrm{ExactBlur}}

To remedy this problem, let us consider the algorithm $\Blurexact$ given below. $\Blurexact$ only performs a \emph{binary} decision in each of the $O(\log D)$ recursion levels.
This allows us later to straightforwardly generalise it to the more complicated approximate and deterministic setting. 

\begin{algorithm}
\caption{Simple Randomized Blurry Ball Growing with Exact Distances}
\label{alg:blurry_simple}
\begin{algorithmic}
\State{{\bf Procedure} $\Blurexact(S, D)$} 
\If{$D \leq 1$} 
    \State \textbf{return}  {$S$}
\Else
    \State $\Sbig = \{ u : d(S, u) \le D/2\}$
    \If{(fair coin comes up heads)}
        \State \textbf{return} {$\Blurexact(S, D/2)$}
    \Else
        \State \textbf{return} {$\Blurexact(\Sbig, D/2)$}
    \EndIf 
\EndIf
\end{algorithmic}
\end{algorithm}

If all of the edges of $G$ had length $1$ and $D$ was a power of two, the algorithm $\Blurexact$ would actually be the same as the simple uniformly sampling algorithm discussed above. 
In that case, it would correspond to sampling the value of $\overline{D}$ bit by bit, starting with the most significant bit.
However, in general the two procedures are somewhat different. Assume that $u \in S$, $u$ is the only neighbor of some $v \not\in S$ and $\ell(\{u,v\}) = 2D/3$. 
With probability $1$, $v \not\in \Blurexact(S, D)$. That is, the probability of $\{u,v\}$ being cut is $1 > \ell(\{u,v\})/D$. 

However, we will now (informally) prove that \cref{alg:blurry_simple} nevertheless satisfies $P(\text{$e$ is cut}) = O(\ell(e)/D)$. 

\begin{proof}[Proof (informal)]
Let $p$ be the probability of $e$ being cut, $p_1$ the probability of $e$ being cut provided that the coin comes up heads (we decide not to grow) and $p_2$ the probability that $e$ is cut if the coin comes up tails; we have $p = (p_1+p_2) /2$. 

Recall that we want to prove (by induction) that $p \le C\ell(e)/D$ for some $C > 0$. 
In particular, we are going to show that
\begin{align}
\label{eq:blur_simple}
p \le 10(1 + \mathbb{1}_{middle}(e)) \ell(e)/D. 
\end{align}
Here, $\mathbb{1}_{middle}$ is the indicator of whether $0 < d_u \le d_v < D-1$, i.e., $u$ is not in $S$ and $v$ is sufficiently close to $S$. 

\begin{figure}
    \centering
    \includegraphics[width = .5\textwidth]{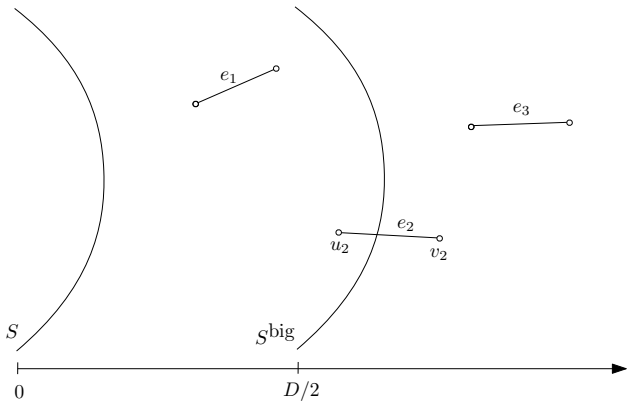}
    \caption{
    One phase of \cref{alg:blurry_simple}:
    The edge $e_1$ is fully contained in $\Sbig$. Hence, $e_1$ will not be cut if the coin comes up tails. No endpoint of the edge $e_3$ is contained in $\Sbig$. Therefore, $e_3$ is not cut if the coin comes up heads. The edge $e_2$ can potentially be cut in both cases. However, for a given edge, the situation that exactly one of the endpoints is contained in $\Sbig$ can only happen once during all recursive calls up until the point where the length of the edge is within a constant factor of $D$.}
    \label{fig:blur}
\end{figure}

To prove the bound \cref{eq:blur_simple}, first consider the case $d_u, d_v < D/2$ (cf. the edge $e_1$ in \cref{fig:blur}). 
We have, by induction, that $p_1 \le 10(1 + \mathbb{1}_{middle}(e)) \ell(e)/(D/2)$, while $p_2 = 0$. 
Here we are using the fact that if $\mathbb{1}_{middle}(e) = 1$ in the recursive call, it is also certainly equal to one now. 
We get 
\[
p = p_1/2 + p_2/2 = p_1/2 \le 10(1 + \mathbb{1}_{middle}(e)) \ell(e)/D,
\]
as needed. An analogous argument works if $d_u, d_v \geq D/2$ (cf. the edge  $e_3$ in \cref{fig:blur}). 

It remains to analyze the case $d_u < D/2 \leq d_v$ (cf. the case of $u = u_2$ and $v = v_2$ in \cref{fig:blur}). First, note that we can assume that $\ell(e) \leq D/10$ and therefore $\mathbb{1}_{middle}(e) = 1$. Moreover, in all of the subsequent recursive calls it will be the case that either $u \in S_{rec}$, or, on the other hand, $d(S_{rec}, v) \ge D_{rec}$. Thus, $\mathbb{1}_{middle}(e) = 0$ during all of the subsequent recursive calls. 

The fact that currently $\mathbb{1}_{middle}(e) = 1$ but $\mathbb{1}_{middle}(e) = 0$ in the recursive call allows us to conclude that
\[
p = \frac{p_1 + p_2}{2} \le \frac{10 \cdot 1 \cdot \ell(e)/(D/2) + 10\cdot 1 \cdot \ell(e)/(D/2)}{2} = 10(1 + 1) \ell(e)/D = 10(1 + \mathbb{1}_{middle}) \ell(e)/D,
\]
as needed.
\end{proof}

Our main result \cref{thm:blurry_growing} is a generalization of \cref{alg:blurry_simple} and the above analysis.  
First, the analysis can also be made to work with approximate distances. One difference is that we multiply $D$ by $(1-\eps)/2$ and not by $1/2$ in the recursive call, to account for the errors we make when computing the set $S^{big}$. By setting $\eps = O(1/\log(D))$, the errors accumulated over the  $O(\log D)$ iterations do not explode. 

Second, we solve a deterministic variant of the problem where the objective is to minimize the (weighted) sum of edges that are cut. 
We achieve this by derandomizing the random choices in \cref{alg:blurry_simple}. For that, it comes in handy that the algorithm samples just one bit in every iteration: in essence, the basic idea of the deterministic variant of the algorithm is that it computes in every step which choice makes the expected number of edges being cut smaller.

\subsection{General Result}
\label{sec:blurry_general}

The main result of this section is \cref{thm:blurry_growing}. It solves the general blurry ball growing problem discussed in \cref{sec:intro}.
We now define this general version of the problem. In particular, we generalize the guarantee for edges to guarantees for input balls: every node $v$ wants the ball $B(v, r(v))$ to end up fully in $S^{sup}$ or $V(G) \setminus S^{sup}$. Our algorithm outputs a set $V^{bad}$ which contains all the nodes for which this condition fails (and potentially even nodes for which the condition is satisfied).
In the randomized version, we show that each node $v$ is contained in $V^{bad}$ with probability $O(r(v)/D)$.
In the deterministic version, we show that the (weighted) number of nodes in $V^{bad}$ is sufficiently small.
We note that explicitly outputting a set of ``bad'' nodes is needed for the applications later on and we anyways have to track certain quantities (e.g. whether a node can potentially become bad) to derandomize the algorithm.

\begin{theorem}[Deterministic And Randomized Blurry Ball Growing Problem]
\label{thm:blurry_growing}
Consider the following problem on a weighted input graph $G$. The input consists of:

\begin{enumerate}
    \item A set $S \subseteq V(G)$.
    \item Each node $v \in V(G)$ has a preferred radius $r(v)$.
    \item In the deterministic version, each node $v \in V(G)$ additionally has a weight $\mu(v)$. 
    \item A distance parameter $D > 0$.
\end{enumerate}

The output is a set $S^{sup}$ with $S \subseteq S^{sup} \subseteq V(G)$ together with a set $V^{bad} \subseteq V(G)$ such that

\begin{enumerate}
    \item for every $v \in S^{sup}$, $d_{G[S^{sup}]}(S,v) \leq D$,
    \item for every $v \in V^{good} := V(G) \setminus V^{bad}$, $B(v,r(v)) \subseteq S^{sup}$ or $B(v,r(v)) \subseteq V(G) \setminus S^{sup}$,
    \item in the determinisic version, $\mu(V^{bad}) = O(\sum_{v \in V(G)} \mu(v)r(v)/D)$
    \item and in the randomized version, $Pr[v \in V^{bad}] = O(r(v)/D)$ for every $v \in V(G)$.
\end{enumerate}
There is an algorithm which returns a pair $(S^{sup}, V^{bad})$ satisfying the above properties in $O(\log(D) + 1)$ steps. The algorithm performs all oracle calls with precision parameter $\eps = \frac{1}{\log(n)}$ and distance parameter no larger than $D$.
\end{theorem}

\paragraph{Proof of \cref{thm:blurry_growing}}
Let us first give some intuition about \cref{alg:blurry}. The set $S^{big}$ corresponds to the set of the same name in \cref{alg:blurry_simple}. The trees $T_{S^{big}}, T_{V\setminus S^{big}}$ give us, informally speaking, the approximate distance from the cut $S^{big} \times (V(G) \setminus S^{big})$. The set $\Vunsafe$ contains all the nodes which can potentially be cut by the set $S^{sup}$ returned at the end. At the beginning we set $\Vunsafe = V(G)$, while in the leaf of the recursion we return $\Vbad = \Vunsafe$. We postpone the intuitive discussion about the set $V^{middle}$. However, note that the randomized version ``ignores'' the set $V^{middle}$. It is only necessary as an input to the deterministic algorithm. However, we still use the set $V^{middle}$ to analyze the randomized version. Finally, the potential $\Phi_1$ ($\Phi_2$)  in the deterministic version of \cref{alg:blurry} can be seen as a pessimistic estimator for the expected number of nodes (according to their weight) which will be labeled as bad at the end of the algorithm, i.e., which are contained in $\Vbad$, if the coin comes up heads (tails).

\begin{algorithm}
\caption{The Blurry Ball Growing Algorithm}
\label{alg:blurry}
\begin{algorithmic}
\State{{\bf Procedure} $\Blur(S, D, \Vunsafe, \Vmiddle)$} \\
Works with an arbitrary precision parameter $\eps \in [0,0.1]$. 
\If{$D \leq 1$} 
    \State \textbf{return}  {$(S,\Vunsafe)$}
\Else
    \State $T_S \gets \oDist_{\eps, D/2}(S)$
    \State $\Sbig = V(T_S)$ 
    \State $T_{\Sbig} \leftarrow \oDist_{\eps, D}(\Sbig)$
    \State $T_{V\setminus \Sbig} \leftarrow \oDist_{\eps,D}(V(G) \setminus \Sbig)$
    \State $\Vunsafe_1 = \Vunsafe \cap \left( \{v \in V(T_{\Sbig}) \colon  d_{T_{\Sbig}}(v) \leq (1+\eps)r(v)\} \cup \{v \in V(G) \colon r(v) > D/10\}\right)$
    \State $\Vunsafe_2 = \Vunsafe \cap \left( \{v \in V(T_{V \setminus \Sbig}) \colon d_{T_{V \setminus \Sbig}}(v) \leq (1+\eps)r(v) \} \cup \{v \in V(G) \colon r(v) > D/10\}\right)$
    \State $\Vpmiddle = \Vmiddle \setminus (\Vunsafe_1 \cap \Vunsafe_2 \cap \{v \in V(G) \colon r(v) \leq D/10\})$
    \State $\forall i \in \{1,2\} \colon \Phi_i = \sum_{v \in \Vunsafe_i \colon r(v) \leq D/10} (1 + \mathbb{1}_{\Vpmiddle}(v))\mu(v)r(v)$ 
    \If{(\textit{randomized} and fair coin comes up heads) or (\textit{deterministic} and $\Phi_1 \leq \Phi_2$)}
        \State \textbf{return} {$\Blur(S, (1- \eps)D/2,\Vunsafe_1, \Vpmiddle)$}
    \Else
        \State \textbf{return} {$\Blur(\Sbig, (1- \eps)D/2, \Vunsafe_2, \Vpmiddle)$}
    \EndIf 
\EndIf
\end{algorithmic}
\end{algorithm}

\begin{proof}
We invoke the deterministic/randomized recursive procedure of \cref{alg:blurry} with precision parameter $\eps = \frac{1}{\log(n)}$.
We let $(S^{sup},V^{bad}) = Blur(S,D,\{v \in V(G) \colon r(v) > 0\},V(G))$. 

We need to prove that the four properties in the theorem statement are satisfied. The proof is structured as follows. 
\cref{lem:blurry_not_too_far} implies that the first property is satisfied. \cref{lem:blur_good_property} implies that the second property is satisfied. In the randomized version, \cref{lem:blur_rand} gives that $Pr[v \in V^{bad}] \leq \frac{20r(v)}{D(1-\eps)^{\max(0,\log(2D))}}$ for every $v \in V(G)$. For $x \in [0,0.5]$, it holds that $1-x \geq e^{-2x}$. Hence, for $n$ being larger than a fixed constant, we have

\[Pr[v \in V^{bad}] \leq \frac{20r(v)}{D(1-\eps)^{\max(0,\log(2D))}} \leq \frac{20 r(v)}{D e^{-2\frac{\max(0,\log(2D))}{\log(n)}}} = O(r(v)/D).\]

In the deterministic version, \cref{lem:blur_det} gives that

\[\mu(\Vbad) \leq \frac{10}{D(1-\eps)^{\max(0,\log_2(2D))}}\sum_{v \in V(G)} \mu(v) r(v) = O \left(\sum_{v \in V(G)} \mu(v) r(v)/D \right),\]

where we again assume that $n$ is a large enough constant.
The recursion depth of the Blur-procedure is $O(\log D)$. Hence, it is easy to see that running the procedure takes $O(\log (D) + 1)$ steps and all oracle calls are performed with precision parameter $\eps = \frac{1}{\log (n)}$ and distance parameter no larger than $D$. This finishes the proof, modulo proving \cref{lem:blurry_not_too_far,lem:blur_good_property,lem:blur_rand,lem:blur_det}, which we will do next. 
\end{proof}

\begin{claim}
\label{lem:blurry_not_too_far}
Let $(S^{sup},.) = \Blur(S,D,.,.)$. For every $v \in S^{sup}$, we have $d_{G[S^{sup}]}(S,v) \leq D$.
\end{claim}
\begin{proof}
We prove the statement by induction on the recursion depth. For the base case $D \leq 1$, we have $S^{sup} = S$ and therefore the statement trivially holds.
Next, consider the case $D > 1$. We either have $(S^{sup},.) = \Blur(S,(1-\eps)D/2,.,.)$ or $(S^{sup},.) = \Blur(\Sbig,(1-\eps)D/2,.,.)$.
In the first case, the induction hypothesis gives that for any $v \in S^{sup}$ we have  $d_{G[S^{sup}]}(S,v) \leq (1-\eps)D/2 \leq D$, as desired. In the second case, the induction hypothesis states that there exists a vertex $u \in \Sbig$ with $d_{G[S^{sup}]}(u,v) \leq (1-\eps)D/2$.
However, from the way $\Sbig$ is defined, properties of $\oDist$, and the fact that $\Sbig \subseteq S^{sup}$, it follows that $d_{G[S^{sup}]}(S,u) \leq (1+\eps)D/2$. Hence, by using the triangle inequality, we obtain

\[d_{G[S^{sup}]}(S,v) \leq d_{G[S^{sup}]}(S,u) + d_{G[S^{sup}]}(u,v) \leq (1+\eps)D/2 + (1-\eps)D/2 \leq D,\]
which finishes the proof.
\end{proof}

\begin{claim}
\label{lem:blur_good_property}
Let $(S^{sup},\Vbad) = \Blur(S,D,\Vunsafe,.)$ for some $\Vunsafe \subseteq V(G)$. For every $v \in \Vunsafe$, if $B(v,r(v)) \cap S^{sup} \neq \emptyset$ and $B(v,r(v)) \setminus S^{sup} \neq \emptyset$, then $v \in \Vbad$.
\end{claim}
\begin{proof}
We prove the statement by induction on the recursion depth.
For the base case $D \leq 1$, we have $\Vbad = \Vunsafe$ and therefore the statement trivially holds.
Next, consider the case $D > 1$.
Let $v \in \Vunsafe$ and assume that $B(v,r(v)) \cap S^{sup} \neq \emptyset$ and $B(v,r(v)) \setminus S^{sup} \neq \emptyset$. 
We have to show that this implies $v \in \Vbad$.

We either have $(S^{sup},.) = \Blur(S,(1-\eps)D/2,\Vunsafe_1,.)$ or $(S^{sup},.) = \Blur(\Sbig,(1-\eps)D/2,\Vunsafe_2,.)$.

We first consider the case $(S^{sup},.) = \Blur(S,(1-\eps)D/2,\Vunsafe_1,.)$.
By assumption, $B(v,r(v)) \cap S^{sup} \neq \emptyset$. Let $u \in S^{sup}$. By \cref{lem:blurry_not_too_far}, $d_G(S,u) \leq d_{G[S^{sup}]}(S,u) \leq (1-\eps)D/2 \leq D/2$. Thus, $u \in \Sbig$ according to the second property of the distance oracle $\oDist$. Hence, $S^{sup} \subseteq \Sbig$ and therefore there exists a vertex $u \in \Sbig$ with $d_G(v,u) \leq r(v)$.
If $r(v) > D/10$, then $v \in \Vunsafe_1$ and it follows by induction that $v \in V^{bad}$. 
If $r(v) \leq D/10$, then $d_G(\Sbig,v) \leq r(v) \leq D/10 \leq D$. Hence, the second property of $\oDist$ implies that $v \in V(T_{\Sbig})$ and the first property of $\oDist$ implies that $d_{T_{\Sbig}}(v) \leq  (1+\eps) d_G(\Sbig,v) \leq (1+\eps)r(v)$ which together with $v \in \Vunsafe$ implies that $v \in \Vunsafe_1$. It follows by induction that $v \in \Vbad$. 

It remains to consider the case $(S^{sup},.) = \Blur(\Sbig,(1-\eps)D/2,\Vunsafe_2,.)$.
By assumption, $B(v,r(v)) \setminus S^{sup} \not= \emptyset$ and as $\Sbig \subseteq S^{sup}$ therefore also $B(v,r(v)) \setminus \Sbig \neq \emptyset$. Hence, $d_G(V(G) \setminus \Sbig,v) \leq r(v)$. If $r(v) > D/10$, then $v \in \Vunsafe_2$ and it follows by induction that $v \in V^{bad}$. If $r(v) \leq D/10$, then $d_G(V(G) \setminus \Sbig,v) \leq r(v) \leq D/10 \leq D$. 
Hence, the second property of $\oDist$ implies that $v \in V(T_{V\setminus \Sbig})$ and the first property of $\oDist$ implies that $d_{T_{V \setminus \Sbig}}(v) \leq (1+\eps)d_G(V(G) \setminus \Sbig,v) \leq (1+\eps)r(v)$ which together with $v \in \Vunsafe$ implies that $v \in \Vunsafe_2$. It follows by induction that $v \in \Vbad$. 
\end{proof}

\begin{definition}
\label{def:blurry_valid_input}
We refer to the tuple $(S,D,\Vmiddle)$ as a valid input if every $v \in V(G) \setminus \Vmiddle$ is either very close or very far (or both), defined as follows.
\begin{enumerate}
    \item very close: $d_G(S,v) \leq (1+\eps) r(v)$
    \item very far: $\max_{u \in B_G(v,(1+\eps)r(v))} d_G(S,u) \geq D$
\end{enumerate}
\end{definition}

\begin{claim}
\label{blurry:valid_input}
Assume $(S,D,\Vmiddle)$ is a valid input. Then, both $(S,(1-\eps)D/2,\Vpmiddle)$ and $(\Sbig,(1-\eps)D/2,\Vpmiddle)$ are valid inputs.
\end{claim}
\begin{proof}
Let $v \in V(G) \setminus \Vpmiddle$.
To show that $(S,(1-\eps)D/2,\Vpmiddle)$ is a valid input, it suffices to show that one of the following holds:

\begin{enumerate}
    \item $d_G(S,v) \leq (1+\eps)r(v)$
    \item $\max_{u \in B_G(v,(1+\eps)r(v))} d_G(S,u) \geq (1-\eps)D/2$
\end{enumerate}

To show that $(\Sbig,(1-\eps)D/2,\Vpmiddle)$ is a valid input, it suffices to show that one of the following holds:

\begin{enumerate}
    \item $d_G(\Sbig,v) \leq (1+\eps)r(v)$
    \item $\max_{u \in B_G(v,(1+\eps)r(v))} d_G(\Sbig,u) \geq (1-\eps)D/2$
\end{enumerate}

First, consider the case that $v \in V(G) \setminus \Vmiddle$. As $(S,D,\Vmiddle)$ is a valid input
, $d_G(S,v) \leq (1+\eps)r(v)$ or $\max_{u \in B_G(v,(1+\eps)r(v))} d_G(S,u) \geq D$. If $d_G(S,v) \leq (1+\eps)r(v)$, then also $d_G(\Sbig,v) \leq (1+\eps)r(v)$ as $S \subseteq \Sbig$. 
Now, assume $\max_{u \in B_G(v,(1+\eps)r(v))} d_G(S,u) \geq D \geq (1-\eps)D/2$. Hence, there exists $w \in B_G(v,(1+\eps)r(v))$ with $d_G(S,w) \geq D$ and thus

\begin{align*}
    \max_{u \in B_G(v,(1+\eps)r(v))} d_G(\Sbig,u) &\geq d_G(\Sbig,w) \\
    &= \min_{s^{big} \in \Sbig} d_G(s^{big},w) \\
    &\geq \min_{s^{big}\in \Sbig} d_G(S,w) - d_G(S,s^{big}) \\
    &= d_G(S,w) - \max_{s^{big} \in \Sbig} d_G(S,s^{big}) \\
    &\geq D - (1+\eps)D/2 \\
    &= (1-\eps)D/2,
\end{align*}

as desired.
It remains to consider the case $v \in \Vmiddle$. Hence, $v \in \Vmiddle \setminus \Vpmiddle$ and therefore $v \in \Vunsafe_1 \cap \Vunsafe_2$ and $r(v) \leq D/10$.
As $v \in \Vunsafe_1$ and $r(v) \leq D/10$, $d_{T_{\Sbig}}(v) \leq (1+\eps)r(v)$ and therefore $d_G(\Sbig,v) \leq (1+\eps)r(v)$, which already finishes the proof that $(\Sbig,(1-\eps)D/2,\Vpmiddle)$ is a valid input.
As $v \in \Vunsafe_2$ and $r(v) \leq D/10$, $d_{T_{V \setminus \Sbig}}(v) \leq (1+\eps)r(v)$. Hence, $d_G(V \setminus \Sbig,v) \leq (1+\eps)r(v)$ and therefore $B_G(v,(1+\eps)r(v)) \cap V \setminus \Sbig \neq \emptyset$. As every $w \in V \setminus \Sbig$ satisfies $d_G(S,w) \geq D/2 \geq (1-\eps)D/2$, we have $\max_{u \in B_G(v,(1+\eps)r(v))} d_G(S,u) \geq (1-\eps)D/2$, which finishes the proof that $(S,(1-\eps)D/2,\Vpmiddle)$ is a valid input.
\end{proof}

\begin{claim}
\label{lem:blurry_unsafe_implies_middle}
Assume $(S,D,\Vmiddle)$ is a valid input. For every $v \in \Vunsafe_1 \cap \Vunsafe_2$ with $r(v) \leq D/10$, we have $v \in \Vmiddle$.
\end{claim}
\begin{proof}
Let $v \in \Vunsafe_1 \cap \Vunsafe_2$ with $r(v) \leq D/10$. We have to show that $v \in \Vmiddle$. As $(S,D,\Vmiddle)$ is a valid input, it suffices by \cref{def:blurry_valid_input} to show that $d_G(S,v) > (1+\eps)r(v)$ and $\max_{u \in B_G(v,(1+\eps)r(v))} d_G(S,u) < D$.
As $v \in \Vunsafe_1$ and $r(v) \leq D/10$, $d_G(\Sbig,v) \leq d_{T_{\Sbig}}(v) \leq (1+\eps)r(v)$.
Therefore,

\begin{align*}
    \max_{u \in B_G(v,(1+\eps)r(v))} d_G(S,u) &\leq d_G(S,v) + (1+\eps)r(v) \\
    &\leq \max_{u \in \Sbig} d_G(S,u) + d_G(\Sbig,v) + (1+\eps)r(v) \\
    &\leq (1+\eps)(D/2) + 2(1+\eps)r(v) \\
    &< D, 
\end{align*}

as needed.
As $v \in \Vunsafe_2$ and $r(v) \leq D/10$, $d_G(V(G) \setminus \Sbig,v) \leq d_{T_{V \setminus \Sbig}}(v) \leq (1+\eps)r(v)$. Hence,

\[d_G(S,v) \geq d_G(S,V \setminus \Sbig) - d_G(V \setminus \Sbig,v) \geq D/2 - (1+\eps)r(v) > (1+\eps)r(v),\]

which finishes the proof.
\end{proof}

\begin{claim}[Randomized Lemma]
\label{lem:blur_rand}
Let $(.,\Vbad) = Blur_{rand}(S,D,\{v \in V(G) \colon r(v) > 0\},V(G))$ for $D > 0$. For every $v \in V(G)$, $Pr[v \in \Vbad] \leq \frac{20r(v)}{D(1-\eps)^{\max(0,\log_2(2D))}}$.
\end{claim}

\begin{proof}
Consider the following more general claim:
Let $(.,\Vbad) = Blur_{rand}(S,D,\Vunsafe, \Vmiddle)$ with $(S,D,\Vmiddle)$ being a valid input and $\{v \in \Vunsafe \colon r(v) = 0\} = \emptyset$. For a vertex $v \in V(G)$, we define $p_{v,S,D,\Vunsafe,\Vmiddle} = Pr[v \in \Vbad]$. Then,

\[
p_{v,S,D,\Vunsafe,\Vmiddle} \leq
\begin{cases}
\frac{(1+ \mathbb{1}_{\Vmiddle}(v)) 10r(v)}{D(1-\eps)^{\max(0,\log_2(2D))}} &\text{ if $v \in \Vunsafe $}\\
0 &\text{ if $v \notin \Vunsafe$}
\end{cases}
\]
We prove the more general claim by induction on the recursion depth.
The base case $D \in (0,1]$ directly follows as $\frac{(1+ \mathbb{1}_{\Vmiddle}(v))10r(v)}{D(1-\eps)^{\max(0,\log_2(2D))})} \geq 10r(v)/D \geq 1$ as long as $r(v) \geq 1$. 
Next, consider the case $D > 1$. Let $p := p_{v,S,D,\Vunsafe,\Vmiddle}$, $p_1 := p_{v,S,(1-\eps)D/2,\Vunsafe_1,\Vpmiddle}$ and $p_2 := p_{v,\Sbig,(1-\eps)D/2,\Vunsafe_2,\Vpmiddle}$.
From the algorithm definition, we have

\[p = p_1/2 + p_2/2.\]

By \cref{blurry:valid_input}, both $(S,(1-\eps)D/2,\Vpmiddle)$ and $(\Sbig,(1-\eps)D/2,\Vpmiddle)$ are valid inputs. Hence, by induction we obtain for $i \in \{1,2\}$ that

\[
p_i\leq
\begin{cases}
\frac{2 \cdot (1+ \mathbb{1}_{\Vpmiddle}(v)) 10r(v)}{D (1-\eps)(1-\eps)^{\max(0,\log_2((1-\eps)D))}}  \leq \frac{2 \cdot (1+ \mathbb{1}_{\Vpmiddle}(v)) 10r(v)}{D(1-\eps)^{\max(0,\log_2(2D))}} &\text{ if $v \in \Vunsafe_i$}\\
0 &\text{ if $v \notin \Vunsafe_i$}
\end{cases}
\]

First, if $v \notin \Vunsafe$, then $v \notin \Vunsafe_1 \cup \Vunsafe_2$ and therefore $p = 0.5p_1 + 0.5p_2 = 0.5 \cdot 0 + 0.5 \cdot 0 = 0$, as desired.

From now on, assume that $v \in \Vunsafe$. Note that we can furthermore assume that $r(v) \leq D/10$ as otherwise we claim $p \leq 1$ which trivially holds.
First, consider the case that $v \in \Vunsafe_1 \cap \Vunsafe_2$. As $r(v) \leq D/10$, \cref{lem:blurry_unsafe_implies_middle} implies that $v \in \Vmiddle$ and together with the algorithm description it follows that $v \in \Vmiddle \setminus \Vpmiddle$. 

Hence,

\[
p 
\leq 0.5p_1 + 0.5p_2 
\leq 2 \cdot 0.5 \cdot \frac{2 \cdot (1+ \mathbb{1}_{\Vpmiddle}(v)) 10r(v)}{D(1-\eps)^{\max(0,\log_2(2D))}} 
= \frac{(1+ \mathbb{1}_{\Vmiddle}(v)) 10r(v)}{D(1-\eps)^{\max(0,\log_2(2D))}} ,
\]

as desired.

It remains to consider the case $v \notin \Vunsafe_1 \cap \Vunsafe_2$.
By induction, this implies $p_1 = 0$ or $p_2 = 0$ and therefore \[p \leq 0.5p_1 + 0.5p_2 \leq 0.5 \frac{2 \cdot (1+ \mathbb{1}_{\Vpmiddle}(v)) 10r(v)}{D(1-\eps)^{\max(0,\log_2(2D))}} = \frac{(1+ \mathbb{1}_{\Vmiddle}(v)) 10r(v)}{D(1-\eps)^{\max(0,\log_2(2D))}},\]
which finishes the proof.
\end{proof}

\begin{claim}[Deterministic Lemma]
\label{lem:blur_det}
Let $(.,\Vbad) = Blur_{det}(S,D,\{v \in V(G) \colon r(v) > 0\},V(G))$ for $D > 0$. Then,
$\mu(\Vbad) \leq \frac{10}{D(1-\eps)^{\max(0,\log_2(2D))}}\sum_{v \in V(G)} \mu(v) r(v)$.
\end{claim}
\begin{proof}

Consider the following more general claim: Let $(.,\Vbad) = Blur_{det}(S,D,\Vunsafe,\Vmiddle)$ with $(S,D,\Vmiddle)$ being a valid input and $\{v \in \Vunsafe \colon r(v) = 0\} = \emptyset$. Then, 

\begin{align*}
   \mu(\Vbad) &\leq \frac{1}{(1-\eps)^{\max(0,\log_2(2D))}} \left( \sum_{v \in \Vunsafe, r(v) \leq D/10} (1 + \mathbb{1}_{\Vmiddle}(v))\frac{\mu(v)r(v)}{D}  + \sum_{v \in \Vunsafe \colon r(v) > D/10} \mu(v) \right). 
\end{align*}

We prove the more general claim by induction on the recursion depth. 
The base case $D \in(0,1]$ trivially holds as for every $v \in V(G)$, $r(v) = 0$ or $r(v) \geq 1 > D/10$.
Next, consider the case $D > 1$.
Assume that $\Phi_1 \leq \Phi_2$. 
In particular, $2\Phi_1 \leq \Phi_1 + \Phi_2$ and therefore

\begin{align*}
    &2 \sum_{v \in \Vunsafe_1 \colon r(v) \leq D/10} (1 + \mathbb{1}_{\Vpmiddle}(v))\mu(v)r(v) \\
    &\leq \sum_{i=1}^2 \sum_{v \in \Vunsafe_i \colon r(v) \leq D/10} (1 + \mathbb{1}_{\Vpmiddle}(v))\mu(v)r(v) \\
    &=  \sum_{v \in \Vunsafe_1 \cup \Vunsafe_2 \colon r(v) \leq D/10} (1 + \mathbb{1}_{\Vpmiddle}(v))\mu(v)r(v) + \sum_{v \in \Vunsafe_1 \cap \Vunsafe_2 \colon r(v) \leq D/10} (1 + \mathbb{1}_{\Vpmiddle}(v))\mu(v)r(v) \\
    &\leq \sum_{v \in \Vunsafe \colon r(v) \leq D/10} (1 + \mathbb{1}_{\Vpmiddle}(v))\mu(v)r(v) + \sum_{v \in \Vunsafe_1 \cap \Vunsafe_2 \colon r(v) \leq D/10} \mathbb{1}_{\Vmiddle \setminus \Vpmiddle}(v) \mu(v)r(v) \\
    &= \sum_{v \in \Vunsafe \colon r(v) \leq D/10} (1 + \mathbb{1}_{\Vmiddle}(v))\mu(v)r(v) ,
\end{align*}
where the second inequality follows from the following three facts: First, $\Vunsafe_1 \cup \Vunsafe_2 \subseteq \Vunsafe$. Second, as $(S,D,\Vmiddle)$ is a valid input, \cref{lem:blurry_unsafe_implies_middle} states that for every $v \in \Vunsafe_1 \cap \Vunsafe_2$ with $r(v) \leq D/10$, we have $v \in \Vmiddle$. Third, it directly follows from the algorithm description that there exists no $v \in (\Vunsafe_1 \cap \Vunsafe_2) \cap \Vpmiddle = \emptyset$ with $r(v) \leq D/10$. 

Now, let $D' := (1-\eps)D/2$. From the induction hypothesis, it follows that

\begin{align*}
    \mu(\Vbad) 
    &\leq \frac{1}{(1-\eps)^{\max(0,\log_2(2D'))}} \left( \sum_{v \in \Vunsafe_1, r(v) \leq D'/10} (1 + \mathbb{1}_{\Vpmiddle}(v))\frac{\mu(v)r(v)}{D'}  + \sum_{v \in \Vunsafe_1 \colon r(v) > D'/10} \mu(v) \right) \\
    &\leq \frac{1}{(1-\eps)^{\max(0,\log_2(2D'))}} \left( \sum_{v \in \Vunsafe_1, r(v) \leq D/10} (1 + \mathbb{1}_{\Vpmiddle}(v))\frac{\mu(v)r(v)}{D'}  + \sum_{v \in \Vunsafe_1 \colon r(v) > D/10} \mu(v) \right) \\
    &\leq \frac{1}{(1-\eps)^{\max(0,\log_2(2D))}} \left( 2 \sum_{v \in \Vunsafe, r(v) \leq D/10} (1 + \mathbb{1}_{\Vpmiddle}(v))\frac{\mu(v)r(v)}{D}  + \sum_{v \in \Vunsafe_1 \colon r(v) > D/10} \mu(v) \right) \\
    &\leq \frac{1}{(1-\eps)^{\max(0,\log_2(2D))}} \left( \sum_{v \in \Vunsafe, r(v) \leq D/10} (1 + \mathbb{1}_{\Vmiddle}(v))\frac{\mu(v)r(v)}{D}  + \sum_{v \in \Vunsafe \colon r(v) > D/10} \mu(v) \right),
\end{align*}
as desired. The case $\Phi_2 < \Phi_1$ follows in the exact same manner and is therefore omitted.

\end{proof}

\subsection{Corollaries}
\label{sec:blurry_corollaries}

In this section we show how the rather general \cref{thm:blurry_growing} implies the solution to the edge variant of the blurry growing problem from \cref{thm:blurry_informal}.
In particular, we prove here \cref{thm:blurry_edge} that solves both the randomized and deterministic version of the problem. 

\begin{definition}[Subdivided Graph]
Let $G$ be a weighted graph. The subdivided graph $G_{sub}$ of $G$ is defined as the weighted graph that one obtains from $G$ by replacing each edge $e = \{u,v\} \in E(G)$ with one new vertex $v_e$ and two new edges $\{u,v_e\}$ and $\{v_e,v\}$. Moreover, we define $\ell_{G_{sub}}(u,v_e) = 0$ and $\ell_{G_{sub}}(v_e,v) = \ell_G(u,v)$ where we assume that $u$ has a smaller ID than $v$.
\end{definition}
\begin{lemma}[Simulation of the Subdivided Graph]
\label{lem:simulation_subdivided}
Assume that some problem $\mathcal{P}$ defined on a weighted graph $H$ can be solved in $T$ steps and with performing all oracle calls with precision parameter $\eps$ and distance parameter no larger than $D$ for arbitrary $T > 0$, $\eps$ and $D$.
Now, assume that the weighted input (and communication) graph is $G$. Then, we can solve the problem $\mathcal{P}$ on the weighted graph $G_{sub}$ in $O(T)$ steps performing all oracle calls with precision parameter $\eps$ and distance parameter no larger than $D$. 
\end{lemma}
\begin{proof}
For each new vertex $w$ that subdivides the edge $\{u,v\} \in E(G)$, the node $w$ is simulated by the node $u$ (where we assume that $u$ has a smaller ID than $v$). That is, if $w$ wants to send a message to $v$, then $u$ sends that message to $v$. Similarly, if $v$ wants to send a message to $w$, then $v$ sends the message to $u$ instead.
The node $u$ also performs all the local computation that $w$ would do. 
As $w$ has exactly two neighbors, our computational model only allows it to perform $\tilde{O}(1)$ (P)RAM operations in each step and therefore each node in the original graph only needs to do additional work proportional to its degree and which can be performed with depth $\tilde{O}(1)$. Hence, each node can efficiently simulate all the new nodes that it has to simulate. It remains to discuss how to simulate the oracles in the graph $G_{sub}$ with the oracles for the original graph $G$.
The global aggregation oracle in $G_{sub}$ can be simulated in $G$ as follows: first, each node computes the sum of the values of all the nodes it simulates (including its own value), which can efficiently be done in PRAM. Then, one can use the aggregation oracle in $G$ to sum up all those sums, which is equal to the total sum of all the node values in $G_{sub}$.
For the forest aggregation oracle, let $F_{sub}$ be a rooted forest in $G_{sub}$. Now, let $F$ be the rooted forest with $V(F) = V(F_{sub}) \cap V(G)$ and which contains each edge $\{u,v\} \in E(G)$ if both $\{u,w\}$ and $\{w,v\}$ are contained in the forest $F_{sub}$. Moreover, the set of roots of $F$ is given by all the roots in $F_{sub}$ that are vertices in $G$ together with all vertices in $G$ whose parent in $F_{sub}$ is a root. Note that for every node $v \in V(F)$, $d_F(v) \leq d_{F_{sub}}(v) \leq D$. Moreover, it is easy to see that the aggregation on $F_{sub}$ can be performed in $O(1)$ steps using only aggregations on $F$ as an oracle. 
Finally, one can also simulate the distance oracle $\fO^{Dist}_{\eps,D}$ in $G_{sub}$ in $O(1)$ steps and only performing one oracle call to $\fO^{Dist}_{\eps,D}$ in $G$.
\end{proof}

It is easy to deduce the specific version for edges from the above \cref{thm:blurry_growing} and the fact that we can efficiently simulate an algorithm on the subdivided graph. 

\begin{corollary}
\label{thm:blurry_edge}
Consider the following problem on a weighted input graph $G$. The input consists of the following.

\begin{enumerate}
    \item A set $S \subseteq V(G)$.
    \item In the deterministic version, each edge $e \in E(G)$ has a weight $\mu(e)$.
    \item A distance parameter $D > 0$.
\end{enumerate}

The output is a set $S^{sup}$ with $S \subseteq S^{sup} \subseteq V(G)$. Let $E^{bad}$ denote the set consisting of those edges having exactly one endpoint in $S^{sup}$, then $S^{sup}$ satisfies

\begin{enumerate}
    \item for every $v \in S^{sup}, d_{G[S^{sup}]}(S,v) \leq D$,
    \item in the deterministic version, $\mu(E^{bad}) = O(\sum_{e \in E(G)} \mu(e) \ell(e)/D)$
    \item and in the randomized version, $Pr[e \in E^{bad}] = O(\ell(e)/D)$ for every $e \in E(G)$.
\end{enumerate}
There is an algorithm that solves the problem above in $O(\log (D) + 1)$ steps, performing all oracle calls with precision parameter $\eps = \frac{1}{\log(n)}$ and distance parameter no larger than $D$.
\end{corollary}
\begin{proof}
Run the algorithm of \cref{thm:blurry_growing} on the subdivided graph with input $r(v_e) = \ell(e)$ and $\mu(v_e) = \mu(e)$ for every edge $e \in E$. For all other vertices in the subdivided graph set $r(v) = 0$ and $\mu(v) = 0$.
\end{proof}

\section{A General Clustering Result}
\label{sec:main_clustering}

In this section we prove our main clustering result \cref{thm:steroids}. 
For the application to low stretch spanning trees in \cref{thm:low_stretch_spanning_tree}, it is important that the result works by essentially only having access to an approximate distance oracles and that it works even if we start with an input set of \emph{terminals} and require that each final cluster contains at least one such terminal. 
We note that the algorithm of our main result \cref{thm:steroids} uses the blurry ball growing algorithm of \cref{sec:blurry} as a subroutine.
Since our main result \cref{thm:steroids} is rather general, we start by sketching a simpler version of our result in \cref{sec:main_clustering_intro}. Afterwards, we prove \cref{thm:steroids} in \cref{sec:main_clustering_main}. Finally, we derive useful corollaries of \cref{thm:steroids} in \cref{sec:main_clustering_corollaries}. 

\subsection{Intuition and Proof Sketch}
\label{sec:main_clustering_intro}

We now sketch the proof of \cref{thm:steroids_simple} -- a corollary of the general result \cref{thm:steroids}. \cref{thm:steroids_simple} was discussed in \cref{sec:intro}, we restate it here for convenience. 

\steroidsSimple*

Our approach to prove \cref{thm:steroids_simple} is somewhat similar to the one taken in \cref{sec:unweighted_strong_clustering} to derive our strong-diameter clustering result. 
As in the proof of \cref{thm:strong_diameter_clustering}, we solve it by repeatedly solving the following problem $O(\log n)$ times: in the $i$-th iteration, we split the still active terminals $Q_i \subseteq Q$ based on the $i$-th bit in their identifier into a set of blue terminals $Q_i^\fB$ and a set of red terminals $Q_i^\fR$. 
Then, we solve the following ``separation'' problem: 

In that problem, we are given a set $Q_i$ that is $R_i$-ruling in $G_i$ for $R_i = (1+\eps/\poly\log n)^i R$. 
We want to select a subset $Q_{i+1} = Q_{i+1}^\fR \sqcup Q_{i+1}^\fB$ of $Q_i = Q_i^\fR \sqcup Q_i^\fB$ and cut a small fraction of edges in $G_i$ to get a new graph $G_{i+1}$ such that the following three properties hold:

\begin{enumerate}
    \item Separation Property: The sets $Q_{i+1}^\fR$ and $Q_{i+1}^\fB$ are disconnected in $G_{i+1}$. 
    \item Ruling Property: The set $Q_{i+1}$ is $(1+\eps/\poly\log n)R_i$-ruling in $G_{i+1}$. 
    \item Cut Property: The number of edges cut is at most $\tO\left( \frac{1}{\eps R}   \right)\cdot \sum_{e \in E(G)} \ell(e)$.
\end{enumerate}

If we can solve this partial problem, then we simply repeat it $O(\log n)$ times going bit by bit and obtain an algorithm proving \cref{thm:steroids_simple} (cf. the reduction of \cref{thm:steroids} to \cref{lem:clustering_RBsplit} in the general proof in \cref{sec:main_clustering_main}). 

Our solution that achieves the three properties above is more complicated than the proof of \cref{thm:strong_diameter_clustering} in \cref{sec:unweighted_strong_clustering}: we need to be more careful because we only have access to approximate distances and also because we want to cluster all of the vertices. The different steps of our algorithm are illustrated in \cref{fig:main_clustering}.

\begin{figure}
    \centering
    \includegraphics[width = \textwidth]{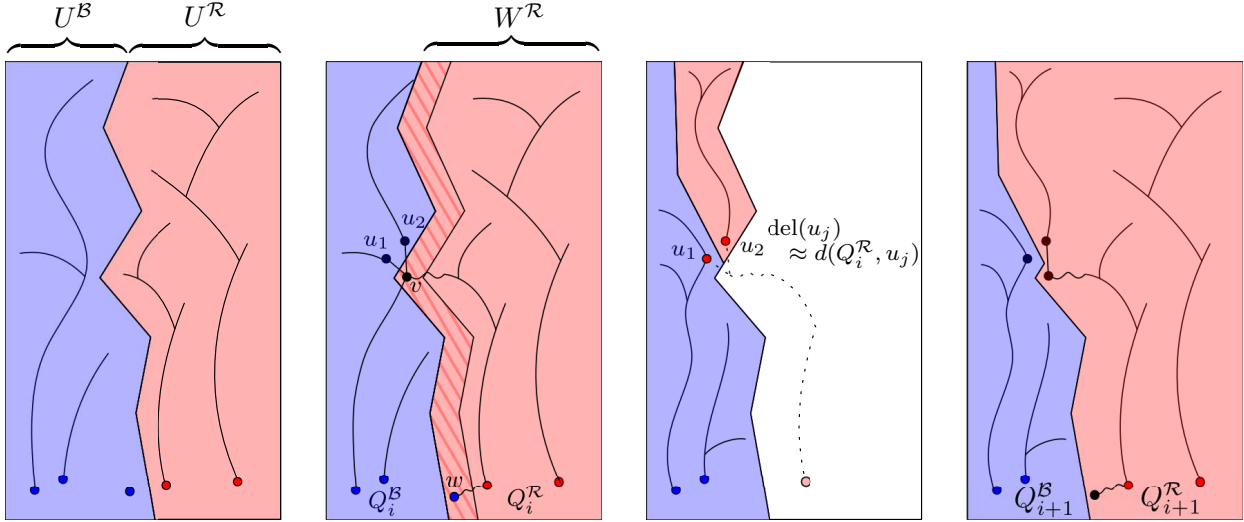}
    \caption{The four pictures illustrate the $i$-th phase of the algorithm. }
    \label{fig:main_clustering}
\end{figure}

We start by computing an approximate shortest path forest $F$ with the active terminals $Q_i$ being the set of roots.
We define the set of blue nodes $U^\fB$ and red nodes $U^\fR$ as the set of the nodes such that the root of their tree in $F$ is in $Q_i^\fB$ and $Q_i^\fR$, respectively. This step is illustrated in the first picture of \cref{fig:main_clustering}.

What happens if we cut all edges between $U^\fB$ and $U^\fR$ and define $Q_{i+1}^\fR = Q_i^\fR$ and $Q_{i+1}^\fB = Q_i^\fB$? The separation and the ruling property will be clearly satisfied. However, we do not have any guarantees on the number of edges cut, hence the cut property is not necessarily satisfied. 

To remedy this problem, we use the tool of blurry ball growing developed in \cref{sec:blurry}. 
In particular, we choose one of the two colors (which one we discuss later). Let us name it $\fA$ and define $W^\fA$ as the set of nodes returned by the blurry ball growing procedure from \cref{thm:blurry_edge} starting from $U^\fA$ with distance parameter $D = \eps/\poly\log n \,\cdot R$. This step is illustrated in the second picture of \cref{fig:main_clustering}.

Let us note that by the properties of \cref{thm:blurry_edge}, if we now delete all the edges between $W^\fA$ and $U^{{\overline{\fA}}} \setminus W^\fA$ (here, ${\overline{\fA}} \in \{\fR, \fB \} \setminus \fA$), the cut property would be satisfied. For $Q_{i+1}^\fA = Q_i^\fA$ and $Q_{i+1}^{\overline{\fA}} = Q_i^{\overline{\fA}} \setminus W^\fA$, we also get the separation property. The problem is the ruling property: for example, in the second picture of \cref{fig:main_clustering} the path from the node $u_1$ to its root in $F$ contains a node $v$ in  $W^\fA$ and is therefore destroyed. Hence, the set $Q_{i+1}^{\overline{\fA}}$ may fail to be $(1+\eps/\poly\log n)R$-ruling in the respective component of $G_{i+1}$. 

The final trick that we need is to realize that, although we are not done yet, we still made some progress, which allows us to set up a recursion: if we choose $\fA$ to be the color class such that $|U^\fA| \ge |V(G)|/2$, for at least half of the nodes, in particular those in $W^\fA$, we can now safely say that they will belong to a connected component containing a node from $Q^\fA$ in the final partition. 
For the nodes in $U^{\overline{\fA}} \setminus W^\fA$ we do not know yet, however, we can simply solve the problem there recursively. This step is illustrated in the third picture of \cref{fig:main_clustering}.

This recursion works as follows. We will recurse on the graph $G_{rec} = G[U^{\overline{\fA}} \setminus W^\fA]$. The set of terminals $Q_{rec}^{\overline{\fA}}$ in the recursive problem is simply $Q^{\overline{\fA}}_i \cap V(G_{rec})$. 
The set of terminals $Q_{rec}^\fA$ contains every node $u \in V(G_{rec})$ such that the parent of $u$ in the forest $F$ is contained in $W^\fA$. 
To reflect the fact that $u$ is not a terminal in the original problem, we introduce the notion of delays -- see \cref{sec:preliminaries} for their definition. The delay of $u$ is (roughly) set to the computed approximate distance to $Q^\fA$. This means that in the recursive call, the shortest path forest starting from $Q_{rec}^\fA$ behaves as if it started from $Q^\fA$, modulo small errors in distances of order $\eps R/\poly\log n$ that we inflicted by using approximate distances and by using blurry ball growing to obtain the set $W^\fA$. The recursive problem together with a solution is depicted in the third picture of \cref{fig:main_clustering}.

When we return from the recursion, the nodes of $G_{rec}$ are split into those belonging to terminals in $Q_{rec}^{\prime \fR}$ and $Q_{rec}^{\prime \fB}$. Note that in the example given in \cref{fig:main_clustering} we have $Q_{rec}^{\prime \fR} = \{u_2\}$.
We define $Q^{\prime \overline{\fA}} = Q_{rec}^{\prime \overline{\fA}}$ and $Q^{\prime \fA} = Q^\fA$. We also mark the nodes of $W^\fA$ as belonging to terminals in $\fA$. 
The final solution for the ``separation" problem is depicted in the fourth picture of \cref{fig:main_clustering}.

To finish the $i$-th iteration, we define $Q_{i+1}^\fR = Q^{\prime \fR}$ and $Q_{i+1}^\fB = Q^{\prime \fB}$. We cut all edges between the nodes belonging to terminals in $Q^{\prime \fR}$ and $Q^{\prime \fB}$. 
This definition of $Q_{i+1}$ preserves the separation property. 

Moreover, one can check that in every recursive step we distort the distances multiplicatively by $1+\eps/\poly\log(n)$ and additively by $\eps/\polylog(n) \cdot R$. This implies that the ruling property is satisfied.  
Similarly, the cut property is satisfied since each recursive step contributes only $\tO\left( \frac{1}{\eps R}   \right)\cdot \sum_{e \in E(G)} \ell(e)$ to the final number of edges cut.

\subsection{Main Proof}
\label{sec:main_clustering_main}

We are now ready to state and prove our main clustering result. As before in \cref{sec:blurry}, we first consider a version where the goal is to minimize the number of vertices $v$ whose ball of radius $r(v)$ is not fully contained in one of the clusters. Later, the edge cutting version follows as a simple corollary. 

Moreover, as written above, the theorem allows each terminal to be assigned a delay. Allowing these delays helps us with solving the clustering problem recursively and the delays are also convenient when we apply our clustering result to efficiently compute low-stretch spanning trees.

The final algorithm invokes the blurry ball growing procedure a total of $O(\log^2 n)$ times, each time with parameter $D$.
That's the reason why the ruling guarantee at the end contains an additive $O(\log^2 n)D$ term.

In the simplified version presented above, we set $D$ equal to $\eps R/\poly(\log n)$.

\begin{restatable}{theorem}{steroids}
\label{thm:steroids}
Consider the following problem on a weighted input graph $G$.
The input consists of the following.
\begin{enumerate}
    \item A weighted subgraph $H \subseteq G$.
    \item Each node $v \in V(H)$ has a preferred radius $r(v)$. 
    \item In the deterministic version, each node $v \in V(H)$ additionally has a weight $\mu(v)$.
    \item There is a set of center nodes $Q \subseteq V(H)$, with each center node $q \in Q$ having a delay $\del(q) \geq 0$.
    \item There is a parameter $R$ such that for every $v \in V(H)$ we have $d_{H,\del}(Q,v) \leq R$.
    \item There are two global variables $D > 0$ and $\eps \in \left[ 0,\frac{1}{\log^2(n)}\right]$. 
\end{enumerate}

The output consists of a partition $\fC$ of $H$, a set $Q' \subseteq Q$ and two sets $\Vgood \sqcup \Vbad = V(H)$ such that
\begin{enumerate}
        \item each cluster $C \in \fC$ contains exactly one node $q'_C$ in $Q'$,
        \item  for each $C \in \fC$ and $v \in C$, we have $d_{H[C],\del}(q'_C,v) \le  (1+\eps)^{O(\log^2 n)} d_{H,\del}(Q,v) + O(\log^2 n) D $,
        \item for every $v \in V^{good}$, $B_H(v,r(v)) \subseteq C$ for some $C \in \fC$ and
        \item in the deterministic version, $\mu(\Vbad) = O(\log(n) \cdot \sum_{v \in V(H)} \mu(v) r(v) / D)$ \\
        and in the randomized version, for every $v \in V(H) \colon  \P[v \in \Vbad] = O(\log(n)r(v)/D)$.
\end{enumerate}

 There is an algorithm that solves the problem above in $O(\log^3 n)$ steps, performing all oracle calls with precision parameter $\eps$ and distance parameter no larger than $(1+\eps)^{O(\log^2 n)} R + O(\log^2 n)D$.

\end{restatable}

\begin{proof}

Recall that $b$ denotes the number of bits in the node-IDs of the input graph.
The algorithm computes a sequence of weighted graphs $H = H_0 \supseteq H_1 \supseteq \ldots \supseteq H_b$ with $V(H_i) = V(H)$ for $i \in \{0,1,\ldots,b\}$, a sequence of centers $Q = Q_0 \supseteq Q_1 \supseteq Q_2 \supseteq \ldots \supseteq Q_b$ and a sequence of good nodes $V = V^{good}_0 \supseteq V^{good}_1 \supseteq V^{good}_2 \supseteq \ldots \supseteq V^{good}_b$. Moreover, $V^{good} := V^{good}_b$, $V^{bad} := V(G) \setminus V^{good}$ and $V^{bad}_i := V \setminus V^{good}_i$.
The connected components of $H_b$ will be the clusters of the output partition $\fC$ and there will be exactly one node in $Q' := Q_b$ contained in each cluster of $\fC$.

The following invariants will be satisfied for every $i \in \{0,1,\ldots,b\}$:

\begin{enumerate}
    \item Separation Invariant: Let $u,v \in Q_i$ be two nodes contained in the same connected component of $H_i$. Then, the first $i$ bits of the IDs of $u$ and $v$ coincide. 
    \item Ruling Invariant: For every $v \in V(H)$, we have
    \[
    d_{H_i,\del}(Q_i,v) \le   (1+\eps)^{i \cdot O(\log n)} d_{H,\del}(Q,v) + \left(\sum_{j=1}^{i \cdot O(\log n)}(1+\eps)^j\right)3D.
    \]
    \item Good Invariant: For every $v \in V^{good}_i$, $B_{H_i}(v,r(v)) = B_{H}(v,r(v))$.
    \item Bad Invariant (Deterministic): $\mu(V^{bad}_i) = i \cdot O(\sum_{v \in V(H)} \mu(v) r(v) /D)$. \\
    \phantom{aaaaaaaaaaaaa}(Randomized): For every $v \in V(H) \colon  \P[v \in V^{bad}_i] = i \cdot O(r(v)/D)$.
\end{enumerate}

It is easy to verify that setting $H_0 = H$, $Q_0 = Q$ and $V^{good}_0 = V(H)$ results in all of the invariants being satisfied for $i = 0$. For $i = b$, the separation invariant implies that each cluster of $\fC$ (that is, a connected component of $H_b$) contains at most one node in $Q' = Q_b$. Together with the ruling invariant, this implies that every cluster $C \in \fC$ contains exactly one node $q'_C$ in $Q'$ such that for each node $v \in C$, it holds that

\begin{align*}
    d_{H[C],\del}(q'_C,v) &= d_{H_b,\del}(Q_b,v) \\
    &\leq (1+\eps)^{b \cdot O(\log n)}d_{H,\del}(Q,v) +  \left(\sum_{j=1}^{b \cdot O(\log n)}(1+\eps)^j\right)3D \\
    &= (1+\eps)^{O(\log^2 n)}d_{H,\del}(Q,v) + \left(\sum_{j=1}^{O(\log^2 n)} O(1)\right)3D \\
    &= (1+\eps)^{O(\log^2 n)}d_{H,\del}(Q,v) + O(\log^2 n)D,
\end{align*}

where we used that $\eps \leq \frac{1}{\log^2(n)}$.

Furthermore, it follows from the good invariant that for every $v \in V^{good} = V^{good}_b$, $B_H(v,r(v)) = B_{H_b}(v,r(v))$. In particular, all vertices in $B_H(v,r(v))$ are contained in the same connected component in $H_b$ and therefore $B_H(v,r(v)) \subseteq C$ for some cluster $C \in \fC$. For the deterministic version, the bad invariant implies

\[\mu(V^{bad}) = \mu(V^{bad}_b) = b \cdot O(\sum_{v \in V(H)} \mu(v)r(v)/D) 
= O\left( \log(n) \sum_{v \in V(H)} \mu(v)r(v)/D\right).
\]

For the randomized version, the bad invariant implies that for every $v \in V(H)$,

\[ \P[v \in V^{bad}] =  \P[v \in v^{bad}_b] = b \cdot O(r(v)/D) = O(\log(n) r(v)/D).\]

Hence, we output a solution satisfying all the criteria.

Let $i \in \{0,1,\ldots,b-1\}$. It remains to describe how to compute $(H_{i+1},Q_{i+1},V^{good}_{i+1})$ given 
$(H_i,Q_i,V^{good}_i)$ while preserving the invariants.

In each phase, we split $Q_i = Q_i^\fR \sqcup Q_i^\fB$ according to the $i$-th bit in the unique identifier of each node. 
Then, we apply \cref{lem:clustering_RBsplit} with $H_\Lref{lem:clustering_RBsplit} = H_i, Q_\Lref{lem:clustering_RBsplit}^\fR = Q_i^\fR, Q_\Lref{lem:clustering_RBsplit}^\fB = Q_i^\fB$ and $ R_\Lref{lem:clustering_RBsplit} = (1+\eps)^{O(\log^2 n)} R + O(\log^2 n)D$. In the randomized version we set the recursion depth parameter $i_\Lref{lem:clustering_RBsplit} = \lceil \log_2(D)\rceil$.

Note that the input is valid as for every $v \in V(H_\Lref{lem:clustering_RBsplit})$ we have

\[d_{H_\Lref{lem:clustering_RBsplit},\del}(Q_\Lref{lem:clustering_RBsplit},v) = d_{H_i,\del}(Q_i,v) \leq (1+\eps)^{i \cdot O(\log n)} d_{H,\del}(Q,v) + \left(\sum_{j=1}^{i \cdot O(\log n)}(1+\eps)^j\right)3D \leq R_\Lref{lem:clustering_RBsplit},\]

where the first inequality follows from the ruling invariant.
Finally, we set
$H_{i+1} = H'_\Lref{lem:clustering_RBsplit}, Q_{i+1} = Q'^\fR_\Lref{lem:clustering_RBsplit} \sqcup Q'^\fB_\Lref{lem:clustering_RBsplit}$ and $V^{good}_{i+1} = V^{good}_i \cap (V^{good})_\Lref{lem:clustering_RBsplit}$.

\begin{claim}
Computing $(H_{i+1},Q_{i+1},V_{good,i+1})$ from $(H_{i},Q_{i},V_{good,i})$ as written above preserves the four invariants.
\end{claim}
\begin{proof}
We start with the separation invariant. Let $u,v \in Q_{i+1}$. Assume that $u$ and $v$ are in the same connected component of $H_{i+1}$. In particular, this implies that $u$ and $v$ are also in the same connected component of $H_i$ and therefore the separation invariant for $i$ implies that the first $i$ bits of the IDs of $u$ and $v$ coincide. Moreover, the separation property of \cref{lem:clustering_RBsplit} together with our assumption that $u$ and $v$ are in the same connected component of $H_{i+1} := H'_\Lref{lem:clustering_RBsplit}$ implies that the first $i+1$ bits of the IDs of $u$ and $v$ coincide, as desired.
Next, we check that the ruling invariant is satisfied.
The ruling property of \cref{lem:clustering_RBsplit} implies that for every $v \in V(H_{i+1})$ we have

\begin{align*}
    d_{H_{i+1},\del}(Q_{i+1},v) &= d_{H'_\Lref{lem:clustering_RBsplit},\del}(Q'^\fR_\Lref{lem:clustering_RBsplit} \cup Q'^\fB_\Lref{lem:clustering_RBsplit},v) \\
    &\leq  (1+\eps)^{2(i_\Lref{lem:clustering_RBsplit}+1)} d_{H_\Lref{lem:clustering_RBsplit},\del}(Q^\fR_\Lref{lem:clustering_RBsplit} \cup Q^\fB_\Lref{lem:clustering_RBsplit},v)  + \left( \sum_{j=1}^{i_\Lref{lem:clustering_RBsplit}} (1+\eps)^{2j} \right) 3D \\
    &\leq (1+\eps)^{O(\log n)}d_{H_i,\del}(Q_i,v) + \left( \sum_{j=1}^{O(\log n)} (1+\eps)^{j} \right) 3D \\
        &\leq (1+\eps)^{O(\log n)} \left((1 + \eps)^{i \cdot O(\log n)} d_{H,\del}(Q,v) + \left(\sum_{j=1}^{i \cdot O(\log n)}(1+\eps)^j\right)3D\right) \\
        &+ \left( \sum_{j=1}^{O(\log n)} (1+\eps)^{j} \right) 3D \\
    &= (1 + \eps)^{(i+1) \cdot O(\log n)} d_{H,\del}(Q,v) + \left(\sum_{j=1}^{(i+1) \cdot O(\log n)}(1+\eps)^j\right)3D.
\end{align*}

Hence, the ruling property is preserved.
To check that the good invariant is preserved, consider an arbitrary node $v \in V^{good}_{i+1}$. We have to show that $B_H(v,r(v))$ is fully contained in one of the connected components of $H_{i+1}$.

First, $v \in V^{good}_{i+1}$ directly implies $v \in V^{good}_i$ and therefore the good invariant implies that $B_{H_i}(v,r(v)) = B_{H}(v,r(v))$. Second, $v \in (V^{good})_\Lref{lem:clustering_RBsplit}$ together with the good property of \cref{lem:clustering_RBsplit} implies that
\[B_{H_{i+1}}(v,r(v)) = B_{H'_\Lref{lem:clustering_RBsplit}}(v,r(v)) = B_{H_\Lref{lem:clustering_RBsplit}}(v,r(v)) = B_{H_i}(v,r(v)) .\]
Hence, $B_{H_{i+1}}(v,r(v)) = B_{H}(v,r(v))$, as needed.

It remains to check the bad property. For the deterministic version, we have

\begin{align*}
   \mu(V^{bad}_{i+1}) &\leq \mu(V^{bad}_i) + \mu(V^{bad}_{\Lref{lem:clustering_RBsplit}}) \\
   &\leq i \cdot O \left( \sum_{v \in V(H)} \mu(v) r(v) /D \right) + O \left(\sum_{v \in V(H)} \mu(v)r(v)/D \right) \\
   &\leq (i+1)\cdot O \left( \sum_{v \in V(H)} \mu(v) r(v) /D \right).
\end{align*}

For the randomized version, we have for every $v \in V(H)$

\begin{align*}
     \P[v \in V^{bad}_{i+1}] &\leq  \P[v \in V^{bad}_i] +  \P[v \in V^{bad}_{\Lref{lem:clustering_RBsplit}}] \\
    &\leq i \cdot O(r(v)/D) + \left(\frac{1}{2^{\lceil \log(D) \rceil}} + \frac{1}{D}\right) O(r(v)) \\
    &\leq (i+1) \cdot O(r(v)/D),
\end{align*}

as needed.

\end{proof}
Each of the $O(\log n)$ invocations of the algorithm of \cref{lem:clustering_RBsplit} takes $O(\log^2 n)$ steps, and all oracle calls are performed with precision parameter $\eps$ and distance parameter no larger than $(1+\eps)^{O(\log n)} R_{\Lref{lem:clustering_RBsplit}} + \left( \sum_{j=1}^{O(\log n)} (1+\eps)^{2j} \right) 3D = (1+\eps)^{O(\log^2 n)} R + O(\log^2 n)D$. This finishes the proof of \cref{thm:steroids}.
\end{proof}

\begin{lemma}
\label{lem:clustering_RBsplit}
Consider the following problem on a weighted input graph $G$. The input consists of the following.

\begin{enumerate}
    \item A weighted subgraph $H \subseteq G$.
    \item Each node $v \in V(H)$ has a preferred radius $r(v)$.
    \item In the deterministic version, each node $v \in V(H)$ additionally has a weight $\mu(v)$.
    \item There is a recursion depth parameter $i \in \mathbb{N}_0$. In the randomized version $i$ is part of the input. In the deterministic version we define $i =  1 + \lfloor \log(\sum_{v \in V(H)} \mu(v)r(v))\rfloor$ if $\sum_{v \in V(H)} \mu(v)r(v) > 0$ and $i = 0$ otherwise.
    \item There are sets $Q = Q^\fR \sqcup Q^\fB \subseteq V(H)$ with each center $q \in Q$ having a delay $\del(q) \geq 0$. 
    \item There is a parameter $R$ such that for every $v \in V(H)$ we have $d_{H,\del}(Q,v) \leq R$.
    \item There are two global variables $D > 0$ and $\eps \in \left[ 0,\frac{1}{\log(n)} \right]$.
\end{enumerate}

The output consists of two sets $Q'^\fR \subseteq Q^\fR$ and $Q'^\fB \subseteq Q^\fB$, a weighted graph $H' \subseteq H$ with $V(H') = V(H)$ together with a partition $V(H) = V^{good} \sqcup V^{bad}$ such that 

\begin{enumerate}
    \item Separation Property: For each connected component $C$ of $H'$, $Q'^\fR \cap C = \emptyset$ or $Q'^\fB \cap C = \emptyset$.

    \item Ruling Property: For every $v \in V(H')$, $d_{H',\del}(Q'^\fR \cup Q'^\fB,v) \le   (1+\eps)^{2(i+1)} d_{H,\del}(Q,v) + \left( \sum_{j=1}^i (1+\eps)^{2j} \right) 3D$. 
    \item Good Property: For every $v \in V^{good}$, $B_{H'}(v,r(v)) = B_H(v,r(v))$.
    \item Bad Property, Deterministic Version: $\mu(V^{bad}) \leq \left(1 - \frac{1}{2^i}\right) O(\sum_{v \in V(H)} \mu(v)r(v)/D)$. \\
    \phantom{aaaaaaaaaaaa} Randomized Version: For every $v \in V(H)$, $ \P[v \in V^{bad}] = (\frac{1}{2^i} + \frac{1-2^{-i}}{D})O(r(v))$. 
\end{enumerate}
 
 There is an algorithm that solves the problem above in $O((i+1) (\log (D) + 1))$ steps, performing all oracle calls with precision parameter $\eps$ and distance parameter no larger than $(1+\eps)^{2i} R + \left( \sum_{j=1}^i (1+\eps)^{2j} \right) 3D$.

 \end{lemma}

\begin{proof}

We will first consider the special, base case with $i = 0$. 
Then, we analyse the general case. 

\paragraph{Base Case: }
We start with the base case $i = 0$.

Let $F$ be the weighted and rooted forest returned by $\oDist_{\eps,R}(H,Q,\del)$. Note that we are allowed to perform this oracle call as the distance parameter $R$ satisfies $R \leq (1+\eps)^0 R + \left( \sum_{j=1}^0 (1+\eps)^{2j} \right) 3D$. Moreover, as for every $v \in V(H)$, $d_{H,\del}(Q,v) \leq R$, the second property of the distance oracle ensures that $V(F) = V(H)$.

For $\fA \in \{\fR,\fB\}$, we define 

\[U^\fA = \{v \in V(F) = V(H) \colon \root_F(v) \in Q^\fA\}.\]
 
Note that $V(H) = U^\fR \sqcup U^\fB$ as for every $v \in V(F)$, $\root_F(v) \in Q = Q^\fR \sqcup Q^\fB$. 
The output now looks as follows:
We set $Q'^\fR = Q^\fR \cap U^\fR$ and $Q'^\fB = Q^\fB \cap U^\fB$. We obtain the graph $H'$ from $H$ by deleting every edge with one endpoint in $U^\fB$ and the other endpoint in $U^\fR$. We set $V^{bad} = \{v \in V(H) \colon r(v) > 0\}$ and $V^{good} = V(H) \setminus V^{bad}$.

We now verify that all the four properties from the theorem statement are satisfied.

\paragraph{Separation property} Let $C$ be a connected component of $H'$. 
As we obtained $H'$ from $H$ by deleting every edge with one endpoint in $U^\fB$ and one endpoint in $U^\fR$, we directly get that $C \subseteq U^\fB$ or $C \subseteq U^\fR$. If $C \subseteq U^\fB$, then $C \cap Q'^\fR = \emptyset$ and if $C \subseteq U^\fR$, then $C \cap Q'^\fB = \emptyset$. 

\paragraph{Ruling property} Let $v \in V(H')$. From the way we defined $H'$, it directly follows that every edge in the forest $F$ is also contained in $H'$ i.e., $E(F) \subseteq E(H')$. As $\root_F(v) \in Q'^\fR \cup Q'^\fB$, it therefore follows that 
\[d_{H',\del}(Q'^\fR \cup Q'^\fB,v) \leq d_{H',\del}(\root_F(v),v) \leq \del(\root_F(v)) + d_F(v).\]

The first property of the distance oracle directly states that $\del(\root_F(v)) + d_F(v) \leq (1+\eps) d_{H,\del}(v)$ and therefore combining the inequalities implies

\[
d_{H',\del}(Q'^\fR \cup Q'^\fB,v) \leq (1+\eps) d_{H,\del}(v) \leq (1+\eps)^{2(0+1)} d_{H,\del}(Q,v) + \left( \sum_{j=1}^0 (1+\eps)^{2j} \right) 3D,
\]

as needed.

\paragraph{Good property} We have $V^{good} = \{v \in V(H) \colon r(v) = 0\}$ and for every node $v$ with $r(v) = 0$ it trivially holds that $B_{H'}(v,r(v)) = B_H(v,r(v))$. 

\paragraph{Bad property} For the deterministic case, note that by definition $i = 0$ implies $\sum_{v \in V(H)} \mu(v) r(v) = 0$. Hence, for every $v \in V(H)$ with $r(v) > 0$, we have $\mu(v) = 0$. As $V^{bad} = \{v \in V(H) \colon r(v) > 0\}$, we therefore have $\mu(V^{bad}) = 0 = (1 - \frac{1}{2^0})O(\sum_{v \in V(H)} \mu(v)r(v)/D)$.
For the randomized case, $v \in V^{bad}$ implies $r(v) > 0$, but then $P[v \in V^{bad}] = 1 \leq (\frac{1}{2^0} + \frac{1-2^{-0}}{D})O(r(v))$.

\paragraph{Recursive Step}
Now, assume $i > 0$.
We compute the sets $U^\fR$ and $U^\fB$ in the same way as in the base case.
The recursive step is either a \emph{red step} or a \emph{blue step}. In the randomized version, we flip a fair coin to decide whether the recursive step is a red step or a blue step.
In the deterministic version, the recursive step is a red step if $\sum_{u \in U^\fR} \mu(u)r(u) \geq \sum_{u \in U^\fB} \mu(u)r(u)$ and otherwise it is a blue step.
Set $\fA = \fR$ and $\bar{\fA} = \fB$ if the step is a red step and otherwise set $\fA = \fB$ and $\bar{\fA} = \fR$.  

We invoke the deterministic/randomized version of \cref{thm:blurry_growing} with input $G_\Tref{thm:blurry_growing} = H, S_\Tref{thm:blurry_growing} = U^\fA$ and $D_\Tref{thm:blurry_growing} = D$. After the invocation, we set  $W^\fA = S^{sup}_\Tref{thm:blurry_growing}$ and $\Vbad_i = (V^{bad})_\Tref{thm:blurry_growing}$.

We next perform a recursive call with inputs $H_{rec}, i_{rec}, Q_{rec} = Q^\fR_{rec} \sqcup Q^\fB_{rec}$, $\del_{rec}$ and $R_{rec}$, which are defined below. (The preferred radii, weights, and parameters $D$ and $\eps$ will be the same)  

We set $H_{rec} = H[V(H) \setminus W^\fA]$.
For the randomized version, we set $i_{rec} = i - 1$. For the deterministic version, it follows from the way we decide whether it is a red/blue step that
\[\sum_{v \in V(H_{rec})} \mu(v) r(v) = \sum_{v \in V(H)} \mu(v)r(v) - \sum_{v \in W^\fA} \mu(v)r(v)  \leq \sum_{v \in V(H)} \mu(v)r(v) - \sum_{v \in U^\fA} \mu(v)r(v) \leq \frac{1}{2}\sum_{v \in V(H)} \mu(v) r(v).\]

As $i_{rec} = 1 + \lceil \log(\sum_{v \in V(H_{rec})} \mu(v)r(v))\rceil$ if $\sum_{v \in V(H_{rec})} \mu(v) r(v) > 0$ and $i_{rec} = 0$ otherwise, it therefore follows by a simple case distinction that $i_{rec} = i-1$ in the deterministic version as well. 
For $\fZ \in \{\fR,\fB\}$, we define
\[
Q_{rec}^{\fA} = \{v \in V(H_{rec}) \colon \text{$v$ has a parent $p(v)$ in $F$, and $p(v) \in W^\fA$}\}
\]
and 
\[Q_{rec}^{\bar{\fA}} = \{v \in V(H_{rec}) \colon \text{$v$ is a root in $F$}\}.\]
For every $v \in Q_{rec} := Q_{rec}^\fR \sqcup Q_{rec}^\fB$, we define 
\[
\del_{rec}(v) = (1+\eps)(\del(\root_F(v)) + d_F(v)) + 3D.
\]
Finally, we set \[R_{rec} = (1+\eps)^2 R + 3D.\]

We have to verify that for every $v \in V(H_{rec})$, $d_{H_{rec},\del_{rec}}(Q_{rec},v) \leq R_{rec}$ , as otherwise the input is invalid. We actually show a stronger property, namely that for every $v \in V(H_{rec})$, 
\[d_{H_{rec},\del_{rec}}(Q_{rec},v) \leq (1+\eps)^2 d_{H,\del}(Q,v) + 3D \leq (1+\eps)^2 R + 3D.\]

To that end, we consider a simple case distinction based on whether the entire path from $v$ to $\root_F(v)$ in $F$ is contained in $H_{rec}$ or not. If yes, then $d_{H_{rec}}(\root_F(v),v) \leq d_F(v)$ and as $\del_{rec}(\root_F(v)) = (1+\eps)\del(\root_F(v)) + 3D$, we get

\begin{align*}
  d_{H_{rec},\del_{rec}}(Q_{rec},v) &\leq \del_{rec}(\root_F(v)) + \d_F(v) \\
  &= (1+\eps)\del(\root_F(v)) + 3D + d_F(v) \\
  &\leq (1+\eps)(\del(\root_F(v)) + d_F(v)) + 3D  \\
  &\leq (1+\eps)^2 d_{H,\del}(Q,v) + 3D.
\end{align*}
It remains to consider the case that the path from $v$ to $\root_F(v)$ in $F$ is not entirely contained in $H_{rec}$. Starting from $v$, let $y$ be the first node such that $y$ is contained in $H_{rec}$ but $y$'s parent in $F$ is not.
By definition, $y \in Q^\fA_{rec}$ and $\del_{rec}(y) = (1+\eps)(\del(\root_F(y)) + d_F(y)) + 3D$. We get

\begin{align}
\label{eq:bla}
\nonumber d_{H_{rec},\del_{rec}}(Q_{rec},v) &\leq d_{H_{rec}}(v,y) + \del_{rec}(y) \\
\nonumber &= d_{H_{rec}}(v,y) + (1+\eps)(\del(\root_F(y)) + d_F(y)) + 3D \\
\nonumber &\leq  (1+\eps)(\del(\root_F(v)) + d_F(y) + d_{H_{rec}}(v,y)) + 3D \\
\nonumber &\leq (1+\eps)(\del(\root_F(v)) + d_F(v)) + 3D \\
 &\leq (1+\eps)^2 d_{H,\del}(Q,v) + 3D.
\end{align}

\end{proof}

We verified that the provided input is correct. We denote with $Q^{'\fR}_{rec}, Q^{'\fB}_{rec}, H'_{rec},V^{good}_{rec}$ and $V^{bad}_{rec}$ the output produced by the recursive call.

We now describe the final output. We set

\[Q'^{\fA} := Q^\fA \cap U^\fA \subseteq Q^{\fA}\]

and

\[Q'^{\bar{\fA}} := Q'^{\bar{\fA}}_{rec} \subseteq Q^{\bar{\fA}}_{rec}  \subseteq Q^{\bar{\fA}}.\]

Next, we define the output graph $H'$. To that end, we first define the edge set

\[E_{bridge} = \{\{v,p(v)\} \in E(H) \colon \text{$v \in Q^{'\fA}_{rec}$, $p(v)$ is the parent of $v$ in $F$}\}.\]

Note that each edge in $E_{bridge}$ has one endpoint in $V(H_{rec}) = V(H) \setminus W^\fA$ and one endpoint in $W^\fA$.

We now define

\[E(H') = E(H'_{rec}) \sqcup E(H[W^\fA]) \sqcup E_{bridge}.\]

Finally, we set $V^{bad} = V^{bad}_{rec} \cup V^{bad}_i$ and $V^{good} = V(H) \setminus V^{bad}$.

We next show that the output satisfies all the required properties.

\paragraph{Separation property} Let $C$ be a connected component of $H'$. We have to show that $Q'^\fR \cap C = \emptyset$ or $Q'^\fB \cap C = \emptyset$. For the sake of contradiction, assume there exists $q'^\fA \in Q'^\fA \cap C$ and $q'^{\bar{\fA}} \in Q'^{\bar{\fA}} \cap C$. 
Consider an arbitrary path $P$ from $q'^\fA$ to $q'^{\bar{\fA}}$ in $C$.

As $q'^{\fA} \in W^\fA$ and $q'^{\bar{\fA}} \in V(H_{rec}) = V(H) \setminus W^\fA$, the path $P$ contains at least one edge in $E_{bridge}$. Let $e = \{v,p(v)\}$ be the first edge in $E_{bridge}$ that one encounters on the path $P$ starting from $q'^{\bar{\fA}}$.

We have, $q'^{\bar{\fA}} \in Q'^{\bar{\fA}} = Q'^{\bar{\fA}}_{rec}$ and $v \in Q'^\fA_{rec}$. Moreover, $q'^{\bar{\fA}}$ and $v$ are in the same connected component in $H'_{rec}$, a contradiction with the separation property of the recursive call.

\paragraph{Ruling property} Let $v \in V(H)$. We have to show that 
\[d_{H',\del}(Q'^\fR \cup Q'^\fB,v)  \leq (1+\eps)^{2(i+1)} d_{H,\del}(Q,v) + \left( \sum_{j=1}^i (1+\eps)^{2j} \right) 3D.\]

First, we consider the case $v \in W^{\fA}$. The guarantees of \cref{thm:blurry_growing} implies the existence of a vertex $u \in U^\fA$ with $d_{H[W^\fA]}(v,u) \leq D$. We have

\begin{align*}
    d_{H',\del}(Q'^\fR \cup Q'^\fB,v) &\leq d_{H',\del}(\root_F(u),v) \\
    &\leq d_{H',\del}(\root_F(u),u) + D \\
    &\leq \del(\root_F(u)) + d_F(u) + D \\
    &\leq (1+\eps)d_{H,\del}(Q,u) + D \\
    &\leq (1+\eps)(d_{H,\del}(Q,v) + D) + D \\
    &\leq (1+\eps)d_{H,\del}(Q,v) + 3D \\
    &\leq (1+\eps)^{2(i+1)}d_{H,\del}(Q,v) + \left( \sum_{j=1}^{i} (1+\eps)^{2j} \right)3D.
\end{align*}

It remains to consider the case $v \in V(H_{rec})$. Recall that we already have shown in \cref{eq:bla} that

\[d_{H_{rec},\del_{rec}}(Q_{rec},v) \leq (1+\eps)^2 d_{H,\del}(Q,v) + 3D.\]

By the ruling property of the recursive call, we therefore get

\begin{align*}
   d_{H'_{rec},\del_{rec}}(Q_{rec}'^\fR \cup Q_{rec}'^\fB,v)  &\leq (1+\eps)^{2(i_{rec}+1)} d_{H_{rec},\del_{rec}}(Q_{rec},v) + \left( \sum_{j=1}^{i_{rec}} (1+\eps)^{2j} \right) 3D \\
   &\leq (1+\eps)^{2(i_{rec}+1)}((1+\eps)^2 d_{H,\del}(Q,v) + 3D) + \left( \sum_{j=1}^{i_{rec} + 1} (1+\eps)^{2j} \right) 3D \\
     &\leq (1+\eps)^{2(i+1)}d_{H,\del}(Q,v) + \left( \sum_{j=1}^{i} (1+\eps)^{2j} \right)3D.
\end{align*}

Let $q'_{rec} \in Q_{rec}'^\fR \cup Q_{rec}'^\fB$ with  $d_{H'_{rec},\del_{rec}}(q'_{rec},v) =   d_{H'_{rec},\del_{rec}}(Q_{rec}'^\fR \cup Q_{rec}'^\fB,v)$.
First, consider the case that $q'_{rec} \in Q_{rec}'^{\bar{\fA}}$ and therefore also $q' \in Q'^{\bar{\fA}}$. We have

\begin{align*}
    d_{H',\del}(Q'^\fR \cup Q'^\fB,v) 
    &\leq d_{H',\del}(q'_{rec},v) \\
    &\leq d_{H'_{rec},\del_{rec}}(q'_{rec},v) \\
    &\leq (1+\eps)^{2(i+1)}d_{H,\del}(Q,v) + \left( \sum_{j=1}^i (1+\eps)^{2j} \right) 3D,
\end{align*} 

as needed. It remains to consider the case that $q'_{rec} \in Q_{rec}'^{\fA}$. In particular, this implies that $q'_{rec}$ has a parent $p$ in $F$ which is contained in $W^\fA$. We have

\begin{align*}
    \del_{rec}(q'_{rec}) &= (1+\eps)(\del(\root_F(q'_{rec})) + d_F(q'_{rec})) + 3D \\
    &= (1+\eps)(\del(\root_F(p)) + d_F(p) + \ell(q'_{rec}, p)) + 3D \\
    &\geq (1+\eps)(\del(\root_F(p)) + d_F(p)) + 3D + \ell(q'_{rec}, p) \\
    &\geq (1+\eps) d_{H,\del}(Q,p) + 3D + \ell(q'_{rec},p) \\
    &\geq d_{H',\del}(Q'^\fR \cup Q'^\fB,p) + \ell(q'_{rec},p) \\
    &\geq d_{H',\del}(Q'^\fR \cup Q'^\fB,q'_{rec}) 
\end{align*}

and therefore

\begin{align*}
d_{H',\del}(Q'^\fR \cup Q'^\fB,v) &\leq d_{H',\del}(Q'^\fR \cup Q'^\fB,q'_{rec}) + d_{H'}(q'_{rec},v) \\
&\leq \del_{rec}(q'_{rec}) + d_{H'_{rec}}(q'_{rec},v) \\ &=  d_{H'_{rec},\del_{rec}}(Q_{rec}'^\fR \cup Q_{rec}'^\fB,v) \\
&\leq (1+\eps)^{2(i+1)}d_{H,\del}(Q,v) + \left( \sum_{j=1}^i (1+\eps)^{2j} \right) 3D,
\end{align*}

as needed.

\paragraph{Good property} Let $v \in V^{good}$. We have to show that $B_{H'}(v,r(v)) = B_H(v,r(v))$.  As $v \in V^{good}$, it holds that $v \notin \Vbad_i \cup V^{bad}_{rec}$. As $v \notin \Vbad_i$, we either have $B_H(v,r(v)) \subseteq W^\fA$ or $B_H(v,r(v)) \subseteq V(H) \setminus W^\fA = V(H_{rec})$. If $B_H(v,r(v)) \subseteq W^\fA$, then it follows from $E(H[W^\fA]) \subseteq E(H')$ that $B_{H'}(v,r(v)) = B_H(v,r(v))$. If $B_H(v,r(v)) \subseteq V(H_{rec})$, then $B_H(v,r(v)) = B_{H_{rec}}(v,r(v))$. As $v \notin V^{bad}_{rec}$, it follows that $B_{H_{rec}}(v,r(v)) = B_{H'_{rec}}(v,r(v))$ from the good property of the recursive call. As $E(H'_{rec}) \subseteq E(H')$, it therefore follows that $B_{H'}(v,r(v)) = B_H(v,r(v))$, as needed.

\paragraph{Bad property} We start with the deterministic version.
We have
\begin{align*}
    \mu(V^{bad}) &\leq \mu(V^{bad}_i) + \mu(V^{bad}_{rec}) \\
                 &\leq \frac12 O \left( \sum_{v \in V(H)} \mu(v)r(v)/D \right) + \left(1 - \frac{1}{2^{i_{rec}}}\right) O \left( \sum_{v \in V(H_{rec})} \mu(v)r(v)/D \right) \\
                 &\leq \frac12 O \left( \sum_{v \in V(H)} \mu(v)r(v)/D \right) + \frac12 \left(1 - \frac{1}{2^{i_{rec}}}\right) O \left( \sum_{v \in V(H)} \mu(v)r(v)/D \right) \\
                 &= \left(1 - \frac{1}{2^{i}}\right) O \left( \sum_{v \in V(H)} \mu(v)r(v)/D \right),
\end{align*}
as needed.
Now, we analyze the randomized version. Let $v \in (H)$. We have

\begin{align*}
     \P[v \in V^{bad}] &\leq  \P[v \in \Vbad_i] +  \P[v \in V^{bad}_{rec}] \\
    &\leq \frac12 O(r(v)/D) +  \P[v \in V^{bad}_{rec}|v \in V(H_{rec})] \cdot  \P[v \in V(H_{rec})] \\
    &\leq \frac12 O(r(v)/D) + \left(\frac{1}{2^{i_{rec}}} + \frac{1-2^{-i_{rec}}}{D} \right)O(r(v)) \cdot \frac12 \\
    &= \frac12 O(r(v)/D) + \left(\frac{1}{2^{i}} + \frac{\frac12 -2^{-i}}{D} \right)O(r(v)) \\
    &= \left(\frac{1}{2^{i}} + \frac{1-2^{-i}}{D} \right)O(r(v)),
\end{align*}

as desired.

\subsection{Corollaries}
\label{sec:main_clustering_corollaries}

Next, we present three simple corollaries of the main clustering result \cref{thm:steroids}. First, \cref{cor:clustering_for_l1} informally states that one can efficiently compute a so-called padded low-diameter partition.

\begin{corollary}
\label{cor:clustering_for_l1}
Consider the following problem on a weighted input graph $G$.
The input consists of a global parameter $D$ and in the deterministic version each node $v \in V(G)$ additionally has a weight $\mu(v)$.

The output consists of a partition $\fC$ of $G$ together with two sets $\Vgood \sqcup \Vbad = V(G)$ such that

\begin{enumerate}
    \item the diameter of $\fC$ is $O(D \log^3(n))$,
    \item for every $v \in \Vgood$, $B_G(v,D) \subseteq C$ for some $C \in \fC$,
    \item in the deterministic version, $\mu(\Vbad) \leq 0.1 \cdot \mu(V(G))$,
    \item and in the randomized version, for every $v \in V(G) \colon  \P[v \in \Vbad] \leq 0.1$.
\end{enumerate}
There is an algorithm that solves the problem above in $O(\log^3(n))$ steps, performing all oracle calls with precision parameter $\eps = \frac{1}{\log^2(n)}$ and distance parameter no larger than $O(\log^3(n) D)$.
\end{corollary}

\begin{proof}
We invoke \cref{thm:steroids} with input $H_{\Tref{thm:steroids}} = G$, $r(v)_{\Tref{thm:steroids}} = D$ and 
$\mu_{\Tref{thm:steroids}}(v) = \mu(v)$ for every $v \in V(G)$, $Q_{\Tref{thm:steroids}} = V(G)$ and $\del_{\Tref{thm:steroids}}(v) = 0$ for
every $v \in V(G)$, $R_{\Tref{thm:steroids}} = 0$, $D_{\Tref{thm:steroids}} = c \cdot \log(n) \cdot D$ for a sufficiently large constant $c$ and 
$\eps_{\Tref{thm:steroids}} = \frac{1}{\log^2(n)}$.

The input is clearly valid. In particular, $d_{G,\del_{\Tref{thm:steroids}}}(Q_{\Tref{thm:steroids}},v) = d_G(V(G),v) = 0$.

Let $\fC$ denote the output partition of $G$ and $\Vgood \sqcup \Vbad = V(G)$ the two sets returned by the algorithm.
According to the first two properties of \cref{thm:steroids}, for every cluster $C \in \fC$ there exists a node $v_C \in \fC$ such that
for every $v \in C$

\[d_{G[C]}(v_C,v) \leq O(\log^2 n)D_{\Tref{thm:steroids}} = O(\log^3 n) D\]

and therefore the diameter of $\fC$ is $O(D \log^3 n)$.
Moreover, the third property of \cref{thm:steroids} together with $r(v)_{\Tref{thm:steroids}} = D$ implies that $B_G(v,D) \subseteq C$ for some $C \in \fC$.
In the deterministic version, the fourth property of \cref{thm:steroids} implies

\[\mu(\Vbad) = O(\log (n) \cdot \sum_{v \in V(G)} \mu_{\Tref{thm:steroids}}(v)r(v)_{\Tref{thm:steroids}}/D_{\Tref{thm:steroids}}) = O(1/c)\mu(V) \leq 0.1\mu(V)\]

for $c$ being sufficiently large.

In the randomized version, the fourth property of \cref{thm:steroids} implies for every $v \in V(G)$ that

\[ \P[v \in \Vbad] = O(\log(n) r_{\Tref{thm:steroids}}(v)/D_\Tref{thm:steroids})= O(1/c) \leq 0.1\]

for $c$ being sufficiently large.

The runtime bound directly follows from the runtime bound of \cref{thm:steroids}.

\end{proof}

Next, \cref{thm:sparse_covers_good_parameters} asserts that we can compute a so-called sparse cover.
The result follows from \cref{cor:clustering_for_l1} together with the well-known multiplicative weights update method. We use \cref{thm:sparse_covers_good_parameters} to efficiently compute an $\ell_1$-embedding in \cref{sec:embedding}. 

\begin{theorem}
\label{thm:sparse_covers_good_parameters}
Consider the following problem on a weighted input graph $G$.
The input consists of a global parameter $D$.

The output consists of a cover $\{\fC_1, \fC_2, \dots, \fC_t\}$ for some $t = O(\log n)$, with the hidden constant independent of $D$, such that each $\fC_i$ is a partition of $G$ with diameter $O(D \log^3 n)$ and for every node $v \in V(G)$, $|\{i \in [t] \colon B_G(v,D) \subseteq C \text{ for some $C \in \fC_i$}\}| \geq 2t/3$.

There is an algorithm that solves the problem above in $O(\log^4(n))$ steps using the oracle $\oDist$, performing all oracle calls with precision parameter $\eps = \frac{1}{\log^2(n)}$ and distance parameter no larger than $O(\log^3(n) D)$.

\end{theorem}
\begin{proof}
We set $t = \lceil c \cdot \log(n) \rceil$ for $c$ being a sufficiently large constant. 
At the beginning, we set $\mu_{1}(v) = 1$ for every $v \in V(G)$. 

In the $i$-th step, for $1 \le  i \le t$, we apply \cref{cor:clustering_for_l1} with input $D_{\ref{cor:clustering_for_l1}} = D$ and $\mu_{\ref{cor:clustering_for_l1}}(v) = \mu_i(v)$ for every node $v \in V(G)$.
Let $\fC_i$ denote the partition and $\Vgood_i$ the set returned by the $i$-th invocation.
For every $v \in V(G)$ and $i \in [t]$, we set $\mu_{i+1}(v) = \mu_i(v)$ if $v \in \Vgood_i$ and otherwise we set $\mu_{i+1}(v) = 2\mu_i(v)$. 
In the end, the algorithm returns the cover $\{\fC_1, \fC_2, \dots, \fC_t\}$. 
This finishes the description of the algorithm. 

It remains to argue its correctness. All the properties from the statement directly follow from \cref{cor:clustering_for_l1}, except for the property that each node $v \in V(G)$ satisfies $|\{i \in [t] \colon B_G(v,D) \subseteq C \text{ for some $C \in \fC_i$}\}| \geq t/2$. To show this, it suffices to show that $|\{i \in [t] \colon v \in \Vgood_i\}| \geq t/2$.
Note that $\mu_1(V(G)) = n$ and $\mu_{t+1}(V(G)) = \mu_i(V(G)) + \mu_i(V(G) \setminus \Vgood_i) \leq 1.1\mu_i(V(G))$. Hence, $\mu_{t+1}(V(G)) \leq n \cdot 1.1^t$.
On the other hand, for every $v \in V(G)$, 

\[\mu_{t+1}(V(G)) \geq \mu_{t+1}(v) \geq 2^{t - |\{i \in [t] \colon v \in \Vgood_i\}|}.\]

Therefore, 

\[2^{t - |\{i \in [t] \colon v \in \Vgood_i\}|} \leq n \cdot 1.1^t,\]

which directly implies $|\{i \in [t] \colon v \in \Vgood_i\}| \geq 2t/3$ for $c$ being a sufficiently large constant. 
\end{proof}

The next corollary is the main building block for efficiently computing low-stretch spanning trees in \cref{sec:low_stretch_spanning_tree}.

\begin{corollary}
\label{cor:edge_cutting}
Consider the following problem on a weighted input graph $G$. The input consists of the following.
\begin{enumerate}
    \item A weighted subgraph $H \subseteq G$.
    \item In the deterministic version, each edge $e \in E(H)$ has a weight $\mu(e)$.
    \item There is a set of center nodes $Q \subseteq V(H)$, with each center node $q \in Q$ having a delay $\del(q) \geq 0$.
    \item There is a parameter $R$ such that for every $v \in V(H)$ we have $d_{H,del}(Q,v) \leq R$.
    \item There is a precision parameter $\eps \in [0,1]$.
\end{enumerate}
The output consists of a partition $\fC$ of $H$ and a set $Q' \subseteq Q$. Let $E^{bad}$ denote the set consisting of those edges in $E(H)$ whose two endpoints are in different clusters in $\fC$. The output satisfies  

\begin{enumerate}
    \item each cluster $C \in \fC$ contains exactly one node $q'_C \in Q'$,
    \item for each $C \in \fC$ and $v \in C$, we have $d_{H[C],\del}(q'_C,v) \leq d_{H,\del}(Q,v) + \eps R$,
    \item in the deterministic version, $\mu(E^{bad}) = O \left( \frac{\log^3(n)}{\eps R} \right) \cdot \sum_{e \in E(H)} \mu(e) \ell(e)$,
    \item and in the randomized version, for every $e \in E(H) \colon Pr[e \in E^{bad}] = O \left( \frac{\log^3(n)}{\eps R} \right) \cdot \ell(e)$.
\end{enumerate}

There is an algorithm that solves the problem above in $O(\log^3 n)$ steps, performing all oracle calls with precision parameter $\eps' = \Omega \left(\frac{\eps}{\log^2(n)}\right)$ and distance parameter no larger than $2R$.
\end{corollary}

\begin{proof}

We pretend for a moment that the actual weighted input graph is not $G$ but instead the subdivided graph $G_{sub}$.
We now invoke \cref{thm:steroids} with the following input:

\begin{enumerate}
    \item $H_{\Tref{thm:steroids}} = H_{sub}$
    \item  For every $e \in E(H)$ and the respective $v_e \in V(H_{sub})$, we set $r_{\Tref{thm:steroids}}(v_e) = \ell(e)$ and in the deterministic version $\mu_{\Tref{thm:steroids}}(v_e) = \ell(e)$.
    \item For every $v \in V(H)$, we set $r_{\Tref{thm:steroids}}(v) = 0$ and in the deterministic version $\mu_{\Tref{thm:steroids}}(v) = 0$.
    \item $R_{\Tref{thm:steroids}} = R, Q_{\Tref{thm:steroids}} = Q$ and $\del_{\Tref{thm:steroids}} = \del$
    \item $D_{\Tref{thm:steroids}} = \frac{c_1}{ \log^2(n)} \eps R$ and $\eps_{\Tref{thm:steroids}} = \frac{c_2 }{\log^2(n)} \eps$ for some small enough constants $c_1, c_2 > 0$. \\

\end{enumerate}
We have to verify that the input is valid, namely that for every $v \in V(H_{sub}), d_{H_{sub},\del}(Q,v) \leq R$. This directly follows from the way $H_{sub}$ is constructed from $H$ together with the fact that for every $v \in V(H)$, $d_{H,\del}(Q,v) \leq R$.
Let $\fC$ be the partition of $H$ that one obtains from the partition $(\fC)_{\Tref{thm:steroids}}$ of $H_{sub}$ by removing from each cluster $C \in \fC$ all the nodes that are not contained in $V(H)$.
The final output is $\fC$ and $Q' = (Q')_{\Tref{thm:steroids}}$.

We have to verify that the output satisfies all the properties. The property that each cluster $C \in \fC$ contains exactly one node $q'_C \in Q'$ follows from the way we obtained $\fC$ from $(\fC)_{\Tref{thm:steroids}}$, the fact that $Q'  \subseteq V(H)$ and the fact that each cluster in $(\fC)_{\Tref{thm:steroids}}$ contains exactly one node in $Q'$.
Next, consider any $C \in \fC$ and $v \in C$. Let $C_{sub} \in (\fC)_{\Tref{thm:steroids}}$ with $v \in C_{sub}$. We have

\begin{align*}
    d_{H[C],\del}(q'_C,v) &\leq d_{H_{sub}[C_{sub}], \del}(q'_C,v) \\
    &\leq (1+\eps_{\Tref{thm:steroids}})^{O(\log^2 n)}d_{H_{sub},\del}(Q,v) + O(\log^2 n)D_{\Tref{thm:steroids}} \\
    &\leq \textrm{e}^{c_2 \eps O(1)} d_{H_{sub},\del}(Q,v)    + O(1) c_1 \eps R \\
    &\leq d_{H_{sub},\del}(Q,v) + \eps R,
\end{align*}

where the last inequality follows from the fact that $c_1$ and $c_2$ are small enough.
Next, we verify the third property. It follows from the third property of \cref{thm:steroids} that for every $e \in E^{bad}$, $v_e \in V^{bad}_\Tref{thm:steroids}$. Hence, we obtain from the fourth property of \cref{thm:steroids} that

\begin{align*}
\mu(E^{bad}) 
&\leq \mu_{\Tref{thm:steroids}}( V^{bad}_\Tref{thm:steroids}) \\
&= O(\log(n) \cdot \sum_{v \in V(H_{sub})} \mu_{\Tref{thm:steroids}}(v) \;\cdot\; (r(v))_{\Tref{thm:steroids}}/D_{\Tref{thm:steroids}})\\ 
&= O \left( \frac{\log^3(n)}{\eps R} \right) \cdot \sum_{e \in E(H)} \mu(e) \ell(e).    
\end{align*}

For the randomized version, we have for every $e \in E(H)$ that

\[ \P[e \in E^{bad}] \leq  \P[v_e \in V^{bad}_\Tref{thm:steroids}] = O(\log (n) \cdot r_\Tref{thm:steroids}(v_e) / D_{\Tref{thm:steroids}}) = O \left( \frac{\log^3(n)}{\eps R} \right) \cdot \ell(e).\]

Running the algorithm of \cref{thm:steroids} takes $O(\log^3 n)$ steps, with all oracle calls having precision parameter $\eps_{\Tref{thm:steroids}} = \Omega \left( \frac{\eps}{\log^2 (n)}\right)$ and distance parameter no larger than $(1+\eps_{\Tref{thm:steroids}})^{O(\log^2 n)}R + O(\log^2 n)D_{\Tref{thm:steroids}} \leq 2R$, assuming that the actual weighted input graph is $G_{sub}$ instead of $G$. However, due to \cref{lem:simulation_subdivided}, the same holds true, up to a constant factor in the number of steps, with the actual weighted input graph being $G$. This finishes the proof.

\end{proof}

\bibliographystyle{alpha}
\bibliography{ref}

\appendix

\section{Deterministic $\ell_1$ Embedding}
\label{sec:embedding}

In this section, we show how to deterministically compute an $\ell_1$-embedding of a weighted input graph. Note that in this section we assume that the input graph is connected.
The embedding problem is defined as follows.  

\begin{definition}[$\ell_1$ embedding]
Let $G$ be a weighted and connected graph. An $\alpha$-approximate $\ell_1$-embedding of $G$ into $k$ dimensions is a function that maps each node $v$ to a $k$-dimensional vector $x_v$ such that for any two nodes $u,v \in V(G)$ we have
\[||x_u - x_v||_1  \leq d_G(u,v) \leq \alpha ||x_u - x_v||_1.\]
\end{definition}

For our result, we also need an oracle for computing potentials; the dual problem of transshipment which is a generalization of shortest path. Fortunately, the work of \cite{RGHZL2022sssp} that we use to implement the distance oracle for our parallel and distributed results also constructs the following potential oracle as a byproduct with the same parallel/distributed complexity. 

\begin{restatable}[Potential Oracle $\oPot_\eps$]{definition}{potentialoracle}
\label{def:oracle_potential}

The input is a non-empty set $S \subseteq V(G)$. 
The output is a function $\phi$ assigning each node $v \in V(G)$ a value $\phi(v) \geq 0$ such that

\begin{enumerate}
    \item $\phi(s) = 0$ for every $s \in S$,
    \item $|\phi(u) - \phi(v)| \leq d_G(u,v)$ for every $u,v \in V(G)$, and
    \item $(1+\eps) \phi(v) \geq d_G(S,v)$ for every $v \in V(G)$.
\end{enumerate}

\end{restatable}

As an application of our strong-clustering results for weighted graphs, we can now show a distributed deterministic algorithm that embeds a connected weighted graph in $\ell_1$ space with polylogarithmic stretch. 
The arguments in this section are similar to those in \cite[Section 7]{bartal2021advances}. 
We use our clustering result \cref{cor:clustering_for_l1} with $O(\log n)$ different distance scales $D_i = 2^i$ to get $O(\log n)$ partitions of $G$ for all different scales. Next, we use these partitions to define the final embedding.

\begin{theorem}
\label{thm:l1_embedding}
Let $G$ be a weighted and connected graph.
There is an algorithm that computes an $O(\log^4 n)$-approximate $\ell_1$-embedding of $G$ into $O(\log^3 n)$-dimensional $\ell_1$ space.
The algorithm takes $O(\log^5 n)$ steps with access to the distance oracle $\oDist$, performing all oracle calls with precision parameter $\eps = 1/\log^2 n$.
\end{theorem}

The proof below uses error-correcting codes. We note that the proof works also without them, the only difference then is that the $\ell_1$ embedding is only $O(\log^5 n)$-approximate. 

\begin{proof}
The algorithm computes $O(\log^2 n)$ partitions. For each partition, an embedding into $O(\log n)$ dimensions is computed. Combining these $O(\log^2 n)$ embeddings then results in an embedding into $O(\log^3 n)$ dimensions. The final embedding is then obtained from this embedding by scaling each entry by a factor determined later.

We first describe how the $O(\log^2 n)$ partitions are computed.

For every $D$ in the set $\{2^i \colon i \in \mathbb{Z}, |i| \leq \lceil \log (n \cdot \max_{e \in E(G)} \ell(e)) \rceil\}$ the algorithm invokes \cref{thm:sparse_covers_good_parameters} with input $D_{\Tref{thm:sparse_covers_good_parameters}} = D$. Note that there are $O(\log n)$ such choices for $D$.
The output is a cover $\{\fC^D_1,\fC^D_2,\ldots,\fC_t^D\}$ with $t = O(\log n)$.
Hence, we have $O(\log n)$ covers with each cover consisting of $O(\log n)$ partitions for a total of $O(\log^2 n)$ partitions. 

For each partition $\fC$, we compute an embedding of $V(G)$ into $O(\log n)$ dimensional $\ell_1$-space as follows.

First, each cluster $C \in \fC$ computes a bit string $s^C \in \{0,1\}^{10b}$ for $b = O(\log n)$ such that for every other cluster $C' \neq C$ in $\fC$, $|\{i \in [10b] \colon s^C_i = 1, s^{C'}_i = 0\}| \geq b$. 
Each cluster $C$ can compute such a string internally by applying an error-correcting function $f$ to a $b$-bit identifier of an arbitrary node of $C$. 
It is well-known that such an efficiently computable function $f$ exists. 

\begin{theorem}[cf. \cite{macwilliams1977theory}]
There exists a function $f$ from the set of $b$-bit strings to a set of $10b$-bit strings with the following properties. 
For every two $b$-bit strings $s_1$, $s_2$, the strings $f(s_1), f(s_2)$ differ on at least $b$ positions. 
Moreover, the value of $f(\cdot)$ can be computed in $\poly(b)$ time. 
\end{theorem}

For every $i \in [10b]$, we now define 

\[S_i = \{v \in V(G) \colon \text{$v$ is contained in a cluster $C \in \fC$ with $s^C_i = 0$}\}.\]

If $S_i \neq \emptyset$, then let $\phi_i \leftarrow \oPot_1(S_i)$. 
The $i$-th coordinate of the embedding is set equal to $\phi_i$, i.e., for every node $v$ the $i$-th coordinate is set to $\phi_i(v)$.
If $S_i = \emptyset$, then the $i$-th coordinate of each node is set to $0$.

This finishes the description of the embedding into $O(\log^3 n)$ dimensions.
For each node $v \in V$, we denote by $x_v$ the vector assigned to node $v$.

We finish by proving that there are constants $A_1, A_2$ such that for any two nodes $u,v$ and $n$ large we have
\[
\frac{d_G(u,v)}{A_1 \log(n)} \le ||x_u - x_v||_1 \le A_2\log^3 n \cdot d_G(u,v)
\]
in the following two claims. 

\begin{claim}
\label{cl:one_direction}
\[
||x_u - x_v||_1 \le O(\log^3 n) \cdot d_G(u,v)
\]
\end{claim}
\begin{proof}
If $\Phi$ stands for the set of all $O(\log^3 n)$ potential functions used in the definition of the embedding, we have
\begin{align*}
    ||x_u - x_v||_1 
    \le O(\log^3 n) \cdot \max_{\phi \in \Phi} | \phi(u) - \phi(v)|
    = O(\log^3 n) \cdot d_G(u,v)
\end{align*}
where the second bound follows from the second property in \cref{def:oracle_potential}. 

\end{proof}

\begin{claim}
\label{cl:other_direction}
\[
||x_u - x_v||_1 \ge d_G(u,v) \cdot \Omega(1/\log n)
\]
\end{claim}
\begin{proof}
Consider any $u, v \in V(G)$. Let $i$ be such that $D_i  < d(u, v)/q \le D_{i+1} $ where $q = O(\log^3 n)$ is such that \cref{thm:sparse_covers_good_parameters} outputs clusters of diameter at most $qD$. Such an $i$ has to exist. 

Consider the computed cover $\{\fC_1^{D_i}, \fC_2^{D_i}, \dots, \fC_t^{D_i}\}$ for the distance scale $D_i$. 
Note that for at least $2t/3$ indices $j$ we have that $B_G(u, D_i) \subseteq C$ for some $C \in \fC_j$. The same holds for $v$. Hence, for at least $t/3$ indices $j$ we have that both $u$ and $v$ have this property. 

Fix any such partition $\fC_j$ with this property. Recall that for at least $b$ out of $10b$ potentials $\phi$ we defined using the partition $\fC_j$ we have that exactly one of the nodes $u,v$ is in the set $S$ that defined $\phi$ via $\phi \leftarrow \oPot_1(S)$. 
Without loss of generality, assume $u \in S$ and $v\not\in S$. We have $\phi(u) = 0$ by the first property in \cref{def:oracle_potential}. On the other hand, we have $\phi(v) \ge d_G(S,v)/2 \ge D_i/2$ by the second property in \cref{def:oracle_potential} and the fact that $B_G(v, D_i) \cap S = \emptyset$. 
We are getting 
\begin{align}
\label{eq:lunch}
|\phi(u)-\phi(v)| \ge D_i/2. 
\end{align}

Note that out of all potential functions we defined, \cref{eq:lunch} holds for at least $2t/3 \,\cdot\, b = \Omega(\log^2 n)$ of them by above discussion. This implies
\begin{align*}
    ||x_u - x_v||_1
    \ge \Omega(\log^2 n) \cdot D_i/2
    = \Omega(\log^2 n) \cdot \frac{d_G(u,v)}{4q}
    = \Omega(1/\log n) \cdot d_G(u,v)
\end{align*}
as needed. 
\end{proof}

\end{proof}

\section{Low Stretch Spanning Trees}
\label{sec:low_stretch_spanning_tree}

This section is dedicated to prove the following theorem.

\begin{theorem}[Main Theorem]
\label{thm:low_stretch_spanning_tree}
We can compute a deterministic/randomized $O(\log^5(n))$-stretch spanning tree of a weighted and connected graph $G$ (with additional edge importance $\mu$) in $\poly(\log n)$ steps, with each oracle call using distance parameter at most $O(\diam(G))$ and precision parameter $\eps = \Omega \left(1 /{\log^3(n)}\right)$.
\end{theorem}

Essentially, it states that we can deterministically compute a $\poly(\log n)$-stretch spanning tree with $\poly(\log n)$ calls to an approximate distance oracle.
We start by recalling the notion of a low-stretch spanning tree.

\begin{definition}[Deterministic Low Stretch Spanning Tree]
\label{def:lsst_deterministic}
A deterministic $\alpha$-stretch spanning tree $T$ of a weighted and connected graph $G$ with additional edge-importance $\mu: E(G) \rightarrow \mathbb{R}_{\geq 0}$ is a spanning tree of $G$ such that
\[
\sum_{e = \{u,v\} \in E(G)} \mu(e) d_T(u,v) \le \alpha \sum_{e = \{u,v\} \in E(G)} \mu(e) \ell(u,v). 
\]
\end{definition}

\begin{definition}[Randomized Low Stretch Spanning Tree]
\label{def:lsst_randomized}
A randomized $\alpha$-stretch spanning tree of a weighted and connected graph $G$ is a spanning tree $T$ coming from a distribution $\mathcal{T}$ such that for every edge $e = \{u,v\} \in E(G)$ we have
\[
\E_{T \sim \mathcal{T}}[d_T(u,v)] \le \alpha \ell(u,v). 
\]
Equivalently, we may require
\[
\E_{T \sim \mathcal{T}}[d_T(u,v)] \le \alpha d_G(u,v) 
\]
for any $u,v \in V(G)$. 
\end{definition}

The section is structured into two subsections. In \cref{sec:lsst_star} we show how to construct a so-called star decomposition that is used in \cref{sec:lsst_main} to derive \cref{thm:low_stretch_spanning_tree}. 
We note that the algorithm here is essentially the same as in \cite{elkin2008lower,becker_emek_ghaffari_lenzen2019low_stretch_spanning_trees}. They start by computing a star decomposition, followed by recursively computing a low-stretch spanning tree in each of the clusters of the star decomposition. As we make use of oracles, it is not clear up-front that one can simultaneously recurse on vertex-disjoint subgraphs. The emphasis on this section is to formally prove that we can nevertheless implement the recursion efficiently. 
We refer the interested reader to the paper of \cite{elkin2008lower} for a more intuitive and readable presentation of the low stretch spanning tree algorithm.

\subsection{Star Decomposition}
\label{sec:lsst_star}

In this section we prove \cref{thm:generalized_star}. It asserts that we can build a so-called star decomposition with $\poly(\log n)$ calls to an approximate distance oracle. We start with the definition of star decomposition. 

\begin{definition}[$(\eps,r_0,R)$-Star Decomposition of $G$]
\label{def:star_decomposition}
    Let $G$ be a weighted graph, $\eps > 0$, $r_0 \in V(G)$ and $R \geq 0$ such that $\max_{v \in V(G)} d_G(r_0,v) \leq R$. \\
    The output consists of a partition $V(G) = V_0 \sqcup V_1 \sqcup \ldots \sqcup V_k$ for some $k \geq 0$ with $r_0 \in V_0$ together with a set of edges $E^{bridge} = \{\{y_j,r_j\} \colon j \in [k]\} \subseteq E(G)$ such that for every $j \in [k]$, $y_j \in V_0$ and $r_j \in V_j$. For $j \in \{0,1,\ldots,k\}$, we refer to $r_j$ as the root of cluster $V_j$.
    The output has to satisfy the following conditions.
    \begin{enumerate}
        \item $\forall j \in \{0,1,\ldots,k\} \colon \max_{v \in V_j} d_{G[V_j]}(r_j,v) \leq \frac{3}{4} R$. 
        \item $\forall v \in V_0 \colon d_{G[V_0]}(r_0,v) \leq (1+\eps)d_G(r_0,v)$.
        \item $\forall j \in [k], v \in V_j \colon d_{G[V_0]}(r_0,y_j) + \ell(y_j,r_j) + d_{G[V_j]}(r_j,v) \leq (1+\eps)d_G(r_0,v)$.
    \end{enumerate}
\end{definition}

\begin{theorem}[Deterministic and Randomized Generalized Star Decomposition]
\label{thm:generalized_star}
Consider the following problem on a weighted (and connected) graph $G$. The input consists of the following.

\begin{enumerate}
    \item A partition $V(G) = V_1 \sqcup V_2 \sqcup \ldots \sqcup V_k$ for some $k$.
    \item A node $r_i \in V_i$ for every $i \in [k]$.
    \item A number $R \geq 0$ such that $\max_{i \in [k],v \in V_i} d_{G[V_i]}(r_i,v) \leq R$.
    \item In the deterministic version a priority $\mu(e)$ for every edge $e \in E(G)$.
    \item A precision parameter $\eps \in [0,0.1]$.
\end{enumerate}

The output consists of the following for each $i \in [k] \colon$ a partition $V_i = V_{i,0} \sqcup V_{i,1} \sqcup \ldots \sqcup V_{i,k_i}$ together with a set of edges $E^{bridge}_i$.
We denote by $E^{in}$ the set consisting of those edges in $E$ that have both endpoints in $V_i$ for some $i \in [k]$.
Moreover, we denote by $E^{out}$ the set consisting of those edges in $E$ that have both endpoints in $V_{i,j}$ for some $i \in [k], j \in \{0,1,\ldots,k_i\}$.
The output satisfies

\begin{enumerate}
    \item for every $i \in [k], (V_{i,0} \sqcup V_{i,1} \sqcup \ldots \sqcup V_{i,k_i},E^{bridge}_i)$ is an $(\eps,r_i,R)$-star decomposition of $G[V_i]$,
    \item in the deterministic version, $\mu(E^{in} \setminus E^{out}) = O \left( \frac{\log^3(n)}{ \eps R}\right)\sum_{e  \in E^{in}} \mu(e)\ell(e)$,
    \item and in the randomized version, for every $e \in E^{in}$, $Pr[e \notin E^{out}] = O \left( \frac{\log^3(n)}{ \eps R}\right)\ell(e)$.
\end{enumerate}

There is an algorithm that solves the problem above in $O(\log^3 n)$ steps, performing all oracle calls with precision parameter $\eps' = \Omega \left( \frac{\eps}{\log^2(n)}\right)$ and distance parameter no larger than $R$.
\end{theorem}

\begin{proof}
We define $H \subseteq G$ as the graph with $V(H) = V(G)$ and $E(H) = E^{in}$. \\
Let $F \leftarrow \oDist_{\eps/100,R}(H,\{r_i \colon i \in [k]\})$. We now invoke \cref{thm:blurry_edge} on the graph $H$ with input $S_{\Tref{thm:blurry_edge}} = \{v \in V(H) \colon d_F(v) \leq (2/3)R\}$, $\mu_{\Tref{thm:blurry_edge}}(e) = \mu(e)$ for every edge $e \in E(H)$ and $D_{\Tref{thm:blurry_edge}} = \frac{\eps R}{100}$ and obtain as an output a set $S^{sup}_{\Tref{thm:blurry_edge}}$.

For $i \in [k]$, we set $V_{i,0} = V_i \cap S^{sup}_{\Tref{thm:blurry_edge}}$.
We now invoke \cref{cor:edge_cutting} with the following input, where we denote with $p(v)$ the parent of $v$ in $F$.

\begin{enumerate}
    \item $H_{\ref{cor:edge_cutting}} = H[V(H) \setminus S^{sup}_{\Tref{thm:blurry_edge}}]$
    \item For every $e \in E(H_{\ref{cor:edge_cutting}})$, $\mu_{\ref{cor:edge_cutting}}(e) = \mu(e)$
    \item $Q_{\ref{cor:edge_cutting}} = \{v \in V(H_{\ref{cor:edge_cutting}}) \colon p(v) \in S^{sup}_{\Tref{thm:blurry_edge}}\}$ and for every $q \in Q_{\ref{cor:edge_cutting}}$, $\del(q) = d_F(q) - (2/3)R \geq 0$. 
    \item $R_{\ref{cor:edge_cutting}} = R/2$
    \item $\eps_{\ref{cor:edge_cutting}} = \eps/10$
\end{enumerate}

We have to verify that the input is valid, namely that for every $v \in V(H_{\ref{cor:edge_cutting}}), d_{H_{\ref{cor:edge_cutting}},\del_{\ref{cor:edge_cutting}}}(Q_{\ref{cor:edge_cutting}},v) \leq R_{\ref{cor:edge_cutting}}$.
Consider an arbitrary $v \in V(H_{\ref{cor:edge_cutting}})$ and let $i$ such that $v \in V_i$. Then,

\[d_{H_{\ref{cor:edge_cutting}},\del_{\ref{cor:edge_cutting}}}(Q_{\ref{cor:edge_cutting}},v) \leq d_F(v) - (2/3)R \leq (1+\eps/100)d_{G[V_i]}(r_i,v) - (2/3)R \leq R/2 = R_{\ref{cor:edge_cutting}},\]
as needed.

The output is a partition $\fC_{\ref{cor:edge_cutting}}$ of $H_{\ref{cor:edge_cutting}}$ and a set $Q'_{\ref{cor:edge_cutting}} \subseteq Q_{\ref{cor:edge_cutting}}$.

For every $i \in [k]$, we now output $V_i = V_{i,0} \sqcup V_{i,1} \sqcup \ldots V_{i,k_i}$ such that for every $j \in [k_i]$, $V_{i,j}$ is one of the clusters in $\fC_{\ref{cor:edge_cutting}}$. Moreover, we define

\[E^{bridge}_i = \{\{v,p(v)\} \colon v \in V_i \cap Q'_{\ref{cor:edge_cutting}}\}.\]

We start with verifying that $(V_{i,0} \sqcup V_{i,1} \sqcup \ldots \sqcup V_{i,k_i},E^{bridge}_i)$ is an $(\eps,r_i,R)$-star decomposition of $G[V_i]$.
First, we have to check that $E^{bridge}_i = \{\{y_{ij},r_{ij}\} \colon j \in [k_i]\}$  for some $y_{ij}$ and $r_{ij}$ with $y_{ij} \in V_{i,0}$ and $r_{ij} \in V_{i,j}$.

Note that each cluster in $\fC_{\ref{cor:edge_cutting}}$ (and therefore each $V_{i,j}$ for $j \in [k_i]$) contains exactly one node in $Q'_{\ref{cor:edge_cutting}}$, which we denote by $r_{ij}$. We have $r_{ij} \in V_{ij}$ and as $r_{ij} \in Q'_{\ref{cor:edge_cutting}} \subseteq Q_{\ref{cor:edge_cutting}}$, we directly get that $y_{ij} := p(r_{ij}) \in V_i \cap S^{sup}_{\Tref{thm:blurry_edge}} =: V_{i,0}$, as desired. 
Next, we verify that for every $j \in \{0,1,\ldots,k_i\}$, $\max_{v \in V_{i,j}} d_{G[V_{i,j}]}(r_{ij},v) \leq \frac{3}{4}R$ where we set $r_{i0} = r_i$. We have

\begin{align*}
\max_{v \in V_{i,0}}d_{G[V_{i,0}]}(r_i,v) &= \max_{v \in V_i \cap S^{sup}_{\Tref{thm:blurry_edge}}} d_{G[ V_i \cap S^{sup}_{\Tref{thm:blurry_edge}}]}(r_i,v) \\
&\leq \max_{v \in V_i \cap S_{\Tref{thm:blurry_edge}}} d_{G[ V_i \cap S_{\Tref{thm:blurry_edge}}]}(r_i,v) + 
\max_{v \in V_i \cap S^{sup}_{\Tref{thm:blurry_edge}}} d_{G[ V_i \cap S^{sup}_{\Tref{thm:blurry_edge}}]}(S_{\Tref{thm:blurry_edge}},v) \\
&\leq \max_{v \in S_{\Tref{thm:blurry_edge}}} d_F(v) + \max_{v \in S^{sup}_{\Tref{thm:blurry_edge}}} d_{H[S^{sup}_{\Tref{thm:blurry_edge}}]}(S_{\Tref{thm:blurry_edge}}, v) \\
&\leq (2/3)R + D_{\Tref{thm:blurry_edge}} \\
&\leq (3/4)R.
\end{align*}

For $j \in [k_i]$ and $v \in V_{i,j}$, we have

\[d_{G[V_{i,j}]}(r_{ij},v) \leq  d_{G[V_{i,j}],\del_{\ref{cor:edge_cutting}}}(r_{ij},v) \leq (1+\eps/10) d_{H_{\ref{cor:edge_cutting}},\del_{\ref{cor:edge_cutting}}}(Q_{\ref{cor:edge_cutting}},v) \leq (1+\eps/10) R_{\ref{cor:edge_cutting}} \leq (3/4)R.\]

Next, we verify the second property. Let $v \in V_{i,0}$. We have shown above that

\[d_{G[V_{i,0}]}(r_{i0},v) \leq (2/3)R + D_{\Tref{thm:blurry_edge}} = (2/3)R + \frac{\eps R}{100} \]

and therefore $d_{G[V_{i,0}]}(r_{i0},v) \leq (1+\eps/10)d_{G[V_i]}(r_{i0},v)$ as long as 

\[d_{G[V_i]}(r_{i0},v) \geq \frac{(2/3)R + \frac{\eps R}{100}}{1+\eps/10}.\]

Therefore, it remains to consider the case

\[d_{G[V_i]}(r_{i0},v) < \frac{(2/3)R + \frac{\eps R}{100}}{1+\eps/10} \leq \frac{(2/3)R}{1 + \eps/100},\]

which in particular implies

\[d_F(v) \leq (1+\eps/100)d_{G[V_i]}(r_{i0},v) \leq (2/3)R.\]

Thus, the entire path from $r_{i0}$ to $v$ in $F$ is contained in $V_{i,0}$ and therefore 

\[d_{G[V_{i,0}]}(r_{i0},v) \leq d_F(v) \leq (1+\eps/100) d_{G[V_i]}(r_{i0},v) \leq  (1+\eps)d_{G[V_i]}(r_{i0},v).\]

It remains to verify the third property. Consider an arbitrary $j \in [k_i]$ and $v \in V_{i,j}$. We have

\begin{align*}
    d_{G[V_{i,j}]}(r_{ij},v) &= d_{G[V_{i,j}],\del_{\ref{cor:edge_cutting}}}(r_{ij},v) - \del_{\ref{cor:edge_cutting}}(r_{ij}) \\
    &\leq (1+\eps)d_{H_{\ref{cor:edge_cutting}},\del_{\ref{cor:edge_cutting}}}(Q_{\ref{cor:edge_cutting}},v) - \del_{\ref{cor:edge_cutting}}(r_{ij}) \\
    &\leq (1+\eps/10)(d_F(v) - (2/3)R) - (d_F(r_{ij}) - (2/3)R) \\
    &\leq d_F(v) - d_F(r_{ij}) + \frac{\eps R}{10} \\
    &= d_F(v) -d_F(y_{ij}) - \ell(y_{ij},r_{ij}) + \frac{\eps R}{10} \\
    &\leq (1+\eps/100)d_{G[V_i]}(r_{i0},v) - d_{G[V_i]}(r_{i0},y_{ij}) - \ell(y_{ij},r_{ij}) + \frac{\eps R}{10} \\
    &\leq d_{G[V_i]}(r_{i0},v) - d_{G[V_i]}(r_{i0},y_{ij}) - \ell(y_{ij},r_{ij}) + \frac{\eps R}{5}.
\end{align*}

In particular,

\[\ell(y_{ij},r_{ij}) + d_{G[V_{i,j}]}(r_{ij},v) \leq d_{G[V_i]}(r_{i0},v) - d_{G[V_i]}(r_{i0},y_{ij}) + \frac{\eps R}{5}.\]

Therefore,

\begin{align*}
    d_{G[V_{i,0}]}(r_{i0},y_{ij}) + \ell(y_{ij},r_{ij}) + d_{G[V_{i,j}]}(r_{ij},v) &\leq 
    (1+\eps/10)d_{G[V_i]}(r_{i0},y_{ij}) + d_{G[V_i]}(r_{i0},v) - d_{G[V_i]}(r_{i0},y_{ij}) + \frac{\eps R}{5} \\
    &\leq d_{G[V_i]}(r_{i0},v) + \frac{\eps R}{2} \\
    &\leq (1+\eps) d_{G[V_i]}(r_{i0},v)
\end{align*}

where the last inequality follows from $d_{G[V_i]}(r_{i0},v) \geq \frac{d_F(v)}{1+\eps/100} \geq \frac{(2/3)R}{1+\eps/100} \geq R/2$.

Hence, $(V_{i,0} \sqcup V_{i,1} \sqcup \ldots \sqcup V_{i,k_i},E^{bridge}_i)$ is indeed an $(\eps,r_i,R)$-star decomposition of $G[V_i]$. 

Each edge in $E^{in} \setminus E^{out}$ either has exactly one endpoint in $S^{sup}_{\Tref{thm:blurry_edge}}$ or the two endpoints are contained in different clusters in $\fC_{\ref{cor:edge_cutting}}$. Therefore, by the guarantees of \cref{thm:blurry_edge} and \cref{cor:edge_cutting}, we get in the deterministic version that

\[\mu(E^{in} \setminus E^{out}) = O \left( \sum_{e \in E(H)} \mu(e)\ell(e)/D_{\Tref{thm:blurry_edge}} \right) + O \left( \frac{\log^3(n)}{\eps R} \right) \cdot \sum_{e \in E(H_{\ref{cor:edge_cutting}})} \mu(e) \ell(e) = O \left( \frac{\log^2(n)}{ \eps R}\right)\sum_{e  \in E^{in}} \mu(e)\ell(e)\]

and in the randomized version for every $e \in E^{in}$ that

\[Pr[e \notin E^{out}] = O \left( \frac{\ell(e)}{D_{\Tref{thm:blurry_edge}}}\right) + O \left( \frac{\log^3(n)}{\eps R} \ell(e)\right) = O \left( \frac{\log^3(n)}{\eps R}\right) \ell(e),\]

as desired.

\end{proof}

\subsection{Proof of \cref{thm:low_stretch_spanning_tree}}
\label{sec:lsst_main}

We are now ready to prove the main theorem of this section, \cref{thm:low_stretch_spanning_tree_recursion}. \cref{thm:low_stretch_spanning_tree} follows as a simple corollary of it by setting $(V_1)_{\Tref{thm:low_stretch_spanning_tree_recursion}} = V(G)$, letting $(r_1)_{\Tref{thm:low_stretch_spanning_tree_recursion}}$ be an arbitrary node, $j_{\Tref{thm:low_stretch_spanning_tree_recursion}} = O(\log \diam(G))$, i.e., such that $(4/3)^{j_{\Tref{thm:low_stretch_spanning_tree_recursion}}} = O(\diam(G))$, and $\eps_{\Tref{thm:low_stretch_spanning_tree_recursion}} = \frac{1}{\log(n)}$.

\begin{theorem}
\label{thm:low_stretch_spanning_tree_recursion}
Consider the following problem on a weighted (and connected) input graph $G$. 

\begin{enumerate}
    \item A partition $V(G) = V_1 \sqcup V_2 \sqcup \ldots \sqcup V_k$ for some $k$.
    \item A node $r_i \in V_i$ for every $i \in [k]$.
    \item A natural number $j \in \mathbb{N}$ such that $\max_{i \in [k],v \in V_i} d_{G[V_i]}(r_i,v) \leq (4/3)^{j-2}$
    \item In the deterministic version a priority $\mu(e)$ for every edge $e \in E(G)$.
    \item A precision parameter $\eps \in [0,0.1]$.
\end{enumerate}
The output is a weighted forest $F \subseteq G$ (in the randomized version coming from a distribution $\mathcal{F}$). We denote by $E^{in}$ the set consisting of those edge in $E$ that have both endpoints in $V_i$ for some $i \in [k]$.
The output satisfies the following.

\begin{enumerate}
    \item $V_1,V_2,\ldots,V_k$ are the connected components of $F$.
    \item $\forall i \in [k], v \in V_i \colon d_F(r_i,v) \leq (1+\eps)^j d_{G[V_i]}(r_i,v)$
    \item Deterministic Version: \\
    $\sum_{e = \{u,v\} \in E^{in}} \mu(e)d_F(u,v) \leq j \cdot (1+\eps)^j O \left( \frac{\log^2(n)}{ \eps}\right) \sum_{e = \{u,v\} \in E^{in}} \mu(e) \ell(u,v)$. 
    \item Randomized Version: $\forall e = \{u,v\} \in E^{in} \colon \E_{F \sim \mathcal{F}}[d_F(u,v)] \leq j  \cdot (1+\eps)^j \cdot O \left( \frac{\log^2(n)}{ \eps}\right)\ell(u,v)$.
          
\end{enumerate}
We can compute the output in $(j+1)\cdot \poly(\log n)$ steps, with each oracle call using precision parameter $\eps' = \Omega \left( \frac{\eps}{\log^2(n)}\right)$.

\end{theorem}

\begin{proof}
We prove the statement by induction on $j$. For $j = 1$, the third property of the input together with all edge weights being non-negative integer implies that the diameter of $G[V_i]$ is $0$ for each $i \in [k]$. Hence, it is easy to verify that the forest $F$ one obtains by calling $\oDist_{\eps,(4/3)^{j-2}}( \{r_1,r_2,\ldots,r_k\})$ satisfies all the conditions.

Now, consider an arbitrary $j > 1$ and assume that the statement holds for $j - 1$.

We first invoke \cref{thm:generalized_star} with the same input that we received (Setting $R = (4/3)^{j-2}$).

As an output, we obtain for every $i \in [k]$ a partition $V_{i,0} \sqcup V_{i,1} \sqcup \ldots \sqcup V_{i,k_i}$ and a set of edges $E_i^{bridge}$ such that $(V_{i,0} \sqcup V_{i,1} \sqcup \ldots \sqcup V_{i,k_i},E^{bridge}_i)$ is a $(G[V_i],\eps,r_i,R)$-star decomposition of $G[V_i]$. For $j \in \{0,1,\ldots,k_i\}$, we denote with $r_{ij}$ the root of the cluster $V_{i,j}$.

Now, we perform a recursive call with input partition $V(G) = (V_{1,0} \sqcup V_{1,1} \sqcup \ldots \sqcup V_{1,k_1}) \sqcup \ldots \sqcup (V_{k,0} \sqcup V_{k,1} \sqcup \ldots \sqcup V_{1,k_k})$, nodes $r_{ij} \in V_{i,j}$ for every $i \in [k], j \in \{0,1,\ldots,k_i\}$, setting $j_{rec} = j -1$ and with the same priorities and precision parameter $\eps$.

First, we have to verify that the input to the recursive call is valid. This requires us to show that  $\max_{i \in [k], j \in [k_i],v \in V_{i,j}} d_{G[V_{i,j}]}(r_{ij},v) \leq (4/3)^{j_{rec} - 2} = (4/3)^{j - 3}$.
Consider an arbitrary $i \in [k]$. As $(V_{i,0} \sqcup V_{i,1} \sqcup \ldots \sqcup V_{i,k_i},E^{bridge}_i)$ is a $(G[V_i],\eps,r_i,R)$-star decomposition of $G[V_i]$, the first guarantee of \cref{def:star_decomposition} states that for every $j \in \{0,1,\ldots,k_i\}$,

\[\max_{v \in V_{i,j}} d_{G[V_{i,j}]}(r_{ij}, v) \leq \frac{3}{4}(4/3)^{j-2} = (4/3)^{j-3},\]

as desired.

Now, let $F_{rec}$ denote the forest obtained from the recursive call.
We now return the forest $F$ that one obtains from $F_{rec}$ by adding all the edges in $\bigcup_{i \in [k]} E^{bridge}_i$ to it.
We now have to verify that $F$ satisfies all the conditions.

We start by verifying the second condition.
Consider an arbitrary $i \in [k]$ and $v \in V_i$. We have to show that
$d_F(r_i,v) \leq (1+\eps)^j d_{G[V_i]}(r_i,v)$.
First, consider the cases that $v \in V_{i,0}$.
As $r_i = r_{i0} \in V_{i,0}$, it follows from the guarantee of the recursive call that 

\[d_{F}(r_i,v) \leq d_{F_{rec}}(r_{i0},v) \leq (1+\eps)^{j_{rec}} d_{G[V_{i,0}]}(r_{i0},v) \leq (1+\eps)(1+\eps)^{j_{rec}}d_{G[V_i]}(r_{i0},v) = (1+\eps)^{j}d_{G[V_i]}(r_i,v),\]
where the last inequality follows from the second guarantee of  \cref{def:star_decomposition}.
It remains to consider the case that $v \in V_{i,j}$ for some $j \in [k_i]$.
According to the guarantees of the recursive call and of $\cref{def:star_decomposition}$, there exists an edge $\{y_{ij},r_{ij}\} \in E^{bridge}_i$ with $y_{ij} \in V_{i,0}$ and $r_{ij}$ being the root of $V_{i,j}$ such that

\begin{align*}
    d_F(r_i,v) &\leq d_{F_{rec}}(r_{i0},y_{i0}) + \ell(y_{ij},r_{ij}) + d_{F_{rec}}(r_{ij},v) \\
    &\leq (1+\eps)^{j_{rec}} d_{G[V_{i,0}]}(r_{i0},y_{ij}) + \ell(y_{ij},r_{ij}) + (1+\eps)^{j_{rec}}d_{G[V_{i,j}]}(r_{ij},v)) \\
    &\leq (1+\eps)^{j_{rec}} \left( d_{G[V_{i,0}]}(r_{i0},y_{ij}) + \ell(y_{ij},r_{ij}) + d_{G[V_{i,j}]}(r_{ij},v) \right) \\
    &\leq (1+\eps)^{j_{rec}} \left( (1+\eps) d_{G[V_i]}(r_{i0},v) \right) \\
    &= (1+\eps)^j d_{G[V_i]}(r_i,v),
\end{align*}

which finishes the proof of second property.

Next, we verify that $F$ is indeed a forest and that $V_1,V_2,\ldots,V_k$ are the connected components of $F$. Note that the second property ensures that any two nodes in $V_i$ are in the same connected component of $F$ for every $i \in [k]$.
Hence, to verify that $F$ is a forest and $V_1,V_2,\ldots,V_k$ are the connected components of $F$, it suffices to show that $F$ has at most $|V(G)| - k$ edges.

We have

\[|E(F)| \leq |E(F_{rec})| + \sum_{i=1}^k |E^{bridge}_i| \leq |V(G)| - \sum_{i=1}^k (k_i + 1) + \sum_{i=1}^k k_i = |V(G)| - k,\]

as desired.

We now verify the third property, which only applies to the deterministic version.
we denote by $E^{out}$ the set consisting of those edge in $E$ that have both endpoints in $V_{i,j}$ for some $i \in [k], j \in \{0,1,\ldots,k_i\}$. We have

\begin{align*}
    \sum_{e = \{u,v\} \in E^{in}} \mu(e)d_F(u,v) &= \sum_{e = \{u,v\} \in E^{in} \setminus E^{out}} \mu(e)d_F(u,v) + \sum_{e = \{u,v\} \in E^{out}} \mu(e)d_F(u,v) \\
    &\leq 2 (1+\eps)^j (4/3)^{j-2}  \sum_{e = \{u,v\} \in E^{in} \setminus E^{out}} \mu(e)  + \sum_{e = \{u,v\} \in E^{out}} \mu(e)d_{F_{rec}}(u,v) \\
    &\leq 2  (1+\eps)^j  (4/3)^{j-2} O \left( \frac{\log^2(n)}{ \eps (4/3)^{j-2}}\right) \sum_{e = \{u,v\} \in E^{in}} \mu(e) \ell(u,v) + \sum_{e = \{u,v\} \in E^{out}} \mu(e)d_{F_{rec}}(u,v) \\
    &\leq O \left( \frac{\log^2(n)}{ \eps}\right)(1+\eps)^j \sum_{e = \{u,v\} \in E^{in}} \mu(e) \ell(u,v) \\
    &+ j_{rec}(1 + \eps)^{j_{rec}}O \left( \frac{\log^2(n)}{ \eps}\right) \sum_{e = \{u,v\} \in E^{out}} \mu(e) \ell(u,v) \\
    &\leq j (1+\eps)^j O \left( \frac{\log^2(n)}{ \eps}\right)\sum_{e = \{u,v\} \in E^{in}} \mu(e) \ell(u,v),
\end{align*}

as needed.

Finally we verify the fourth property, which only applies to the randomized version. Let $e = \{u,v\} \in E^{in}$ be an arbitrary edge.

\begin{align*}
    \E_{F \sim \mathcal{F}}[d_F(u,v)] &\leq \E_{F \sim \mathcal{F}}[d_F(u,v)| e \notin E^{out}] Pr[e \notin E^{out}] + \E_{F \sim \mathcal{F}}[d_F(u,v) | e \in E^{out}] \\
    &\leq 2(1+\eps)^j(4/3)^{j-2}O \left( \frac{\log^2(n)}{\eps R}\right) \ell(u,v) + j_{rec}(1 + \eps)^{j_{rec}} \left( \frac{\log^2(n)}{\eps R}\right) \ell(u,v) \\
    &\leq j \cdot (1+\eps)^j \cdot O \left( \frac{\log^2(n)}{\eps} \right) \ell(u,v),
\end{align*}

as desired.

It remains to analyze the running time. Running the algorithm of \cref{thm:generalized_star} takes $\poly(\log n)$ steps and all oracle calls use precision parameter $\eps' = \Omega \left( \frac{\eps}{\log^2(n)}\right)$. Performing the recursive call takes $(j_{rec} + 1) \poly(\log n)$ steps and all oracle calls use precision parameter $\eps' = \Omega \left( \frac{\eps}{\log^2(n)}\right)$. Hence, the algorithm overall performs $(j + 1)\poly(\log n)$ steps and uses precision parameter $\eps' = \Omega \left( \frac{\eps}{\log^2(n)}\right)$, as desired.  
\end{proof}

\section{$D$-separated clustering}
\label{sec:sparse-cover}

In this section we show how to construct a solution to the $D$-separated clustering problem defined in \cref{sec:intro}. Recall that a $D$-separated clustering problem asks us to cluster at least a constant fraction of nodes into $D$-separated clusters of diameter $\tO(D)$. 

Before stating the main result, we note that we make our result stronger by using weaker oracles. 
The reason to do so is that this stronger result is also used in \cite{RGHZL2022sssp}. We will now define the weaker distance oracle that we use. 

\begin{definition}[Weak distance oracle $\oWeak_{\eps,D}$]
\label{def:oracle_dist_weak}

The input of this oracle consists of a subset $S \subseteq V(G)$. 

The output is a weighted forest $F \subseteq G$ rooted at $S$. 
The output has to satisfy the following:
\begin{enumerate}
    \item Every $v \in V(F)$ satisfies $d_F(v) \le (1+\eps)D$. 
    \item Every $v \in V(G)$ with $d_G(S, v) \le D$ has $v \in V(F)$. 
\end{enumerate}
\end{definition}

An important property of \cref{def:oracle_dist_weak} is that the weak distance oracle $\oWeak_{\eps, D}$ can only be used to computing approximate distances in the same weighted graph $G$. The way we defined $\oWeak_{\eps, D}$ does not allow the use of it on a subgraph of the input graph $G$, nor does it allow changing the weights of edges of $G$ between different calls to the oracle during the algorithm. 


\begin{restatable}{theorem}{strongclustering}[Deterministic $D$-separated Strong-Diameter Clustering]
\label{lem:strongdiam}
Let $G$ be a weighted graph, $\Vrem \subseteq V(G)$ and $D > 0$. We assume that each node of $G$ has a unique identifier from $[2^b]$ for $b=O(\log n)$ and set $\eps = \frac{1}{100b\log(n)}$. 

We can compute an $\eps D$-separated $O(D \log^4(n))$-strong-radius clustering $\fC_{out}$ such that $V(\fC_{out}) \subseteq \Vrem$ and $|V(\fC_{out})| \ge |\Vrem|/3$.

The computation consists of $\poly(\log n)$ \congest rounds in $G$ and $\poly(\log n)$ calls to $\oWeak_{\eps, D'}$ and $\oForestAgg_{D'}$ for various $D' \in [D , 10 D]$. 
\end{restatable}

We note that \cref{lem:strongdiam} implies \cref{thm:separated_decomposition} that we restate here for convenience. 

\separatedDecomposition*

\begin{proof}
As noted in the fifth item of \cref{thm:congestpa-simulation}, in an undirected graph, a call to a distance oracle can be implemented trivially in $O(D)$ rounds. 

We iterate \cref{lem:strongdiam} $O(\log n)$ times, each time setting $S_\Tref{lem:strongdiam}$ to be the set of yet unclustered nodes.

\end{proof}

In \cite{RGHZL2022sssp}, a similar object called a sparse neighborhood cover is needed. We define that object next. 

\begin{definition}[Sparse Neighborhood Cover]
Finally, a sparse neighborhood cover of a graph $G$ with covering radius $R$ is a collection of $\gamma = O(\log n)$ clusterings $\fC_1, \dots, \fC_\gamma$ such that for each node $v \in V(G)$ there exists some $i \in \{1, \dots, \gamma\}$ and some $C \in \fC_i$ with $B(v, R) \subseteq C$. 
\end{definition}

The following result is another straightforward corollary of \cref{lem:strongdiam}. 

\begin{restatable}{theorem}{sparseCover}
\label{thm:sparse_cover}
Let $G$ be a weighted graph and $D \geq 0$.
Assume we have access to oracle $\oWeak_{\eps,D'}$ for various $D' \in [\frac{D}{\log^7 n},D]$ and $\eps = \frac{1}{\log^3 n}$ and access to an oracle $\oForestAgg_{D'}$ for $D' \le 2D$. 

We can compute a sparse neighborhood cover with covering radius $\frac{D}{\log^7 n}$ such that each cluster $C$ in one of the clusterings comes with a rooted tree  $T_C$ of diameter at most $D$. 
The algorithm runs in $\tO(1)$ \congest rounds and needs to call $\tO(1)$ times oracles $\oWeak,\oForestAgg$.
\end{restatable}

The proofs of \cref{lem:strongdiam,thm:sparse_cover} consist of three parts.
First, in \cref{sec:weak_clustering} we show how to compute a $D$-separated $\tilde{O}(D)$-weak-diameter clustering (see below for the definition) that clusters a large fraction of the vertices (\cref{lem:weak_diam}). 
Second, in \cref{sec:strong_radius} we show how to use the weak-diameter clustering to derive \cref{lem:strongdiam}). 
Finally, we derive \cref{thm:sparse_cover} by $O(\log n)$ invocations of \cref{lem:strongdiam}. 

The weak- and strong- diameter clustering algorithms are mostly generalizations of the \congest algorithms from \cite{rozhon_ghaffari2019decomposition,chang_ghaffari2021strong_diameter} to a weighted setting with access to approximate distances. 

The added difficulty is that due to the fact that we deal with approximate, and not exact, distances, some invariants from the original algorithms are now slowly deteriorating during the algorithm, and hence the analysis requires some more care.  

\paragraph{Definitions}

In this appendix, we slightly change and extend definitions of clusters and clusterings. 
Whenever we talk about a cluster, it is, formally, not just a set of nodes $C \subseteq V(G)$, but a pair $(C, T_C)$, where $C \subseteq V(G)$ and $T_C$ is a rooted tree. 
We have $C \subseteq V(T_C)$ and we think of $T_C$ as the \emph{Steiner tree} collecting the nodes of $C$. 
The \emph{radius} $R$ of a cluster $(C, T_C)$ is the radius of $T_C$. 
If $V(T_C) = C$, we talk about a \emph{strong-radius} cluster, otherwise we talk about a \emph{weak-radius} cluster. 
Sometimes, we still informally talk about ``a cluster $C$'' when the tree $T_C$ is clear from context.

Recall that a clustering $\fC = (C_j, T_{C_j})_{j \in J}$ is a collection of disjoint clusters. We say that $\fC$ is $D$-separated if for every $(C_i, T_{C_i}), (C_j, T_{C_j}) \in \fC, i \not= j,$ we have $\dist_G(C_i, C_j) \ge D$.

\subsection{Weak-Radius Clustering}
\label{sec:weak_clustering}

We start by proving a weighted version of a weak-diameter decomposition result from \cite{rozhon_ghaffari2019decomposition}. 

\begin{lemma}[Deterministic Weak-Radius Clustering]
\label{lem:weak_diam}
Let $G$ be a weighted graph, $\Vrem \subseteq V(G)$ be a subset of its nodes, $D > 0$ and $\delta > 0$. We assume that each node of $G$ has a unique identifier from $[2^b]$ for $b=O(\log n)$. We define $\eps := \frac{1}{100b\log(n)}$. 

We can compute a $(1-\eps)^b D$-separated $O(D \log^3(n) \cdot 1/\delta)$-weak-radius clustering $\fC = (C_j, T_{C_j})_{j \in J}$ such that 
\begin{enumerate}
    \item $V(\fC) \subseteq \Vrem$,
    \item $ | \Vrem \setminus V(\fC) | \le \delta \cdot | \Vrem |$.
    \item We can implement an analogue of the oracle $\oForestAgg$ on $\{T_{C_j}\}_{j \in J}$ in $\poly\log(n)$ \congest rounds and calls to the oracle $\oForestAgg_{\eps, 2D}$. 
    \item  Consider any $W \subseteq \Vrem$ and $(C,T_C) \in \fC$. 
    If $\dist_G(W, \Vrem \setminus W) > (1+\eps) D$, then we have the following. 
    Whenever we have $C \cap W \not = \emptyset$, then $C \subseteq W$ and for every $u \in V(T_C)$ we have $\dist_G(u, W) \le (1+\eps)D/2$. 
\end{enumerate}

The computation consists of $\poly\log(n)$ \congest rounds in $G$ and $\poly(\log n)$ calls to distance oracles $\oWeak_{\eps, D'}$ for various $D' \in [D/2, D]$ and $\oForestAgg_{D'}$ for $D' \le 2D$.
\end{lemma}

The lemma is proven by a routine adaptation of the algorithm from \cite{rozhon_ghaffari2019decomposition}.

\begin{proof}
The algorithm consists of $b$ phases and each phase consists of $O(\frac{1}{\delta} \cdot b \cdot \log( n )) = O(\log^2 (n) / \delta)$ steps. 
At the beginning, all nodes in $\Vrem$ are alive, that is, $\Valive = \Vrem$, but some nodes in $\Vrem$ stop being alive during the algorithm and then we call them dead. 
Throughout the course of the algorithm we maintain a partition of alive nodes into clusters. 
The Steiner tree $T_C$ of any such cluster $C$ can also contain any nodes from $V(G)$. 
At the very beginning, each node $u$ is a trivial cluster $(\{u\}, T_{\{u\}})$. 

Each cluster is assigned a unique identifier. In the beginning, the identifier of $\{u\}$ is simply defined as the identifier of $u$ and $C$ keeps that identifier, although during the algorithm $C$ can lose some nodes, including $u$. On the other hand, $T_C$ never loses nodes. 
We will maintain as an invariant that the radius of each $T_C$ grows only by $O(D)$ in each of the $\poly\log(n)$ many steps, and that for the current clustering $\fC$ we can implement the oracle $\oForestAgg$ on $\{T_{C}\}_{C \in \fC}$ although the trees $T_C$ may not be edge-disjoint. 

\paragraph{One Phase}
We now discuss the algorithm in greater detail. 
At the beginning of phase $i$, we mark all clusters such that the $i$-th bit in their identifier is $0$ as active and the rest of the clusters are marked as passive.

We now describe one step of the algorithm, in which active clusters potentially grow. Each active cluster is in one of two states. Either it is \emph{growing} or \emph{finished}. At the beginning of the phase, every active cluster is growing. 
Let $\Vactive \subseteq \Valive$ be the set of all nodes in active clusters. 

In each step of the $i$-th phase, we define $S$ as the set consisting of all the nodes in growing clusters and use the oracle $\oWeak_{\eps,D\cdot (1-\eps)^i}$ with input $S$.
Note that we are allowed to use the oracle as $D \cdot (1-\eps)^b \ge D / 2$.  
The oracle outputs a forest $F$ with $S$ being the set of roots.

\paragraph{Growth of a cluster}
For each growing cluster $C$ consider all trees in $F$ such that their root is in $C$. We denote this forest as $F_C$ and think about is as a Steiner tree that collects all nodes ``proposing to join'' the cluster $C$ and which enables us to aggregate information about the proposing nodes. 
Every cluster $C$ computes $|F_C \cap (\Vrem \setminus \Vactive)|$. 
If $|F_C \cap (\Vrem \setminus \Vactive)| \ge \frac{\delta}{b} |C|$, the cluster $C$ grows. 
That is, all nodes in $F_C \cap (\Vrem \setminus \Vactive)$ join $C$.  These nodes leave their respective passive clusters.  This means that 
\begin{align*}
C_{new} := C \cup (F_C \cap (\Vrem \setminus \Vactive)).
\end{align*}

We now discuss how the Steiner tree $T_C$ is extended. 
Let $F'_C$ denote the rooted forest one obtains from $F_C$ by deleting all the nodes in $F_C$ whose subtree does not contain nodes of $C_{new}$. 
That is, $F'_C$ is a ``pruned'' version of the Steiner tree $F_C$ that does not contain nodes from $\Vrem$ that are not needed for aggregation. 

The vertex set of the new Steiner tree $T_{C_{new}}$ is $V(T_{C_{new}}) = V(T_C) \cup F'_C$. The edge set is $E(T_{C_{new}}) = E(T_C) \cup \{(u,v) \in E(F'_C) | u \not\in T_C \}$, where we use that each edge $e \in E(F'_C)$ knows its orientation from $u$ to $v$ in the direction of the root. That is, when combining the edge sets of the tree $T_C$ and the forest $F'_C$, we drop each edge of $F'_C$ that would result in some node having two parents. 

Otherwise, when $C$ decides not to grow, all nodes in $F_C \cap (\Vrem \setminus \Vactive)$ are removed from $\Valive$. 
In that case, $C$ finishes and stops growing in the current phase. 

This finishes the description of the algorithm, up to implementation details considered later. 

\begin{claim}
The total number of deleted nodes is at most $\delta \cdot   |\Vrem|$. 
\end{claim}
\begin{proof}
  We prove that in each phase we delete at most $\delta |\Vrem|/b$ nodes. The claim then follows. 
  Consider one phase of the algorithm. Let $C$ be some cluster and $n_C$ the total number of nodes contained in $C$ at the end of the phase. 
The cluster $C$ is responsible for deleting nodes of total weight at most $\delta n_C / b$. Hence, the total number of deleted nodes is upper bounded by $\sum_{C \text{ active in phase }i} \delta n_C /b \le \delta |\Vrem|/b$, as needed.  
\end{proof}

\begin{claim}
\label{cl:stop_growing}
In each phase, each cluster stops growing after $O(\log n \cdot b/\delta)$ steps. 
\end{claim}
\begin{proof}
For the sake of contradiction, assume there exists a cluster $C$ that was growing for all $t = O(\log  n  \cdot b/\delta)$ steps of a given phase. Then, $|C| \ge (1 + \delta/b)^t >  n  $, a contradiction. 
\end{proof}

\begin{claim}
After the algorithm terminates, for any two clusters $C_1 \neq C_2$ we have $\dist_G(C_1, C_2) \ge (1-\eps)^b D$. 
\end{claim}
\begin{proof}
We prove the following invariant by induction: After the $i$-th phase, if two clusters do not agree on their first $i$ bits, then their distance in $G$ is at least $(1-\eps)^i D$. The claim then follow by setting $i = b$. 
Suppose the invariant holds after the $(i-1)$-th phase (case $i=0$ is an easy special case). We show that this implies that the invariant also holds after the $i$-th phase.

To do so, we first show the following property. After each step during the $i$-th phase, if two clusters do not agree on their first $i-1$ bits, then their distance in $G$ is at least $(1-\eps)^{i-1}D$.

Note that our induction hypothesis states that the property above holds at the beginning of the $i$-th phase. Hence, it remains to show that each step $j$ during phase $i$ preserves the property.

For this, it suffices to show that a node can only switch from one cluster to another cluster during the $j$-th step if the two clusters agree on their first $i-1$ bits.

To that end, consider some node $v$ and two clusters $C_1$ and $C_2$ before the $j$-th step. Assume that $v$ is contained in cluster $C_1$ before the $j$-th step and during the $j$-th step $v$ decides to join the cluster $C_2$.
It directly follows from the algorithm description that this can only happen if $\dist_G(C_2,v) \leq (1+\eps)(1-\eps)^i D < (1-\eps)^{i-1}D$.
In particular, $\dist_G(C_1,C_2) < (1-\eps)^{i-1}D$ and therefore the two clusters $C_1$ and $C_2$ agree on their first $i-1$ bits, as desired.

It remains to consider two clusters that agree on their first $i-1$ bits but disagree on their $i$-th bit.
Exactly one of those clusters, say $C$, is active during the $i$-th phase and by \cref{cl:stop_growing} we know that $C$ decides to stop growing at some point during the phase, by deleting all the nodes that proposed to join it in that step.
In particular, each alive node outside $C$ with a distance of at most $(1- \eps)^i D$ to $C$ either gets killed or decides to join an active cluster different from $C$. Since nodes in active clusters remain in the same active cluster until the end of the $i$-th phase, the distance of $C$ to any other cluster with a different $i$-th bit is at least $(1-\eps)^iD$ at the end of the $i$-th phase. 
This finishes the proof of the induction statement.

\end{proof}

\begin{claim}
For any cluster $C$, the radius of $T_C$ is $O(D \log^3(n)/\delta)$. 
\end{claim}
\begin{proof}
The radius of $T_C$ grows by $O((1+\eps)D)$ in each step. 
Each of the $b$ phases consists of $O(b /\delta \cdot \log n)$ steps. Hence, the radius of $T_C$ can be bounded by $O(b^2\log( n )D/\delta) = O(D \log^3(n) / \delta)$. 
\end{proof}

\begin{claim}
Consider any $W \subseteq \Vrem$ and $(C,T_C) \in \fC$. 
    If $\dist_G(W, \Vrem \setminus W) > (1+\eps) D$, then we have the following. 
    Whenever we have $C \cap W \not = \emptyset$, then $C \subseteq W$ and for any $u \in T_C$ we have $\dist_G(u, W) \le (1+\eps)D/2$. 
\end{claim}
\begin{proof}
Let $W \subseteq \Vrem$ be any set such that $\dist_G(W, \Vrem \setminus W) > (1+\eps) D$. A simple induction argument shows that any cluster $C$ that started as a node $u \in W$ will satisfy $C \subseteq W$ during the course of the algorithm and an analogous statement holds if we replace $W$ by $\Vrem \setminus W$. To show the second part, let $C' \subseteq W$ be some arbitrary cluster during the course of the algorithm. It suffices to show that for any node $u \in V(F'_{C'})$ in the pruned Steiner Tree it holds that $\dist_G(u,W) \leq (1+\eps)D/2$. As $u \in V(F'_{C'})$, there exists some node $v \in \Vrem$ that is contained in the subtree of $u$ in $F_{C'}$. Each node in $F_{C'}$ has a distance of at most $(1+\eps)D$ to its root.
As the root is contained in $W$, we have $\dist_G(W,v) \leq (1+\eps)D$. In particular, $v \notin \Vrem \setminus W$ but $v \in \Vrem$ and therefore $v \in W$.
Now, let $r$ be the root of $u$ in $F_{C'}$. As $v$ is an ancestor of $u$ in $F_{C'}$, we have

\begin{align*}
    \dist_{F_{C'}}(r,v)  = \dist_{F_{C'}}(r,u) + \dist_{F_{C'}}(u,v) \leq (1+\eps)D.
\end{align*}

In particular,

\begin{align*}
    \dist_G(u,W) \leq \dist_{F_{C'}}(u,\{r,v\}) \leq (1+\eps)D/2,
\end{align*}

as desired.
\end{proof}

\begin{claim}
\label{cl:weak_diam_aggregation}
During the construction, for the current clustering $\fC$ we have that we can implement the oracle  $\oForestAgg$ on $\{T_{C}\}_{C \in \fC}$ in $O(1/\delta \cdot \log^3(n))$ \congest rounds and calls to oracle $\oForestAgg_{2D}$. 
\end{claim}
\begin{proof}
We will show how to implement an associative operation $\oplus$, where $\oplus$ is either $+$ or $\min$, on the clusters of some partial clustering $\fC$ that we have during the algorithm.  

Assume that each node $u$ is given some value $v_u$ and we wish to compute $\bigoplus_{u \in C} v_u$ for each $C \in \fC$. 
Let $F^1, F^2, \dots, F^k$ be the forests found by calls to $\oWeak_{\eps, D'}$ so far and let $F'^1_C, F'^2_C, \dots, F'^k_C$, with $F'^i_C \subseteq F^i$, be the corresponding subforests that we used to grow $C$. 
We will run the following $k$ round algorithm. 
In the $i$-th round, for $0 \le i < k$, each tree $T \subseteq F'^{k-1}_C \subseteq F^{k-i}$ computes $\bigoplus_{u \in T\cap C} v_u$. 
This can be done in one call to the oracle $\oForestAgg_{2D}$ for all clusters $C$ in parallel, as each $T \subseteq F'^{k-i}_C$ is connected and the trees are mutually disjoint.  
Then, we remove the values of all nodes in $T \cap C$, except of the root $r$ of $T$ that takes the new value $v_r^{new} = \bigoplus_{u \in T\cap C} v_u$, that is, the result of the calculation. 
A straightforward induction argument shows that after $k$ rounds, the root $r_C$ of each cluster $C$ (that is, the node that defined $C$ at the beginning of the algorithm) knows the value of $\bigoplus_{u \in C} v_u$. 
Once the value of each cluster is computed, it can be broadcasted from $r_C$ back to all nodes in $C$. 
This is done analogously to the computation step, but the forests $F_1, F_2, \dots, F_k$ are now iterated over in the ascending order. 
\end{proof}

\begin{claim}
The algorithm can be implemented in $O(1/\delta^2 \cdot \polylog n)$ \congest rounds and calls to $\oForestAgg_{2D}$. It needs $O(1/\delta \cdot \log^3(n))$ calls to the oracle $\oWeak_{\eps, D'}$ for $D' \in [D/2, D]$. 
\end{claim}
\begin{proof}
The algorithm has $O(\log n)$ phases and each consists of $O(1/\delta \cdot \log^2 n) $ steps. Let us now discuss the complexity of each step. 
In each step, the decision of each cluster to grow or not can be implemented with one call of the distance oracle $\oWeak_{\eps, D'}$ and $O(1/\delta \cdot \log^3(n))$ \congest rounds and calls to $\oForestAgg_{2D}$, since the decision of each cluster to grow or not amounts to an aggregation on it and on $F_C$. 
After the decision of each cluster, we need to do additional bookkeeping. The only nontrivial part of it is pruning each forest $F_C$ into $F'_C$. To prune $F_C$, each node in it needs to know whether there are $\Vrem$ nodes in the subtree of it. 
This is a special case of the subtree sum problem that can be solved in $\poly\log(n)$ \congest rounds and calls to $\oForestAgg_{2D}$ (see \cite[Appendix B]{RGHZL2022sssp}). 
Hence, the complexity of the algorithm is $1/\delta^2 \poly\log(n)$. 
The total number of invocations of $\oWeak_{\eps, D'}$ is bounded by $O(1/\delta \cdot \log^3(n))$. 
\end{proof}

\end{proof}

\subsection{Strong-Radius Clustering}
\label{sec:strong_radius}

We now prove \cref{lem:strongdiam}, the argument is an adaptation of an argument from \cite{chang_ghaffari2021strong_diameter}. 

\begin{proof}

  We compute the strong-radius clustering by the following algorithm. The algorithm has $k = O(\log n)$ phases and in the $i$-th phase we invoke the weak-radius clustering from \cref{lem:weak_diam} with parameters $\Vrem^i, D^i$ and $\delta^i = \delta = 1/(3k)$ to obtain an output clustering $\fC^i$.
  In each phase, we add new strong-radius clusters to an initially empty set $\fC_{out}$. 
    
    At the beginning, we set $\Vrem^1 = \Vrem$ and define $\fC^0 = \{ \Vrem \}$. 
    In the $i$-th phase of the algorithm we do the following. 
  \begin{enumerate}
      \item We invoke \cref{lem:weak_diam} with parameters $\Vrem^i, D^i = 10D \cdot (1-\eps)^{2ib}$ and $\delta^i$. The output is a $(1-\eps)^bD^i$-separated weak-radius clustering $\fC^i$. 
      \item For each cluster $C \in \fC^i$, there exists a (unique) cluster $\parent(C) \in \fC^{i-1}$ with $C \subseteq \parent(C)$ (\cref{cl:strong_laminar}). 
      If $|C| > |\parent(C))|/2$, then we add $(V(T_{C}), T_{C})$ to $\fC_{out}$. 
      \item 
      Let $C^i_1, C^i_2, \dots, C^i_\ell$ denote the clusters added to $\fC^{out}$ during the $i$-th phase. \\
      We set $\Vrem^{i+1} = \left( \bigcup_{C \in \fC^i} C \right)  \setminus (\bigcup_{j=1}^{\ell} \parent(C^i_j))$. 
      
     \end{enumerate} 
     At the very end we output the output set $\fC_{out}$. This finishes the description of the algorithm.

\begin{claim}
\label{cl:strong_laminar}
For each cluster $C \in \fC^i$, there exists a (unique) cluster $\parent(C) \in \fC^{i-1}$ with $C \subseteq \parent(C)$ (\cref{cl:strong_laminar}). 
\end{claim}  
\begin{proof}
Consider an arbitrary cluster $C \in \fC^i$ and let $\parent(C)$ be any cluster in $\fC^{i-1}$ with $C \cap \parent(C) \neq \emptyset$. It suffices to show that $C \subseteq \parent(C)$.
Let $W = \parent(C)$. 
Note that $W \subseteq \Vrem^i$, since for each cluster $C' \in \fC^{i-1}$ either all nodes of $C'$ are in $\Vrem^i$ or no nodes are in $\Vrem^i$. 
Next, we have $\dist_G(W, \Vrem^i \setminus W) \ge (1-\eps)^b D^{i-1} > (1+\eps)D^i$, as the distance between any two clusters in $\fC^{i-1}$ is at least $(1-\eps)^bD^{i-1}$ and $D^i / D^{i-1} = (1-\eps)^{2b}$. 
Hence, the fourth property of \cref{lem:weak_diam} gives $C \subseteq \parent(C)$, as desired.
\end{proof}
  
\begin{claim}
\label{cl:strong_finish}
$\Vrem^k = \emptyset$. 
\end{claim}  
\begin{proof}
Let $C^0, C^1,\ldots, C^i$ be clusters such that $C^i \in \fC^i$ and $\parent(C^{j}) = C{j -1}$ for $1 \le j \le i$. We have $|C^0| \leq n$ and $|C^{j}| \leq |C^{j -1 }|/2$. Hence, $|C^i| \leq n/2^{i-1}$.
Let $k = 3 + \lceil \log_2 n \rceil$. Then, $\fC^{k-1}$ cannot contain any clusters and therefore $\Vrem^k = \emptyset$. 
\end{proof}

\begin{claim}
\label{cl:strong_deleted}
$|V(\fC_{out})| \ge |\Vrem| / 3$. 
\end{claim}  
\begin{proof}
    \cref{cl:strong_finish} together with the algorithm description implies that for every node $u \in \Vrem$ at least one of the following holds.

    \begin{enumerate}
        \item $u \in  V(\fC_{out})$
        \item $u \in C$ and $C$ has a sibling $C'$ ($\parent(C) = \parent(C')$, $C \neq C'$) with $V(T_{C'}) \in \fC_{out}$ 
        \item $u \in \Vrem^i$ but the algorithm of \cref{lem:weak_diam} invoked in the $i$-th phase left $u$ unclustered 
    \end{enumerate}    
    
    There are at most $k \cdot \delta |\Vrem| = (1/3)|\Vrem|$ nodes that satisfy the third condition. 
    Moreover, as every cluster $C$ with $V(T_C) \in \fC^{out}$ satisfies $|C| \geq |\parent(C)|/2$, the total number of nodes that satisfy the first condition is at least as large as the total number of nodes that satisfy the second condition. Hence, $ |V(\fC_{out})| \geq (1/2)(2/3)|\Vrem|$, as desired.
\end{proof}

\begin{claim}
\label{cl:strong_padded}
For any two distinct clusters $C^{out}_1,C^{out}_2 \in \fC_{out}$, we have $\dist_G(C^{out}_1, C^{out}_2) \ge \eps D$. 
\end{claim}  
\begin{proof}
As $C^{out}_1,C^{out}_2 \in \fC_{out}$, there exists $i_1, i_2 \in [k]$ and two clusters $(C_1,T_{C_1}) \in \fC^{i_1}$ and $(C_2,T_{C_2}) \in \fC^{i_2}$ such that $C^{out}_1 = V(T_{C_1})$ and $C^{out}_2 = V(T_{C_2})$.
Without loss of generality $i_1 \le i_2$. 
We first compute a lower bound on the distance between $C_1 $ and $C_2$. 
Let $\hat{C}_1, \hat{C}_2 \in \fC^{i_1 - 1}$ such that $C_\ell \subseteq \hat{C}_\ell$ for $\ell \in [2]$. In particular, $\hat{C}_1 = \parent(C_1)$. 
We have $\hat{C}_1 \neq \hat{C}_2$. If $i_1 = i_2$, then $\hat{C}_1 \neq \hat{C}_2$ follows from the fact that $|C_\ell| > (1/2)| \hat{C}_\ell|$ for $\ell \in [2]$. If $i_1 < i_2$, then $\hat{C}_1 \neq \hat{C}_2$ follows from the fact that $\Vrem^{i_1 + 1} \cap \hat{C}_1 = \emptyset$ and $\Vrem^{i_2} \subseteq \Vrem^{i_1 + 1}$.
Hence,  

\begin{align*}
\dist_G(C_1,C_2)\geq \dist_G(\hat{C}_1, \hat{C}_2) \ge (1-\eps)^b D^{i_1 - 1}.    
\end{align*}

Moreover, with $W = \hat{C}_1$ we have $\dist_G(W,\Vrem^{i_1} \setminus W) \geq (1-\eps)^b D^{i_1 - 1} \geq (1+\eps)D^{i_1}$. This implies, as $C_1 \cap W \neq \emptyset$, that $\dist_G(u,W) \le (1+\eps)D^{i_1}/2$ for every $u \in V(T_{C_1}) = C^{out}_1$.
and, similarly, $\dist(T_{C^{i_2}}, \hat{C}^2) \le (1+\eps)D^{i_2}/2 \le (1+\eps)D^{i_1}/2$.
A similar argument shows that $\dist_G(u, \hat{C}_2) \leq (1+\eps)D^{i_2}/2$ for every $u \in V(T_{C_2})= C^{out}_2$.

Now, consider an arbitrary $v_1 \in C^{out}_1$ and $v_2 \in C^{out}_2$. We have 
\begin{align*}
\dist_G(v_1,v_2) &\geq 
\dist_G(\hat{C}_1, \hat{C}_2) - \dist_G(\hat{C}_1,v_1) - \dist_G(\hat{C}_2, v_2) \\
&\ge (1-\eps)^b D^{i_1 - 1} - 2(1+\eps)D^{i_1}/2\\
&= \left( (1-\eps)^b / (1-\eps)^{2b} - (1+\eps)\right) D^{i_1}
\ge 2\eps D^k \ge \eps D, 
\end{align*}
as needed. 
\end{proof}

To conclude the proof, note that each cluster $C \in \fC_{out}$ has strong-radius $O(D \log^3(n) \cdot 1/\delta) = O(D \log^4(n))$, clearly $V(\fC_{out}) \subseteq \Vrem$, the clustering is $\eps D$-separated by \cref{cl:strong_padded} and $|V(\fC_{out})| \ge |\Vrem|/3$ by \cref{cl:strong_deleted}. In each phase we call \cref{lem:weak_diam} once. 
In addition, each cluster needs to check whether it gets added to the output clustering by comparing its size to the size of its parent. If a cluster indeed gets added, it also needs to inform all the nodes in the parent cluster. 
This can be implemented in $O(1)$ calls to the oracle $\oForestAgg$ on $\{T_{C}\}_{C \in \fC^i}$ that acts on clusters that \cref{lem:weak_diam} outputted. By the aggregation property in \cref{lem:weak_diam} the oracle can be implemented in $\poly\log(n)$ \congest rounds and $\poly\log(n)$ calls to distance oracles $\oWeak_{\eps, D'}$ for various $D' \in [D, 10D]$ and the oracle and calls to the oracle $\oForestAgg_{2D}$, as desired. 
\end{proof}

\subsection{Sparse Neighborhood Cover}
\label{sec:sparse_cover}

We finish by proving \cref{thm:sparse_cover}. 

\begin{proof}

Let $\Delta = \Theta(D / \log^4 n)$. 
Our algorithm consists of $t = 100 \lceil \log(n) \rceil$ rounds. Let $\Vrem^1 := V(G)$.
In the $i$-th round, we use the algorithm of \cref{lem:strongdiam} with $\Vrem^\Lref{lem:strongdiam} = \Vrem^i$ and ${D}^\Lref{lem:strongdiam} = \Delta$ to compute a $s = \Omega(\Delta/\log^2 n) = \Omega(D / \log^6 n)$-separated $\Theta(D)$-strong-radius clustering $\fC^i$ with $V(\fC^i) \subseteq \Vrem^i$ and $|V(\fC^i)| \geq |\Vrem^i|/3$. Moreover, we can choose the constant in the definition of $\Delta$ such that $\fC_i$ is $D/4$-strong-radius clustering. 
We then set $\Vrem^{i+1} = \Vrem^i \setminus V(\fC^i)$.

We now extend each cluster $(C,T_C) \in \fC^i$ to some new cluster $(\hat{C},T_{\hat{C}})$. This defines a new clustering $\hat{\fC}^i$. 
To do so, we first invoke the distance oracle $\oWeak_{\eps,s/10}$ with input $V(\fC^i)$.
Let $F$ denote the rooted forest returned by the oracle.
For each $(C, T_C) \in \fC^i$, let $F_C \subseteq F$ be the forest which contains all trees in $F$ whose root is contained in $C$. We set $\hat{C} = C \cup V(F_C)$ and  define $T_{\hat{C}}$ as the rooted tree one obtains by adding the trees in $F_C$ to $T_C$, i.e, $V(T_{\hat{C}}) = V(F_C) \cup V(T_C)$ and $T_{\hat{C}}$ and $T_C$ have the same root. 
Clearly, any node in $T_{\hat{C}}$ has a distance of at most $D/4 + (1+\eps)(s/10) = D/2$ to the root. 
Moreover, as the clustering $\fC^i$ is $s$-separated, it follows that $\hat{\fC^i}$ is $\hat{s}$-separated with $\hat{s} := s/10 \leq s - 2(1+\eps)(s/10)$ and for any node $u$ with $\dist_G(u,C) \leq \hat{s}$ we have $u \in \hat{C}$.  
In particular, as $\Vrem^t = \emptyset$, for every node $u \in V(G)$ there exists some $i$ such that the ball of radius $\hat{s}$ around $u$ is fully contained in one of the clusters of $\hat{C}^i$.
Hence, $(\hat{\fC}^i)_{i \in [t]}$ is a sparse neighborhood cover with covering radius $\hat{s} \ge D/\log^7 n$ and each cluster $C$ of it comes with a rooted tree $T_C$ of diameter at most $2 \cdot (D/2) = D$. 

It directly follows from the guarantees of \cref{lem:strongdiam} that the computation consists of $\poly(\log n)$ \congest rounds in $G$ and $\poly(\log n)$ calls to distance oracles $\oWeak_{\eps, D'}$ for $\eps := \frac{1}{100b \log(n)} \ge 1/\log^3 n$ and various $D' \in [D/\log^7 n, D]$ and $\poly\log n$ calls to the oracle $\oForestAgg_{2D}$, as needed.
\end{proof}


\end{document}